\documentclass{article}
\usepackage[utf8]{inputenc}
\usepackage[a4paper,margin=1.3in]{geometry}

\usepackage{amsmath,dsfont}
\usepackage{amsthm,subcaption,graphicx}
\usepackage{amsfonts,amssymb}
\usepackage[font=small,labelfont=bf]{caption}
\usepackage{authblk}

\usepackage{natbib}
\usepackage{subcaption}
\usepackage[font={footnotesize,it}]{caption}
\usepackage{xcolor}
\usepackage{enumitem}
\usepackage{blindtext}
\usepackage{url}
\usepackage{hyperref}

\newtheorem{thm}{Theorem}
\newtheorem{lemma}[thm]{Lemma}
\newtheorem{defn}[thm]{Definition}
\newtheorem*{remark}{Remark}

\definecolor{oranje}{cmyk}{0,0.50,0.84,0}

\newcommand{\Cov}{\mathrm{Cov}}
\newcommand{\Var}{\mathrm{Var}}
\newcommand{\argmin}[1]{\underset{#1}{\mathrm{argmin}}\ }
\newcommand{\argmax}[1]{\underset{#1}{\mathrm{argmax}}\ }
\newcommand{\bs}[1]{\boldsymbol{#1}}
\newcommand{\ee}{\mathrm{exp}}
\newcommand{\dd}{\mathrm{d}}
\newcommand{\mc}[1]{\mathcal{#1}}

\usepackage{algorithmic,algorithm}

\title{Flexible co-data learning for high-dimensional prediction}
\author[1]{Mirrelijn M. van Nee\thanks{ 
The first author is supported by ZonMw TOP grant COMPUTE CANCER (40-
00812-98-16012).}}
\author[2]{Lodewyk F.A. Wessels}
\author[1]{Mark A. van de Wiel}
\affil[1]{{\small{Epidemiology and Biostatistics, Amsterdam University Medical Centers}}}
\affil[2]{{\small{Molecular Carcinogenesis, Oncode Institute and Netherlands Cancer Institute}}}
\date{}

\begin{document}

\maketitle

\begin{abstract}
Clinical research often focuses on complex traits in which many variables play a role in mechanisms driving, or curing, diseases. 
Clinical prediction is hard when data is high-dimensional, but additional information, like domain knowledge and previously published studies, may be helpful to improve predictions. 
Such complementary data, or co-data, provide information on the covariates, such as genomic location or p-values from external studies.
Our method enables exploiting multiple and various co-data sources to improve predictions. 
We use discrete or continuous co-data to define possibly overlapping or hierarchically structured groups of covariates.
These are then used to estimate adaptive multi-group ridge penalties for generalised
linear and Cox models. 
We combine empirical Bayes estimation of group penalty hyperparameters with an extra
level of shrinkage. 
This renders a uniquely flexible framework as any type of shrinkage can be used on the group level. 
The hyperparameter shrinkage learns how relevant a specific co-data source is, counters overfitting of hyperparameters for many groups, and accounts for structured co-data. 
We describe various types of co-data and propose suitable forms of hypershrinkage.
The method is very versatile, as it allows for integration and weighting of multiple co-data sets, inclusion of unpenalised covariates and posterior variable selection. 
We demonstrate it on two cancer genomics applications and show that it may improve the performance of other dense and parsimonious prognostic models substantially, and stabilises variable selection.
\end{abstract}


\section{Introduction}


High-dimensional data is increasingly common in clinical research in the form of omics data, e.g. data on gene expressions, methylation levels and copy number alterations.
Omics are used in clinical applications to predict various outcomes, in particular binary and survival, possibly using clinical covariates like age and gender in addition to omics in the predictor. 
Examples in cancer genomics include predicting diagnosis of cancer, therapy response and time to recurrence of a tumour.


Unfortunately, many clinical omics studies are hampered by small sample size (e.g. $n = 100$), either due to budget or practical constraints. 
In addition to the main data, however, auxiliary information on the covariates usually exists, in the form of domain knowledge and/or results from external studies. 
In cancer genomics, acquired domain knowledge is made available in published disease related signatures and online encyclopediae such as the Gene Ontology \citep{ashburner2000gene} and Kegg pathways \citep{kanehisa2000kegg}. 
External, similar studies are available in repositories like The Cancer Genome Atlas \cite[TCGA]{tomczak2015cancer}, from which summary statistics like p-values or false discovery rates can be derived.
In general we use the term \textit{co-data}, for \textit{\underline{co}mplementary \underline{data}}, to refer to any data that complements the main data by providing information on the covariates.
Figure \ref{fig:codata} depicts some examples.

We contrast co-data learning with meta-analysis, as both use multiple data sources. 
For the latter, the focus lies on \textit{estimation} of model parameters over a \textit{common} research population. Hence, the research question and population should be similar for all data sources.
For learning from co-data on the other hand, the focus lies on \textit{prediction} for new samples from the \textit{main} research population.
External research may therefore differ from the main data in outcome, e.g. different types of disease, or in research population, e.g. animals versus cell lines. 
While these external data cannot be used in a meta-analysis, or simply be concatenated to the main data, it may still contain valuable information for predictions and be summarised for use as co-data.

One would like to build upon all relevant existing knowledge when learning predictors and selecting covariates for the main data, thus learn from multiple and various co-data sets.
Co-data vary in relevance and type of data.
How much, if anything at all, can be learnt from co-data depends on the application and data at hand, and is in general unknown. 
The type of co-data may be continuous or discrete, e.g. external p-values or group membership, possibly further constrained or structured, e.g. hierarchical groups.



Various methods have been developed which focus on predicting a specific type of response combined with one source of co-data, see for instance \citep{boonstra2013incorporating,tai2007incorporating,treppmann2017integration}. 
Extending these methods to different types of response or co-data is not always straightforward, as approximation or optimisation algorithms often do not generalise trivially.
Typically, co-data or, more general, prior information is included in statistical prediction models by letting it guide the choice for a specific penalty (or prior) that penalises (or shrinks) model parameters. 
As this choice highly affects the model fit in high-dimensional data, the ability of the fitted prediction model to generalise well to new samples heavily relies on a carefully tuned penalty or prior.
Penalties for group lasso \citep{yuan2006model} and for latent overlapping group lasso \citep{jacob2009group} penalise covariates in groups to favor \textit{group sparse} solutions, selecting groups of covariates. 
While being able to use these group penalties to incorporate additional structure on the group level such as grouped trees \citep{liu2010moreau} and hierarchical groups \citep{yan2017hierarchical}, only one overall hyperparameter is used to tune the penalty. 
This makes the penalty unable to adapt locally to the main data when part of the groups or structure is non-informative for, or in disagreement with the main data, leading to sub-optimally performing prediction models.

Recent work has focused on \textit{group adaptive} penalties, see \citep{wiel2016better,munch2018adaptive,velten2018adaptive}, in which groups of covariates share the same prior or penalty parameterised by a group-specific hyperparameter.
The hyperparameters are learnt from the data, effectively learning how informative the co-data is and how important each covariate group is for the prediction problem at hand.
Whereas these penalties or priors are able to adapt locally on the group level, these methods do not allow for including any structure on the groups. Moreover, the methods tend to overfit in the number of hyperparameters for an increasing number of groups.



Here we present a method for ridge penalised generalised linear models that is the first to combine adaptivity on the group level with the ability to handle multiple and various types of co-data.
While the main data still drives the regression parameter estimation, the co-data can impact the penalties which act as inverse weights in the regression.
By adequately learning penalties from valuable co-data, prediction and covariate selection for omics improve.
The method is termed \texttt{ecpc}, for Empirical bayes Co-data learnt Prediction and Covariate selection. 
A moment-based empirical Bayes approach is used to estimate the adaptive group ridge penalties efficiently, opening up the possibility to introduce an extra layer of shrinkage on the group level.
Any type of shrinkage can be used in this layer, rendering a unique, flexible framework to improve predictions because:
\begin{enumerate}[noitemsep,nolistsep]
\item much as a penalty on the covariate level shrinks regression coefficients towards $0$ to counter overfitting and improve parameter estimates, a penalty on the group level shrinks adaptive group penalties to an ordinary, non-adaptive ridge penalty. 
Therefore, the method is able to learn how informative co-data is, ranging from no shrinkage for informative, stable co-data, to full shrinkage for non-informative co-data;

\item instead of including group structure on the covariate level, a structured penalty is included on the group level directly. 
The method utilises this facet to incorporate known structure of overlapping groups, to handle hierarchically structured groups and to handle continuous data by using a data-driven adaptive discretisation.
\end{enumerate}
Multiple co-data are handled by first combining each co-data set with a penalty suitable for that specific co-data source, then integrating various co-data by learning co-data weights with the same moment-based empirical Bayes approach.
Lastly, the framework allows for unpenalised covariates and posterior variable selection. 
Our approach to use a dense model (ridge regression) plus posterior selection is motivated by a three-fold argument: i) biology: for complex traits such as cancer most of the genome is likely to have an effect \citep{boyle2017expanded}; ii) statistics: even in sparse settings dense modelling plus posterior selection can be rather competitive to sparse modelling \citep{bondell2012consistent}, while better facilitating to shift on the grey-scale from sparse to dense; iii) data:  the use of co-data aligns well with dense modelling, allowing it to have a smooth impact on penalties and parameter estimates.


The paper is outlined as follows. 
Section \ref{par:codata} elaborates on generic types of co-data. 
Section \ref{par:method} then presents the model and methods to estimate the model parameters.
Here, we present the penalised estimator for adaptive group penalties using an extra layer of any type of shrinkage, which forms the basis for handling various types of co-data. 
Several model extensions are presented in Section \ref{par:extensions}. 
Section \ref{par:simulation} presents a simulation study illustrating how the extra layer of shrinkage enables the method to learn to shrink group weights when needed. 
Section \ref{par:application} then demonstrates the method on two applications in cancer genomics using multiple co-data, showing that \texttt{ecpc} improves or matches benchmark methods that are either group-adaptive or incorporate additional group structure, but from which none are able to incorporate both. 
Finally, Section \ref{par:discussion} concludes and discusses the method.

\section{Co-data}\label{par:codata}
Co-data complements the main data from which the predictor has to be learnt. Whereas the main data contain information about the \textit{samples}, the co-data contain information about the \textit{covariates}. 
Co-data can be retrieved from external sources, e.g. from public repositories, or derived from the main data, as long as the response is not used. 
To exemplify different types of co-data, we show some prototypical examples in Figure \ref{fig:codata}. Here we describe the generic structures of co-data underlying the examples.

\textbf{Non-overlapping groups of covariates:} the covariates are grouped in non-overlapping groups. An example in cancer genomics is groups of genes located on the same chromosome.

\textbf{Overlapping groups:} the covariate groups are overlapping, for example, groups representing pathways, i.e. for some biological process all genes involved are grouped. As genes often play a role in multiple processes, the resulting groups are overlapping.


\textbf{Structured groups:} relations between groups are represented in a graph. Gene ontology, for example, represents groups of genes in a directed acyclic graph. Each node in the graph represents a biological function corresponding to a group of genes that (partly) fulfill that function. Nodes at the top of the hierarchy represent general functions and are refined in more specific biological functions downwards in the graph. Each node represents a subset of genes of its parent nodes.

\textbf{Continuous co-data:} as opposed to the discrete groups in the previous examples, the co-data are continuous. An example in cancer genomics is p-values derived from a previously published, similar study. Another example is standard deviations of each covariate computed from the data without using the response.

\begin{figure}
    \centering
    \begin{subfigure}[b]{0.24\linewidth}
        \centering
        \includegraphics[width=0.8\linewidth]{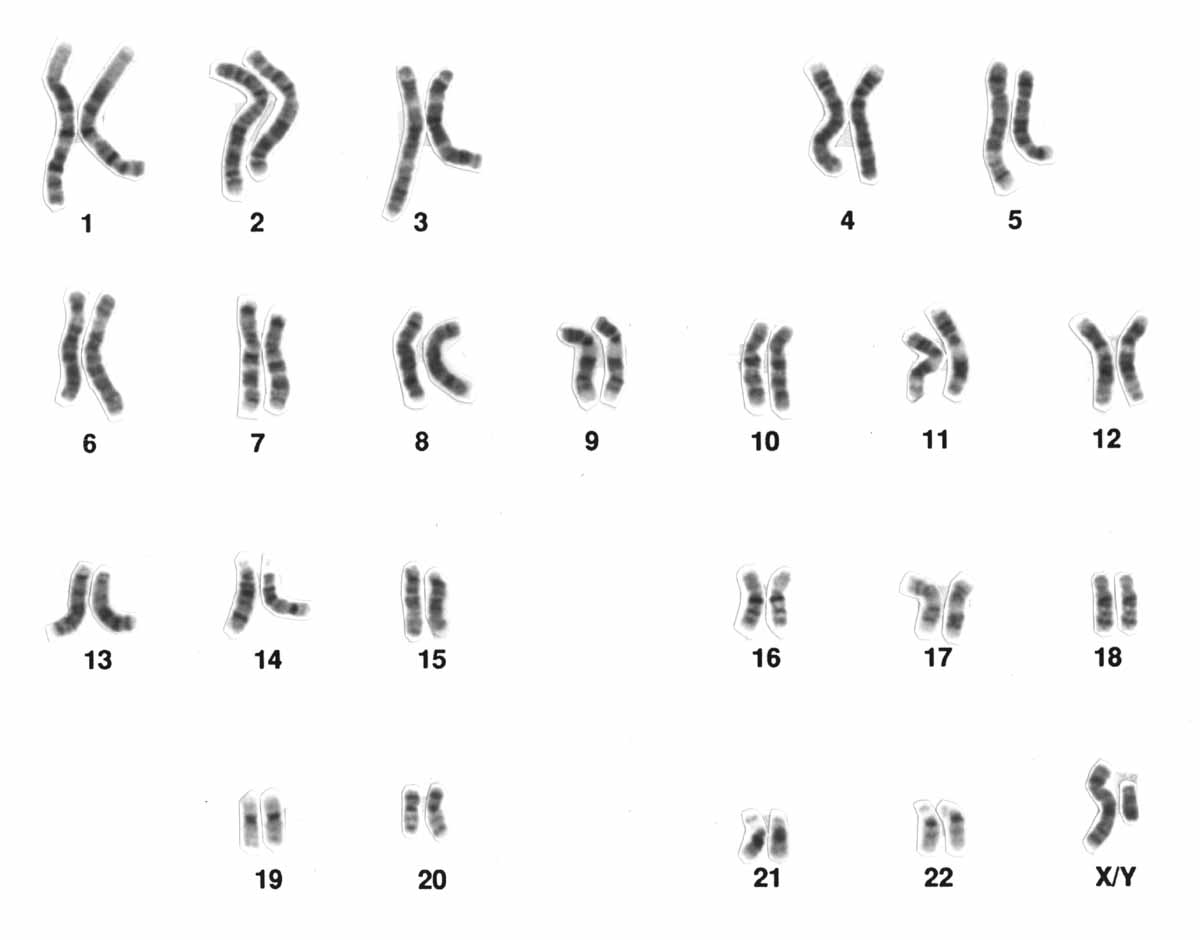}
        \caption{}
    \end{subfigure}
    \begin{subfigure}[b]{0.24\linewidth}
        \centering
        \includegraphics[width=\linewidth]{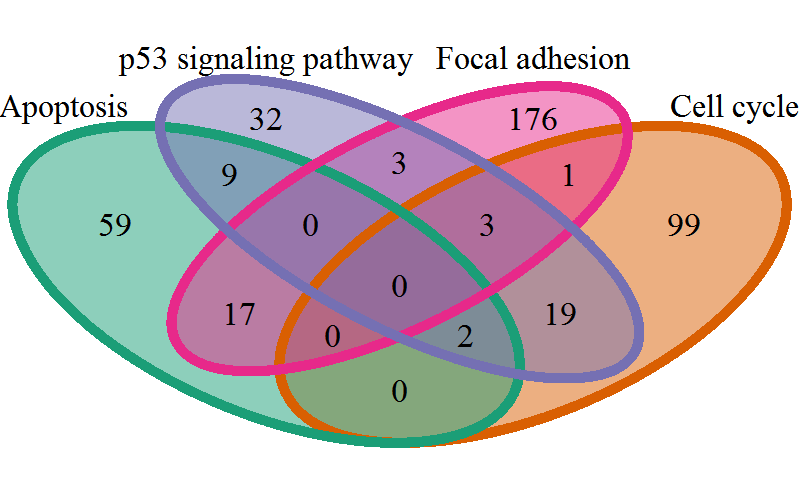}
        \caption{}
    \end{subfigure}
    \begin{subfigure}[b]{0.24\linewidth}
        \centering
        \includegraphics[width=\linewidth]{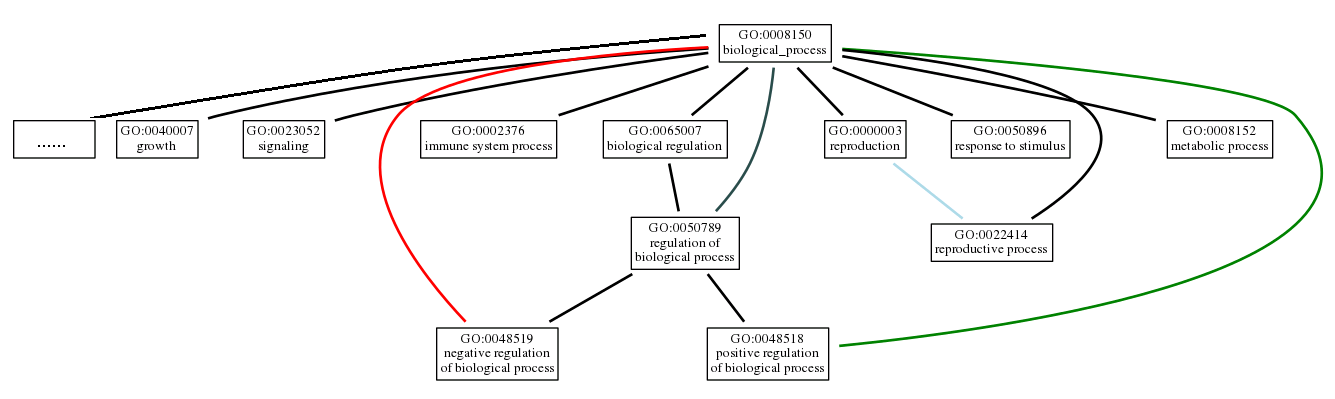}
        \caption{}
    \end{subfigure}
    \begin{subfigure}[b]{0.24\linewidth}
        \centering
        \includegraphics[width=\linewidth]{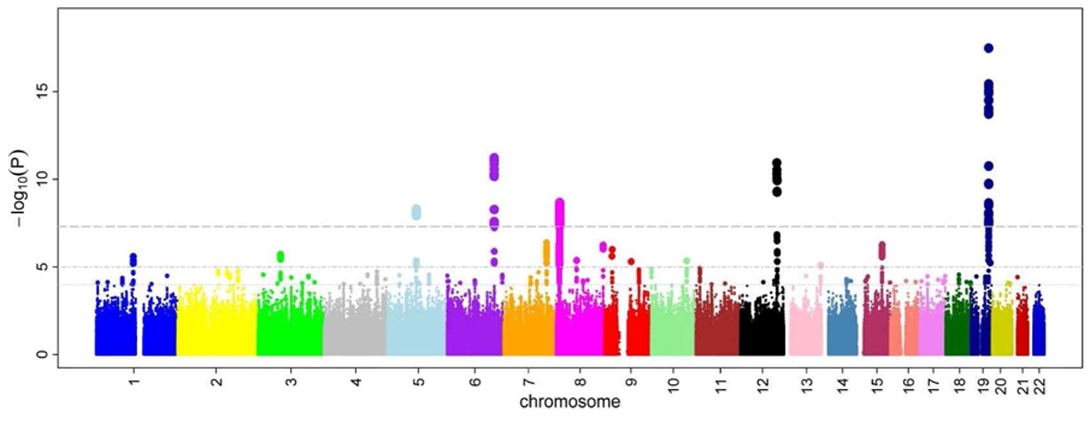}
        \caption{}
    \end{subfigure}
    \caption{Examples of different types of co-data in cancer genomics. (a) Chromosomes: non-overlapping groups of genes on the same chromosome. (b) Pathways: overlapping groups of interacting genes or molecules. (c) Gene ontology: groups structured in a directed acyclic graph (DAG) representing relationships in for example biological function. (d) p-values: continuous p-values derived from an external study.}
    \label{fig:codata}
\end{figure}

\section{Method}\label{par:method}
\subsection{Notation}
Let us first give some notation and some definitions to describe the data and co-data.
Let $\bs{Y}\in\mathbb{R}^{n}$ denote the response vector, $X\in\mathbb{R}^{n\times p}$ denote the observed high-dimensional data matrix, $p\gg n$, and let $Z^{(d)}\in\mathbb{R}^{p\times G^{(d)}},\ d=1,..,D$, defined below, denote $D$ different \textit{co-data matrices} representing \textit{groupings} of covariates:
\begin{defn}
Define each \textbf{grouping} $\mc{G}^{(d)}$, $d=1,..,D$, as a collection of sets $\mc{G}^{(d)}_g$ (called \textbf{groups}) of covariate indices in $\{1,..,p\}$, such that each covariate belongs to at least one group:
\begin{align}
    \{1,..,p\} = \bigcup_{\mc{G}^{(d)}_g\in\mc{G}^{(d)}}\mc{G}^{(d)}_g,\ \forall d=1,..,D.
\end{align}
Denote the grouping size, i.e. number of groups in each grouping, by $G^{(d)}:=|\mc{G}^{(d)}|$, and denote the group size of group $g$ in grouping $d$, i.e. the number of covariates in that group, by $G^{(d)}_g:=|\mc{G}^{(d)}_g|$.
\end{defn}
Note that we use superscripts for the \textit{groupings} number and subscripts for the \textit{group} number in that grouping. The notation is illustrated in Figure \ref{fig:groupings}.
Covariates with missing co-data should preferably be grouped in a separate group as the missingness might be informative.
The groups can possibly be overlapping or structured as in a hierarchical tree, illustrated in Figure \ref{fig:schemeEstPrior}. 
Each co-data matrix $Z^{(d)}$ is defined by a grouping as follows, and illustrated in Figure \ref{fig:groupings}.
\begin{defn}
For each grouping $\mc{G}^{(d)}$, $d=1,..,D$, we define the corresponding \textbf{co-data matrix} $Z^{(d)}$ as the matrix with matrix element $[Z^{(d)}]_{kg}$ on the $k^{th}$ row and $g^{th}$ column given by:
\begin{align}
\begin{split}
    [Z^{(d)}]_{kg}&=\left\{\begin{array}{ll}
        \frac{1}{|\mc{I}^{(d)}_k|} & \text{if }g\in \mc{I}^{(d)}_k \\
        0 & \text{if not}
    \end{array}\right.,\  d=1,..,D,\ k=1,..,p,\ g=1,..,G^{(d)},\\
\end{split}
\end{align}
 where $|\mc{I}^{(d)}_k|$ is the number of groups of grouping $d$ covariate $k$ is in, and $\mc{I}^{(d)}_k$ is the set of indices of the groups to which $\beta_k$ belongs in grouping $d$,
$ \mc{I}^{(d)}_k:=\{g\in\{1,..,G^{(d)}\}: k\in\mc{G}^{(d)}_g\}$.
Effectively, $\bs{Z}^{(d)}_k$ will be used to pool the information from the groups in grouping $d$ that $k$ belongs to.
\end{defn}

\begin{figure}
    \centering
    \includegraphics[width=0.6\linewidth]{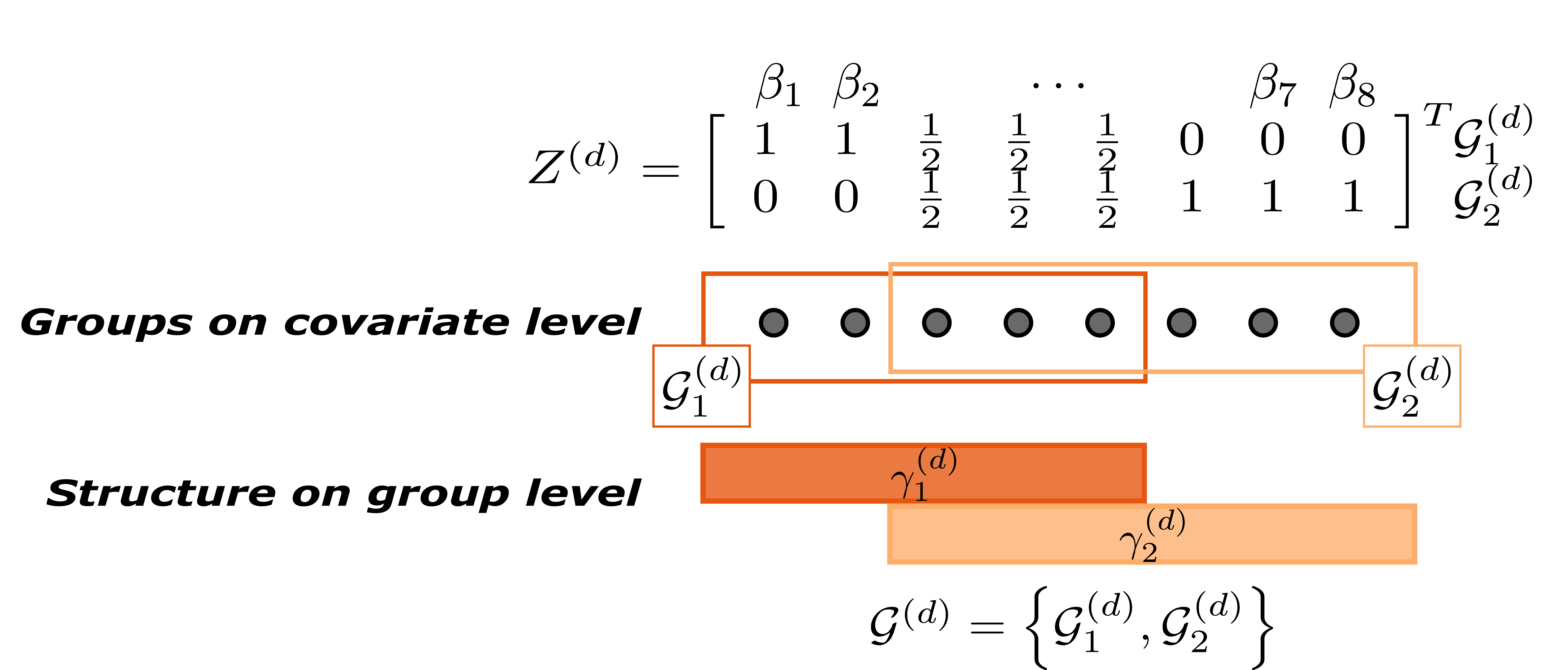}
    \caption{Illustration of notations and definitions. Grey balls represent covariates, colored rectangles groups of covariates. Grouping $\mc{G}^{(d)}$ consists of $G^{(d)}=2$ overlapping groups, $\mc{G}^{(d)}_1$ and $\mc{G}^{(d)}_2$, of sizes $G^{(d)}_1=5$ and $G^{(d)}_2=6$. The grouping defines the co-data matrix $Z^{(d)}$. Each group $\mc{G}^{(d)}_i$ corresponds to a weight $\gamma_i^{(d)}$ on the group level.}
    \label{fig:groupings}
\end{figure}

\subsection{Model}\label{par:model}
We regress $\bs{Y}$ on $X$ using a generalised linear model (GLM) with regression coefficient vector $\bs{\beta}\in\mathbb{R}^p$. We impose a normal prior on $\bs{\beta}$ with a global prior variance $\tau^2_{global}$ and local prior variance $\tau^2_{k,local}$. The local prior variances are regressed on the co-data $Z^{(d)}$, $d=1,..,D$, with each of the $D$ group weight vectors $\bs{\gamma}^{(d)}\in\mathbb{R}^{G^{(d)}}_+$ modeling the relative importance of the groups in grouping $d$, and the grouping weight $w^{(d)}\in\mathbb{R}_+$ the relative importance of grouping $d$. The model is then as follows:
\begin{align}
\begin{split}
    Y_i | \bs{X}_i,\beta &\overset{ind.}{\sim} \pi\left(Y_i | \bs{X}_i,\beta\right),\ E_{Y_i|\bs{X}_i,\bs{\beta}}(Y_i)=g^{-1}(\bs{X}_i\bs{\beta}),\ i=1,..,n,\\
    \beta_k&\overset{ind.}{\sim}N(0,\tau_{global}^2\tau^2_{k,local}),\ k=1,..,p,\\
    \tau^2_{k,local}&=\sum_{d=1}^Dw^{(d)}\bs{Z}^{(d)}_k\bs{\gamma}^{(d)},\ k=1,..,p,
    \end{split}\label{eq:model}
\end{align}
with $\pi\left(Y_i | \bs{X}_i,\beta\right)$ some exponential family distribution with corresponding link function $g(\cdot)$, $\bs{X}_i$ denoting the $i$th row of $X$, and $E_{Y_i|\bs{\beta}}$ denoting the expectation with respect to the probability density/mass function $\pi(Y_i|\bs{X}_i,\bs{\beta})$, where we leave out dependence on $\bs{X}_i$ since we consider $X$ as fixed. 
Note that when some groups are overlapping and say $\beta_k$ belongs to $|\mc{I}^{(d)}_k|$ different groups, we average the group weights. 
Large group weights $\gamma_g^{(d)}$ correspond to large prior variances.

We adopt the Bayesian formulation in Equation \eqref{eq:model} to estimate the prior parameters with an empirical Bayes approach explained in Section \ref{par:estimation}. For the final predictor however, we make use of the equivalence between the maximum a posteriori estimate for $\bs{\beta}$, $\hat{\bs{\beta}}$, and the penalised maximum likelihood estimate (MLE), and predict the response $Y_{new}$ for new samples $X_{new}$ in a frequentist manner. That is, we predict new samples by $\hat{Y}_{new}=g^{-1}(\bs{X}_{new}\hat{\bs{\beta}})$.

The prior is similar to the prior used in the method \texttt{GRridge} proposed in \citep{wiel2016better}, but has additional grouping weights, such that multiple groupings (called \textit{partitions} in \citep{wiel2016better}) can be evaluated simultaneously instead of iteratively. Moreover, whereas \texttt{GRridge} tends to overfit for many co-data groups, we introduce an extra level of shrinkage on the prior parameter level to counter this. This extra level has a substantial practical impact as it opens up the possibility of using the wealth of existing shrinkage literature to handle various types of co-data to improve predictions, as explained below in Section \ref{par:extrashrinkage}.

\subsection{Estimation}\label{par:estimation}
The unknown model parameters are the regression coefficients $\bs{\beta}$ and the prior parameters, also called \textit{hyperparameters}, $\left\{\tau^2_{global},\bs{\gamma}^{(1)},..,\bs{\gamma}^{(D)},w^{(1)},..,w^{(D)}\right\}$, where the local variances $\tau^2_{k,local}$ are omitted as those relate directly to $\bs{\gamma}^{(d)}$ and $\bf{w}$ via Equation \eqref{eq:model}. We use an empirical Bayes approach \citep{van2019learning}: estimate the hyperparameters and plug those in the prior to find the penalised maximum likelihood estimate for $\bs{\beta}$:
\begin{align}\label{eq:estbeta}
\begin{split}
    \hat{\bs{\beta}} &=\argmax{\beta} \left\{\log\pi(\bs{Y}|X,\bs{\beta}) - \frac{1}{\hat{\tau}^2_{global}}\sum_{k=1}^p \frac{1}{\hat{\tau}^2_{k,local}}\beta_k^2\right\}.
\end{split}
\end{align}
Note that this is just ordinary ridge regression with a weighted penalty, which can easily be solved with existing software, e.g. with the \texttt{R}-package \texttt{glmnet}.
Hence, the main task is to estimate the hyperparameters.
We do so in a hierarchical fashion in three steps, illustrated in Figure \ref{fig:schemeEstPrior}. These steps can be summarised as follows, details given below:
\begin{enumerate}[noitemsep,nolistsep]
    \item Overall level of regularisation $\hat{\tau}^2_{global}$: for linear regression, we maximise the marginal likelihood directly as it is analytical, setting all local variances to $1$. For other types of regression (for now, logistic and Cox), we use the canonical approach of cross-validation, which can be computed efficiently \citep{hastie2004efficient}.
    
    \item Group weights for each grouping, $\bs{\gamma}^{(d)},\ d=1,..,D$, given $\hat{\tau}^2_{global}$: we use penalised moment-based estimates based on an initial, ordinary ridge estimate $\tilde{\bs{\beta}}$ using the ridge penalty related to $\hat{\tau}^2_{global}$. The regularisation of the moment-based estimating equations accounts for structure in the groups and overfitting when the number of groups approaches or exceeds the number of samples. Various penalty functions can be used for various types of co-data. The penalty functions are parameterised by hyperpenalties $\lambda^{(d)}$, which are estimated in a data-driven way using splits of the groups. 
    
    \item Grouping weights $\bs{w}=(w^{(1)},..,w^{(D)})^T$, given $\hat{\tau}^2_{global}$ and $\hat{\bs{\gamma}}^{(1)},..,\hat{\bs{\gamma}}^{(D)}$: 
    we use moment-based estimates for the grouping weights. 
\end{enumerate}

\subsubsection{Group weights for each grouping, \texorpdfstring{$\bs{\gamma}^{(d)},\ d=1,..,D$}{gamma}}\label{par:extrashrinkage}
We use the empirical Bayes method of moments (MoM) to estimate the group weights for each grouping separately \citep{van2019learning}. 
\texttt{GRridge} \citep{wiel2016better} implements the moment-based estimates for the prior variance for linear and logistic regression. Here, to present a coherent framework we first repeat the main steps. After, we explain the new, extra level of shrinkage, used to obtain stable local variance estimates. Below, we sometimes refer to the extra level of shrinkage as \textit{hypershrinkage}, to clearly distinguish shrinking regression coefficients on the covariate level from shrinking hyperparameters on the group level. We provide details for the MoM estimating equations for linear, logistic and Cox regression in Section \ref{ap:mom} in the Supplementary Material. As a last note, throughout this paper we assume a zero prior mean, as given in Equation \eqref{eq:model}. The MoM can easily be extended to include estimates for a prior mean $\mu_k$, $k=1,..,p$, in case $\beta_k$ should be shrunk to a non-zero target $\mu_k$. Details are given in Section \ref{ap:mom} in the Supplementary Material. 

Let the estimate $\hat{\tau}^2_{global}$ be given, estimated as explained above. The ordinary ridge MLE corresponding to this level of regularisation, $\tilde{\bs{\beta}}(\bs{Y},\hat{\tau}^2_{global})$, is a function of the data $Y$. 
Consider one grouping $\mc{G}^{(d)}$, $d\in\{1,..,D\}$. The MoM equates empirical moments to theoretical moments over all covariates $\beta_k$ in one group $\mc{G}^{(d)}_g\in\mc{G}^{(d)}$, where the theoretical moments are taken with respect to the marginal likelihood $\pi(\bs{Y}|\bs{\gamma}^{(d)},\hat{\tau}^2_{global})$. Setting up the moment equation for all $G^{(d)}$ groups in the grouping $\mc{G}^{(d)}$, we obtain the following equations:
\begin{align}
    \forall g=1,..,G^{(d)}:&\ \frac{1}{|\mathcal{G}^{(d)}_g|}\sum_{k\in\mathcal{G}^{(d)}_g} \tilde{\beta}_k^2 = \frac{1}{|\mathcal{G}^{(d)}_g|}\sum_{k\in\mathcal{G}^{(d)}_g} E_{\bs{Y}|\bs{\gamma}^{(d)},\hat{\tau}^2_{global}}\left[\tilde{\beta}^2_k(\bs{Y},\hat{\tau}^2_{global})\right] \label{eq:MoM1}\\
    &\qquad=\frac{1}{|\mathcal{G}^{(d)}_g|}\sum_{k\in\mathcal{G}^{(d)}_g} E_{\bs{\beta}|\bs{\gamma}^{(d)},\hat{\tau}^2_{global}}\left[E_{\bs{Y}|\bs{\beta}}\left[\tilde{\beta}^2_k(\bs{Y},\hat{\tau}^2_{global})|\bs{\beta}\right]\right]\label{eq:MoM2}\\
    &\qquad=\frac{1}{|\mathcal{G}^{(d)}_g|}\sum_{k\in\mathcal{G}^{(d)}_g} h\left(\bs{\gamma}^{(d)}\right) \label{eq:MoM3},
\end{align}
with $h(\cdot)$ a function of the unknown parameters $\bs{\gamma}^{(d)}$.

The theoretical moments on the right-side of the equation above are analytic for linear regression and are approximated by using a second order Taylor approximation for the inner expectation in Equation \eqref{eq:MoM2} for logistic (see \citep{le1992ridge}) and Cox regression, after which the outer expectation is analytic.
The function $h$ (or its approximation) in Equation \eqref{eq:MoM3} is linear in $\bs{\gamma}^{(d)}$, i.e. solving the moment estimating equations boils down to solving a linear system of $G^{(d)}$ equations and $G^{(d)}$ unknowns $\bs{\gamma}^{(d)}$:
\begin{align}\label{eq:linearsystem}
    A^{(d)}\bs{\gamma}^{(d)} &= \bs{b}^{(d)},
\end{align}
with $A^{(d)}\in\mathbb{R}^{G^{(d)}\times G^{(d)}}$ and $\bs{b}^{(d)}\in\mathbb{R}^{G^{(d)}}$ depending on the data $X$ and initial estimate $\tilde{\bs{\beta}}(\bs{Y},\hat{\tau}^2_{global})$.
Details are given in Section \ref{ap:mom} in the Supplementary Material.

In case of few, non-overlapping groups of equal size, it suffices to solve the linear system directly, truncating negative group weight estimates, potentially resulting from approximation or numerical errors, to $0$.
However, often we have {\em many} groups, potentially unequal in size, or structured in another, potentially hierarchical, way, which demands penalisation of the system to prevent overfitting, as demonstrated in Section \ref{par:simulation}. Hence we propose to replace the solution of Equation \eqref{eq:linearsystem}, which can be cast as a least squares minimisation, by $\hat{\bs{\gamma}}^{(d)}$: 
\begin{align}\label{eq:tauest}
    \hat{\bs{\gamma}}^{(d)}&= (\tilde{\bs{\gamma}}^{(d)})_+,\ \tilde{\bs{\gamma}}^{(d)}=\argmin{\bs{\gamma}^{(d)}} ||A^{(d)} \bs{\gamma}^{(d)} -\bs{b}^{(d)}||^2_2 + f^{(d)}_{pen}\left(\bs{\gamma}^{(d)};\hat{\lambda}^{(d)}\right),
\end{align}
where $(\cdot)_+=\max(0,\cdot)$ denotes the element-wise truncation of the elements of a vector at $0$, and where $\hat{\lambda}^{(d)}$, the estimate for the hyperpenalty parameter $\lambda^{(d)}$, is obtained as explained below. 
Note that solving Equation \eqref{eq:tauest} corresponds to solving a penalised linear regression with penalty function $f^{(d)}_{pen}$. So for most well-known penalties, such as ridge and lasso, software exists to obtain estimates for $\tilde{\bs{\gamma}}^{(d)}$.
Otherwise, a general purpose gradient-based solver may be used, which will usually suffice because $\bs{\gamma}^{(d)}$ is generally not a very large dimensional vector.

The modular approach of decoupling group shrinkage from direct covariate shrinkage not only relieves the computational burden for $p\rightarrow\infty$, but also accommodates generalising to any other group shrinkage scheme.
As a default hyperpenalty, we propose to use a weighted ridge penalty with target $1$ and weighted hyperpenalty parameter $\lambda^{(d)}$ governing the amount of group shrinkage.
The target of $1$ embodies the prior assumption that the grouping is not informative: all group weights are shrunk towards $1$. Then, the weighted ridge prior on the covariate level is shrunk to an ordinary ridge prior.
The hyperpenalty is weighted such that the local variances on the covariate level are a priori independent of the group sizes.
Details are given in Section \ref{ap:hypershrinkage} in the Supplementary Material.
A ridge penalty on the covariate level is used to improve regression coefficient estimates when there are many, possibly correlated covariates. In a similar sense, the ridge penalty on the group level improves the group parameter estimates when there are many groups or overlapping and therefore correlated groups.

Instead of truncating the group weight estimates $\tilde{\bs{\gamma}}$ at $0$, one could employ a penalty that has support on the positive real numbers only, such as the logarithm of the inverse gamma distribution, as it naturally models variance parameters.
Use of an inverse gamma penalty lead, however, to inferior results in our applications. 
An intuitive explanation for this is, while $\bs{\gamma}^{(d)}$ models variance parameters on the covariate level, it does not enter the least squares error criterion in a similar fashion on the group level. 

Next, we explain how we use splits in the groups to determine the hyperpenalty parameter estimates $\hat{\lambda}^{(d)}$, required for Equation \eqref{eq:tauest}.

\subsubsection{Hyperpenalties \texorpdfstring{$\lambda^{(d)},\ d=1,..,D$}{lambda}}\label{par:esthyper}
We would like to find an estimate $\hat{\lambda}^{(d)}$ such that the penalised moment-estimates $\tilde{\bs{\gamma}}^{(d)}(\hat{\lambda}^{(d)})$ are stable and follow any constraints imposed by known group structure.
Instead of using a computationally intensive approach of cross-validation (CV) on the \textit{samples}, we use random splits of the \textit{covariate groups}. This approach relates to ideas from \textit{dropout}, a technique used in deep learning where nodes, which represent functions of (groups of) covariates, are randomly dropped by some probability in each gradient descent step in the training phase to learn robust estimates of the functions \citep{gal2016dropout}, and to techniques used in \citep{wu2018optimal}, in which moment equations are perturbed to retrieve estimates invariant for those perturbations. 

The approach is as follows: split each group $\mathcal{G}_g^{(d)}$ randomly in two parts, $\mathcal{G}^{(d)}_{g,in}$ and $\mathcal{G}^{(d)}_{g,out}$. 
Use only all \textit{in}-parts in the MoM-equations in Equation \eqref{eq:MoM1} to compute a linear system as in Equation \eqref{eq:linearsystem}, with matrix $A^{(d)}_{in}$ and vector $\bs{b}^{(d)}_{in}$ depending on which covariates belong to the \textit{in}-part. 
Similarly, one retrieves a linear system for only \textit{out}-parts with corresponding matrix and vector denoted by $A^{(d)}_{out}$ and $\bs{b}^{(d)}_{out}$.
Any stable estimate $\tilde{\bs{\gamma}}^{(d)}(\lambda^{(d)})$ that adheres to the imposed group structure should fit both the linear systems corresponding to the \textit{in}-part and the \textit{out}-part well, as both parts belong to the same groups.
Therefore we use the estimate $\hat{\lambda}^{(d)}$ for which the penalised estimate $\tilde{\bs{\gamma}}^{(d)}_{in}(\lambda^{(d)})$ of the \textit{in}-part best fits the linear system of the \textit{out}-part, averaged over multiple random splits $\mc{S}$, i.e. the estimate $\hat{\lambda}^{(d)}$ minimises the following mean residual sum of squares (RSS):
\begin{align}\label{eq:lambda}
    \hat{\lambda}^{(d)} &= \argmin{\lambda^{(d)}}RSS_{\bs{\gamma}^{(d)}}(\lambda^{(d)})
    := \argmin{\lambda^{(d)}}\frac{1}{|\mathcal{S}|}\underset{\mathcal{S}}{\sum} ||A_{out}^{(d)}\tilde{\bs{\gamma}}^{(d)}_{in}(\lambda^{(d)})-\bs{b}_{out}^{(d)}||^2_2.
\end{align}
Using cross-validation on the samples would require solving the regression for $\tilde{\bs{\beta}}\in\mathbb{R}^p$ from Equation \eqref{eq:estbeta}, setting up the linear system from Equation \eqref{eq:linearsystem} and solving the penalised regression for $\bs{\gamma}^{(d)}\in\mathbb{R}^{G^{(d)}}$ from Equation \eqref{eq:tauest}, for each fold.
Using splits of the groups only requires the latter two, now not for each fold but for each split.
The computational cost associated with splitting groups is therefore far lower than that associated with cross-validating samples, as $p$ is generally of a much larger order of magnitude than $G^{(d)}$.

\subsubsection{Grouping weights \texorpdfstring{$\bs{w}=(w^{(1)},..,w^{(D})^T$}{w}}
After estimating all group weights $\hat{\bs{\gamma}}^{(d)}$, $d=1,..,D$ for each grouping separately, we combine the groupings in a linear combination with grouping weights $\bs{w}=(w^{(1)},..,w^{(D})^T$. In order to obtain the estimate $\hat{\bs{w}}$, pool all $\mc{G}_{total}$ groups of all groupings and set up the moment equations as above in Equation \eqref{eq:MoM1} to find a linear system as above in Equation \eqref{eq:linearsystem}. By plugging in the estimates $\hat{\bs{\gamma}}^{(d)}$ and rearranging the equations, we obtain a linear system of $ G_{total}$ equations and $D$ unknowns, denoted by the matrix $\tilde{A}\in\mathbb{R}^{G_{total}\times D}$ and vector $\bs{b}_{\bs{w}}\in\mathbb{R}^{G_{total}}$. Details are given in Section \ref{ap:multiplecodata} in the Supplementary Material.

The grouping weights estimate $\hat{\bs{w}}$ is the ordinary least squares estimate truncated at $0$:
\begin{align}\label{eq:linsysP1}
    \hat{\bs{w}}&=(\tilde{\bs{w}})_+,\ \tilde{\bs{w}}=\argmin{\bs{w}} ||\tilde{A}\bs{w}-\bs{b}_{\bs{w}}||^2_2.
\end{align}

Note that, since $D<G_{total}$, the least squares solution leads to stable solutions. For highly correlated groupings, the grouping weights are correlated too, possibly leading to high variance in the grouping weight estimates. One should take care in interpreting grouping weights of highly correlated groupings.

\begin{figure}[h!]
    \centering
    \includegraphics[width=\linewidth]{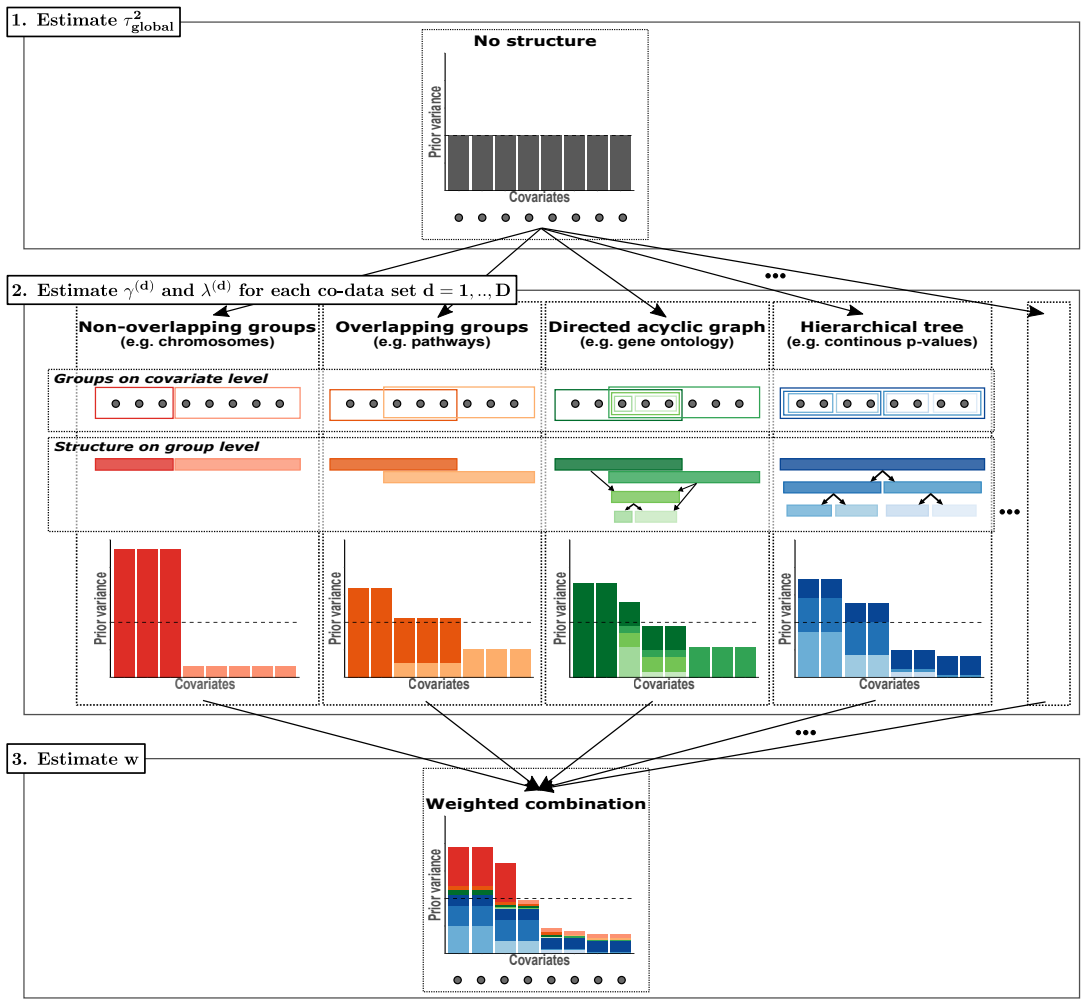}
    \caption{Schematical overview of estimating the hyperparameters. Step 1: overall level of regularisation, the global prior variance $\tau_{global}^2$, is estimated. Step 2: group weights $\bs{\gamma}^{(d)}$ and hyperpenalties $\lambda^{(d)}$, $d=1,..,D,$ are estimated for each co-data set separately using appropriate shrinkage. Step 3: grouping weights $\bs{w}$ are estimated to combine the co-data sets.
    The estimated hyperparameters are used to estimate the regression coefficients $\hat{\bs{\beta}}$ as given in Equation \eqref{eq:estbeta}.}
    \label{fig:schemeEstPrior}
\end{figure}

\subsection{Model extensions}\label{par:extensions}
We strive for a uniquely generic approach that can handle a wide variety of primary data (covariates and response) and co-data. The extensions below accommodate this aim.

\subsubsection{Continuous co-data}\label{par:continuous}
In principle, one could model a covariate specific prior variance as a (parsimonious) function of continuous co-data, like external p-values.
However, such a function is likely non-linear, and needs to be very flexible. 
We choose to approximate this function by adaptive discretisation, resulting in a piece-wise constant function. 
Adaptivity is necessary because the effect sizes are unknown, so for a continuous co-data set it might not be clear how fine a discretisation should be, if the discretisation should be evenly spaced, and if not, where on the continuous scale the discretisation should be finer. 

The approach is as follows. 
First define hierarchical groups, representing varying grid sizes:
i) define the first group as the group including all covariates, ordered according to the continuous co-data. When the co-data is not informative, using this group only would suffice. The group weight corresponding to the first group is defined to be the top \textit{node} in the hierarchical tree;
ii) recursively split each group $g$ at the median co-data value of group $g$ into two groups of half the size. The group weights corresponding to these latter two groups are defined as \textit{child nodes} from the \textit{parent node} for group weight $\gamma_g^{(d)}$ in the hierarchical tree, illustrated in Figure \ref{fig:discretisation}. We obtain a hierarchical tree where each node corresponds to a group weight.

This hierarchy is then used in a hierarchical lasso penalty (see \citep{yan2017hierarchical,jacob2009group,yang2015fast}), which is used as extra level of shrinkage in Equation \eqref{eq:tauest} to select hierarchical groups, illustrated in Figure \ref{fig:discretisation}.
The hierarchical lasso penalty can select a node only if all its parent nodes are selected.
Applied here, each selection of nodes corresponds to a selection of hierarchical groups, hence discretisation. 
For some hyperpenalty $\lambda^{(d)}$ large enough, only the top node in the hierarchy, corresponding to the group weight for the group of all covariates, is selected. For smaller values of the hyperpenalty, nodes lower in the hierarchy corresponding to large group weight estimates (i.e. small penalties) are selected first. Use the estimate for $\hat{\lambda}$ given in Equation \eqref{eq:lambda} to select group weights that correspond to a discretisation that fits the data well.

Each selected group corresponds to one moment equation in \eqref{eq:linearsystem}, enabling small groups deep in the hierarchy to have much larger weights than others. 
These moment equations are endowed with a ridge penalty as in Equation \eqref{eq:tauest} to stably estimate the final group weight estimates.

\begin{figure}
    \centering
    \begin{subfigure}[t]{0.49\linewidth}
    \centering
    \includegraphics[width=0.8\linewidth]{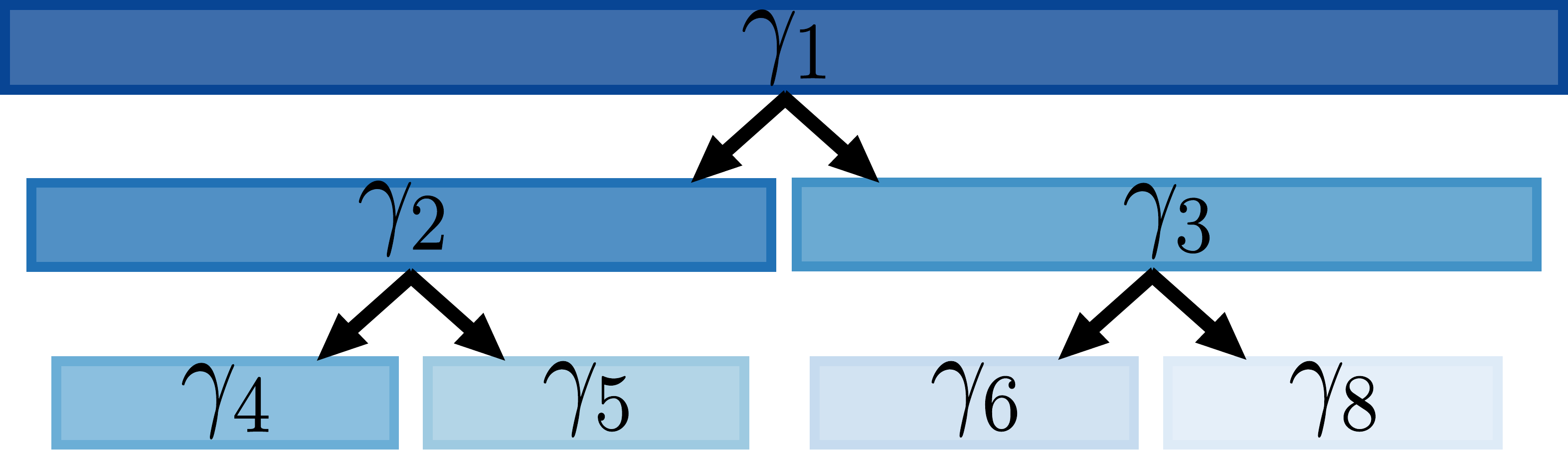}
    \end{subfigure}
    \begin{subfigure}[t]{0.49\linewidth}
    \centering
    \includegraphics[width=0.8\linewidth]{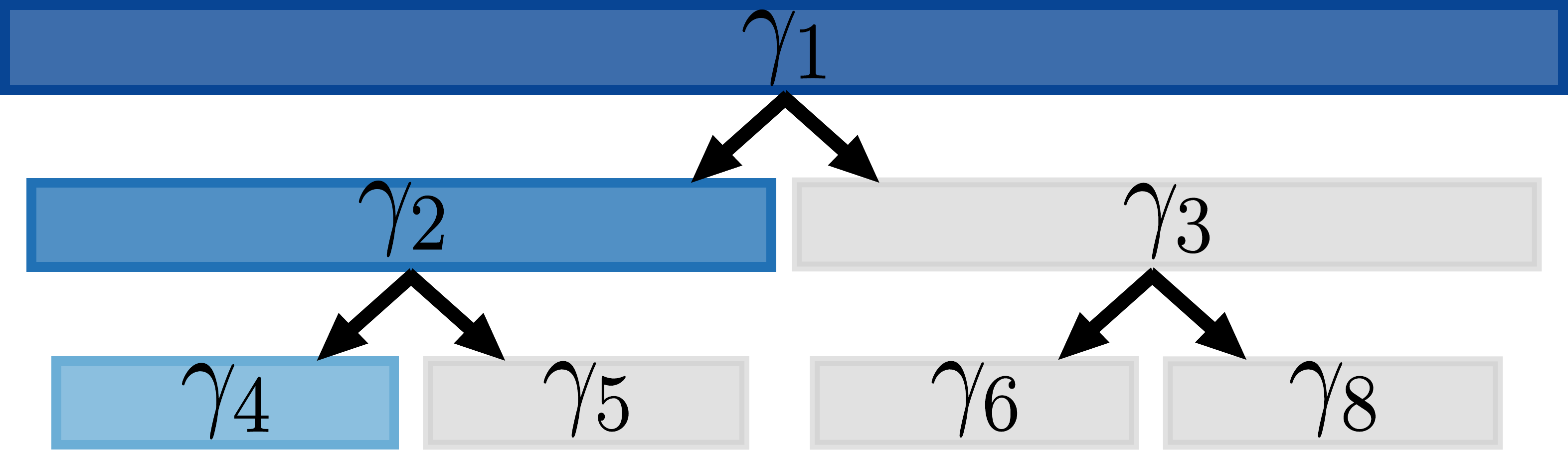}
    \end{subfigure}
    \caption{Left: discretise continuous scale in increasingly smaller groups by splitting at the median of the continuous co-data in that group. Right: use hierarchical lasso \citep{yan2017hierarchical,jacob2009group,yang2015fast} to potentially select group weights only if all its parents in the hierarchy (e.g. $\gamma_1$ is the parent of $\gamma_2$ and $\gamma_3$) are selected. Grey groups are not selected.}
    \label{fig:discretisation}
\end{figure}

\subsubsection{Group selection}
Group lasso and hierarchical lasso are popular methods to select groups of covariates on the covariate level \citep{yan2017hierarchical,jacob2009group,yang2015fast}, possibly shrinking covariates according to some given hierarchy. An alternative for obtaining a group sparse model is to use the proposed method in combination with a (hierarchical) sparse penalty on the group level; by setting group weights to $0$, all covariates in that group are set to $0$. When the number of covariates is much larger than the number of groups, it can be beneficial in terms of computational cost to use the (hierarchical) sparse penalisation on the group level. A similar two-step approach as for continuous co-data described in Section \ref{par:continuous} can be used: a lasso penalty or hierarchical lasso penalty is used to select groups on the group level, whereafter a ridge penalty is used to estimate the group weights of the selected groups.

\subsubsection{Covariate selection for prediction}\label{par:postselection}
In the applications that we consider, covariates may be (highly) correlated and the outcome might be predicted correctly only by a large group of interacting and correlated covariates. For example in genetics, gene expression is often correlated as many genes interact via complex networks of pathways. Moreover, predicting complex diseases might not always be as easy as finding few genes with large effects, as complex diseases could be the result of many small effects \citep{boyle2017expanded}. 
Penalties leading to dense predictors or group sparse predictors are well-known to handle correlated variables better than penalties leading to sparse predictors. 

In practice, however, it might be desirable to find a well-performing parsimonious predictor, e.g. due to budget constraints for practical implementation of the predictor. Various approaches have been proposed for sparsifying predictors.

First, \cite{bondell2012consistent} propose to perform variable selection based on penalised credible regions, searching for the sparsest predictor inside a penalised credible region. Their approach using marginal penalised credible regions is suitable for high-dimensional data, as they show that this approach can give consistent selection.
Second, a similar post-hoc selection strategy using an additional L1 penalty as performed in \texttt{GRridge} \citep{novianti2017better} was shown to perform well in terms of prediction for a number of cancer genomics applications \citep{novianti2017better}.
Third, decoupling shrinkage and selection \citep{hahn2015decoupling} approximates the linear predictor by a sparsified version using adaptive lasso. 

For completeness, we provide technical details of these approaches in Section \ref{ap:posthoc} in the Supplementary Material, and have included these three options in the \texttt{ecpc} software.

After selecting covariates, the regression coefficients are re-estimated using the weighted ridge prior to obtain the final predictor. Whether or not it is better to recalibrate the overall level of regularisation $\tau^2_{global}$ depends on the, unknown, underlying sparsity. If the best possible model is dense, the weighted ridge prior found in the first step should be used to prevent overestimation of the regression coefficients. If the best possible model is in fact sparse, it would be better to recalibrate $\tau^2_{global}$ and set group weights to $1$ to undo overshrinkage due to noise variables. We include both approaches as an option, which may be compared by considering predictive performance.

\subsubsection{Unpenalised covariates}
Sometimes one wishes to include unpenalised covariates, for example clinical covariates like tumour size or age of a patient. It can be shown that, conveniently, the moment estimates for penalised groups are independent of the group parameters for the group of unpenalised covariates. Details are given in Section \ref{ap:unpenalised} in the Supplementary Material. Then, in the model given in Equation \eqref{eq:model}, the Gaussian prior is only imposed on those covariates which are to be penalised.

\section{Simulation study}\label{par:simulation}
We use two applications to cancer genomics in Section \ref{par:application} to illustrate the method, termed \texttt{ecpc}: Empirical bayes Co-data learnt Prediction and Covariate selection, and to compare \texttt{ecpc} to other methods.
The purpose of the simulations is to show the benefit of using an extra level of shrinkage on the group weights. 
We demonstrate that when the co-data is not informative, the group weights and therefore local variances are shrunk to $1$, retrieving prediction errors similar to ordinary ridge. When the co-data is informative, the group weight estimates are shrunk little, improving the predictions compared to ordinary ridge. 

We consider linear regression for some fixed vector of regression coefficients $\beta^0$. We simulate $100$ pairs of training and test sets with the number of samples $n=100$ and the number of covariates $p=300$. We simulate for each pair of training and test sets, for variance parameters $\sigma^2=1,\tau^2=0.1$:
\begin{align}
\begin{split}
    &\bs{\beta}^0\sim N\left(0,\tau^2I_{p\times p}\right),\ \left[X_{train}\right]_{ij},\left[X_{test}\right]_{ij}\overset{i.i.d.}{\sim}N(0,1),\  i=1,..,n,\ j=1,..,p,\\
    &\bs{Y}_{train}\sim N\left(X_{train}\bs{\beta}^0,\sigma^2I_{n\times n}\right),\ \bs{Y}_{test}\sim N\left(X_{test}\bs{\beta}^0,\sigma^2I_{n\times n}\right).
\end{split}
\end{align}
Consider the following non-informative and informative co-data:
\begin{enumerate}[noitemsep,nolistsep]
\item \texttt{Random}: randomly assign the $300$ covariates to $G$ approximately equally sized groups, with $G$ in the range of $1-30$.
\item \texttt{Informative}: assign the covariates to $G$ approximately equally sized groups based on the ranking of the size of each regression coefficient, $|\beta^0_k|$, $k=1,..,p$. So there exists an ordering of the groups such that for each pair of two groups $\mc{G}_i,\mc{G}_j$, $1 \leq i<j\leq G$, and for all $k\in\mc{G}_i$, $l\in\mc{G}_j$: $|\beta^0_k|<|\beta^0_l|$.
\end{enumerate}

We use the default ridge penalty as hypershrinkage for the group weights with $1$ as target, such that the global-local prior variances $\tau^2_{global}\tau^2_{k,local}$ are shrunk to the global prior variance, corresponding to an ordinary ridge prior on the covariate level. We train the following models on the training data for both types of co-data and an increasing number of groups $G$: 1) \texttt{ecpc} with hypershrinkage; 2) \texttt{ecpc} without hypershrinkage, i.e. optimise the objective in Equation \eqref{eq:tauest} without any added penalty function; 3) \texttt{GRridge} \citep{wiel2016better}, which uses a regularisation on the group level based on permutations of the covariates' group indices, and 4) \texttt{ordinary ridge}, a ridge model that uses one overall penalty irrespective of the co-data groups. 

Figure \ref{fig:MSEgroups} shows the mean squared error (MSE) of the predictions on the test data as performance measure.
When the co-data is non-informative, \texttt{ecpc} with hypershrinkage performs similarly to \texttt{ordinary ridge} as the group weights of the random groups are shrunk towards $1$. Besides, \texttt{ecpc} with hypershrinkage outperforms both \texttt{ecpc} without hypershrinkage, as it is not able to shrink the group weights, and \texttt{GRridge}, which uses the more ad-hoc type of regularisation described above.
When the co-data is informative, \texttt{ecpc} with hypershrinkage shrinks little, performing similarly to \texttt{ecpc} without hypershrinkage, and outperforming \texttt{GRridge}. Moreover, all three methods outperform \texttt{ordinary ridge} as they benefit from the co-data. 

\begin{figure}
    \centering
    \includegraphics[width=\textwidth]{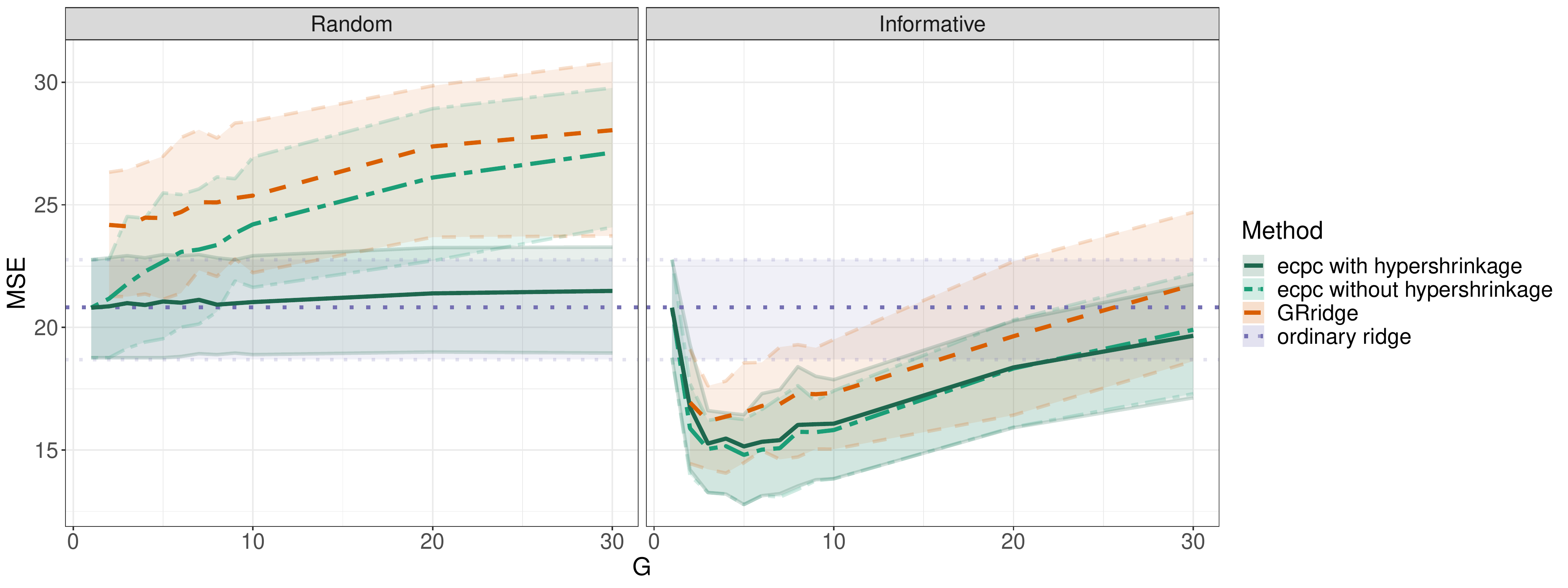}
    \caption{MSE of the predictions on the test sets for various number of groups, based on 100 training and test sets, for various methods and for random co-data (left) or informative co-data (right). The lines indicate the mean MSE and the shaded bands indicate the $25\%,75\%$ quantiles. }
    \label{fig:MSEgroups}
\end{figure}

\section{Application}\label{par:application}
We apply the \texttt{ecpc} method to two data applications in cancer genomics. In the first application the goal is to predict therapy response in colorectal cancer based on microRNA expression. In the second, the goal is to predict cervical cancer stage based on methylation data. Below we present the main results from the first application to illustrate \texttt{ecpc} and to compare performance and covariate selection with other widely used methods, and give a summary of the results of the second application. 
Note that we primarily focus on comparison with methods that allow to handle \textit{multiple} co-data sources, including continuous ones, as this is what we have available for the applications.
Results for baseline predictors are also given as a benchmark to assess the added value of the co-data.
We refer the reader to Section \ref{ap:applications} in the Supplementary Material for additional figures and results for both applications.

\subsection{Predicting therapy response in colorectal cancer}\label{par:appmiRNA}
We apply \texttt{ecpc} on microRNA (miRNA) expression data from a study on colorectal cancer, extensively described in \citep{neerincx2018combination,munch2018adaptive}. The data contain $p=2114$ measured miRNA expression levels for $n=88$ independent individuals, for whom we would like to predict whether a specific therapy described in \citep{neerincx2018combination} will be beneficial (coded 1) or not (coded 0).
In a previous study, \cite{neerincx2015mir} collected tissue from primary and metatastatic tumours plus adjacent normal tissue from a \textit{different} set of non-overlapping samples. 
The miRNA expression levels were measured and compared in a pairwise fashion, comparing metastatic or primary tumor to adjacent normal, to obtain false discovery rates (FDRs).
miRNAs that are expressed differentially in the tumour tissue compared to the adjacent normal tissue are potentially relatively important for predicting the therapy response.
The FDRs have been shown to be indeed informative for the prediction in \citep{munch2018adaptive} where \texttt{gren}, a group-regularised logistic elastic net regression is used. 
Unlike \texttt{ecpc}, \texttt{gren} requires \textit{non-adaptive} partitioning of the FDRs, using fairly arbitrary thresholds. In addition, it combined the two FDRs to limit the number of groups, as it does not allow for hyperparameter shrinkage as in \texttt{ecpc}.
So, we use partly the same co-data as \texttt{gren}, but add others as this can easily be handled by \texttt{ecpc}.

We use five co-data sets, two based on statistics derived from the data and three derived from the FDRs based on the external study: 1) (\texttt{abun}, 10 groups): abundance, i.e. average expression, of the miRNAs, discretised in $10$ non-overlapping, equally-sized groups; 2) (\texttt{sd}, 10 groups): standard deviation of the miRNA expression, discretised in $10$ non-overlapping, equally-sized groups. 
As we expect weights to change at most gradually with abundance and standard deviation, this non-adaptive discretisation should be sufficient to estimate the weights. Changing the number of groups to $5$ or $20$ leads to similar performance as presented below, and are included in Figure \ref{fig:weightsmiRNAabunsd} in the Supplementary Material; 3) (\texttt{TS}, 2 groups): one group with tumor specific miRNAs and one group with the rest. A miRNA is appointed to the tumor specific group if it is differentially expressed (FDR $\leq 0.05$) in the primary and/or metastatic tumor; 4) (\texttt{FDR1}, continuous): continuous FDRs from the comparison metastatic versus adjacent normal non-colorectal tissue; 5) (\texttt{FDR2}, continuous): continuous FDRs from the comparison primary versus normal colorectal tissue.
We generate a hierarchy of groups by recursively splitting the continuous co-data: each of the FDRs are first split into two groups at FDR$=0.5$. The group with FDR$<0.5$ is then recursively split into two groups as long as the minimum group size is not smaller than 20. Only the group with the lowest FDRs is split, as these groups are expected to be of more importance. The groups and hierarchy on the group level are illustrated in Figure \ref{fig:contcodatamiRNA}.

\begin{figure}
    \centering
    \includegraphics[width=\textwidth]{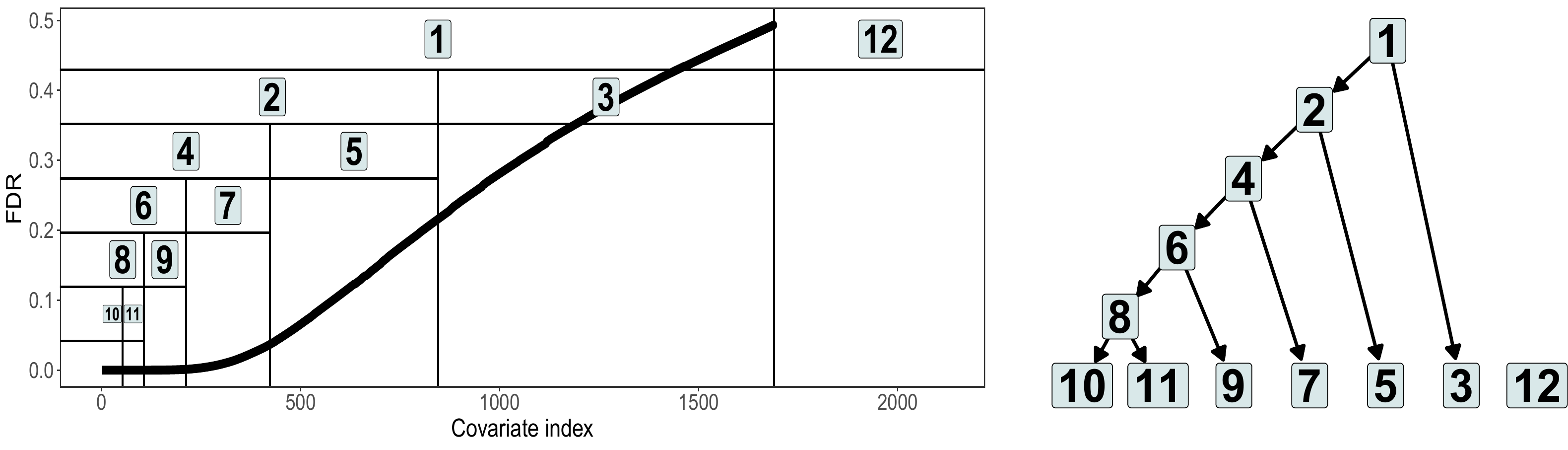}
    \caption{Illustration of the \texttt{FDR2} grouping used in the miRNA expression data. Left: covariates are first split at FDR$=0.5$ into two groups. The group with lower FDR is then recursively split at the median FDR value into two new groups. Right: the hierarchy of the groups, which is used in the extra level of shrinkage to find a discretisation that fits the data well as described in Section \ref{par:continuous}.}
    \label{fig:contcodatamiRNA}
\end{figure}

We use the default ridge penalty for the first three co-data groupings and the combination of the hierarchical lasso and ridge described in Section \ref{par:continuous} for the last two continuous co-data groupings. 
As posterior selection strategy, we use the default strategy using an additional L1-penalty. Below we show the results for posterior selection in a dense setting, as described in Section \ref{par:postselection}. This either matched or outperformed other posterior selection strategies, included in Figure \ref{fig:AUCposthoc} in the Supplementary Material.
We perform a 10-fold cross-validation to compare performance in terms of AUC of several dense and sparse methods. Different folds rendered similar results as shown below. 

\textbf{Estimated model parameters}. Figure \ref{fig:weightsmiRNA} shows the estimated co-data grouping weights and the group weights of grouping \texttt{FDR2} across folds. 
The grouping \texttt{FDR2} obtains on average the largest grouping weight.
Groups in this grouping with lower average FDR obtain a higher prior variance weight or equivalently, lower penalty. 
This corroborates the hypothesis that miRNAs that are more likely to be differentially expressed, are on average more predictive. 
The group weights of the other groupings are shown in Figure \ref{fig:weightsmiRNArest} in the Supplementary Material. 
\texttt{FDR1} shows a similar relation of groups of lower FDR obtaining higher prior variance weight, but obtains grouping weight of $0$ in most folds. In combination with the other groupings, it is the least informative for the prediction. In particular, the comparison between metastatic and normal tissue, \texttt{FDR1}, is less informative for this prediction model than the comparison between primary tumour and normal tissue, \texttt{FDR2}.
The group with lower FDR obtains a larger prior variance in the grouping \texttt{TS} as well, but this grouping obtains low grouping weight in most folds.  
The other groupings, \texttt{abun} and \texttt{sd} contain some information for the prediction, as the group weights are not fully shrunk to $1$. 
Moreover, the grouping weights are non-zero in most folds, indicating that these groupings are informative for the prediction.
The distribution of the estimated regression coefficients is more heavy-tailed when \texttt{ecpc} is used compared to when \texttt{ordinary ridge} is used, illustrated in Figure \ref{fig:miRNAheavytails} in the Supplementary Material. This facilitates posterior selection as the difference between small-sized regression coefficients and large-sized ones is larger.

\textbf{Performance}. Figure \ref{fig:AUCmiRNA} shows the cross-validated AUC for several dense and sparse models. The dense \texttt{ecpc} outperforms \texttt{GRridge}, \texttt{ordinary ridge} and \texttt{random forest}. The sparse \texttt{ecpc} is obtained by combining \texttt{ecpc} with the default post-hoc selection using an additional L1 penalty, tuned such that a fixed number of covariates is selected in each fold. It outperforms \texttt{elastic net} and \texttt{GRridge} with the same post-hoc selection. 
Whereas \texttt{ecpc} handles various co-data simultaneously and is able to give higher weight to more informative co-data, \texttt{GRridge} iterates over all co-data sets and suffers from including redundant co-data. 
Figure \ref{fig:AUCposthoc} in the Supplementary Material shows the performance of \texttt{ecpc} combined with other post-hoc selection methods. The proposed default of an added L1-penalty outperforms the other post-hoc selection methods. 
Besides, \texttt{ecpc} is combined with a lasso penalty on the group level to obtain a group sparse model. Figure \ref{fig:AUCgroupsparsemiRNA} in the Supplementary Material shows the AUC and number of selected covariates for several group sparse models. \texttt{ecpc} selects more groups and therefore more variables than group lasso and hierarchical lasso, and outperforms the latter two in terms of cross-validated AUC. 
Furthermore, the results of \texttt{gren} on this data set are presented in Figure 1b in \citep{munch2018adaptive}. These are competitive to ours, with an AUC around $0.8$, but only for \texttt{gren} using the \texttt{elastic net} parameter $\alpha=0.5$, which is not automatically chosen; other values of $\alpha$ render worse results. Besides, the number of covariates selected by \texttt{gren} is around $75$, which is much larger than the approximate $25$ covariates required by \texttt{ecpc} (Figure \ref{fig:AUCmiRNA}).
Lastly,  we compare with another recent group-adaptive method, \texttt{graper} \citep{velten2018adaptive}, which  can, however, not include overlapping (hierarchical) groups or multiple groupings. Therefore, as \texttt{ecpc} showed grouping \texttt{FDR2} to be informative (Figure \ref{fig:weightsmiRNA}), \texttt{graper} was applied to the leaf groups of this hierarchical grouping. This resulted in an AUC of $0.74$ (sparse setting) or $0.75$ (dense setting), hence somewhat lower than \texttt{ecpc} (Figure \ref{fig:AUCmiRNA}). Note that while the default sparse setting of \texttt{graper} provides inclusion probabilities for all covariates, it does not select covariates.

\textbf{Covariate selection stability}. We fit \texttt{ecpc} and \texttt{elastic net} for $\alpha=0.3$ and $\alpha=0.8$ on subsamples of size $\approx\frac{2}{3}n$, stratified for response, to assess stability of covariate selection. We use leave-one-out cross validation to estimate the global prior variance in \texttt{ecpc}, use the default post-hoc selection procedure to select 25 covariates for each subsample and count the number of overlapping miRNAs in each pairwise comparison of selected sets. For \texttt{elastic net}, we keep the value of $\alpha$ fixed and tune $\lambda$ to select 25 covariates. We repeat the analysis for a selection of $50$ covariates.
The AUC performance on the 50 test sets corresponding to the subsamples is included in Figure \ref{fig:AUCsubsamplesmiRNA} in the Supplementary Material.
Figure \ref{fig:overlapmiRNA} shows histograms of the amount of overlapping covariates between selections. \texttt{ecpc} results in a larger overlap between selections, indicating improved stability of the selection.

\begin{figure}
    \centering
    \begin{subfigure}[c]{0.49\textwidth}
    \centering
    \includegraphics[width=\linewidth]{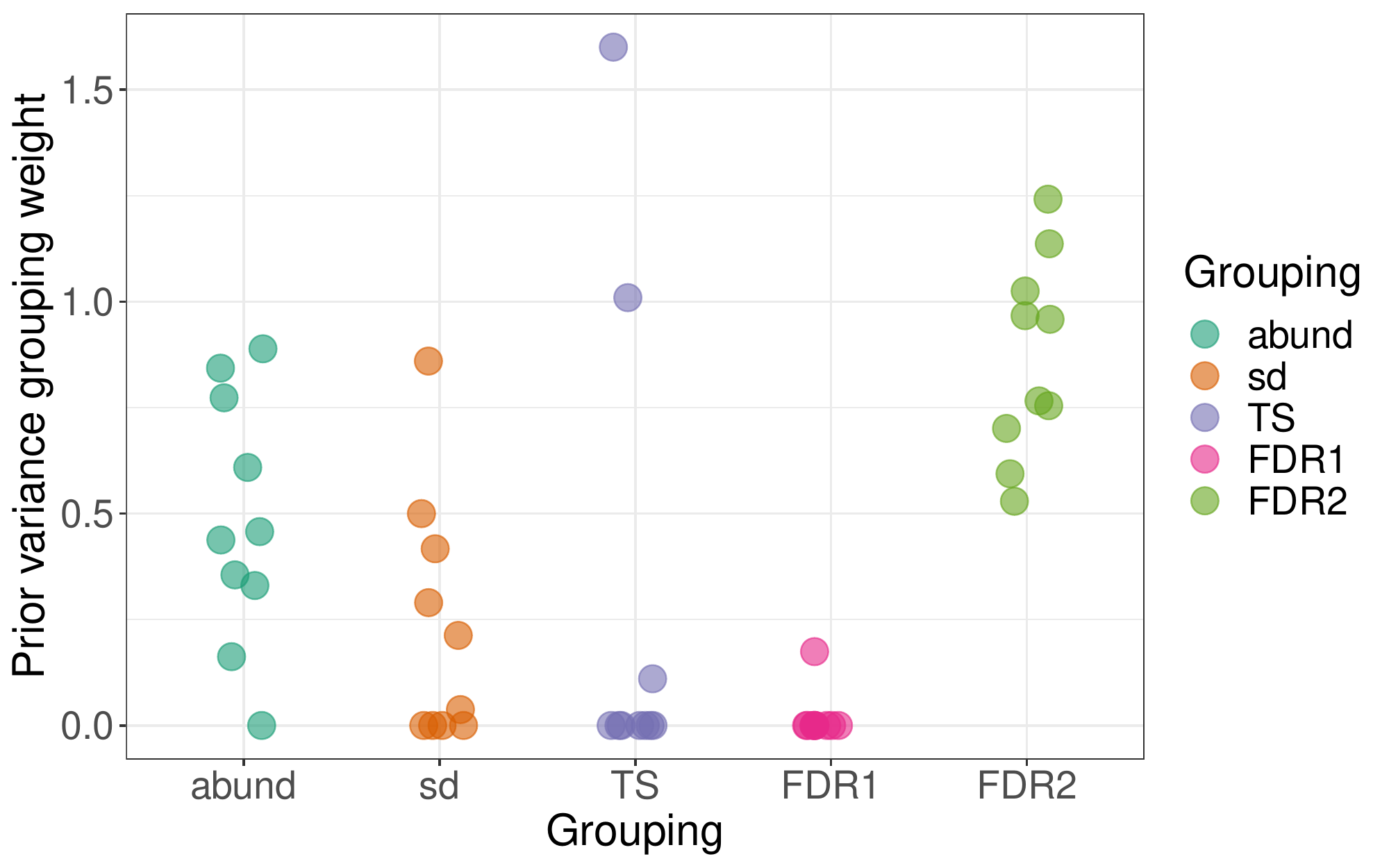}
    \end{subfigure}
    \begin{subfigure}[c]{0.49\textwidth}
    \centering
    \includegraphics[width=\linewidth]{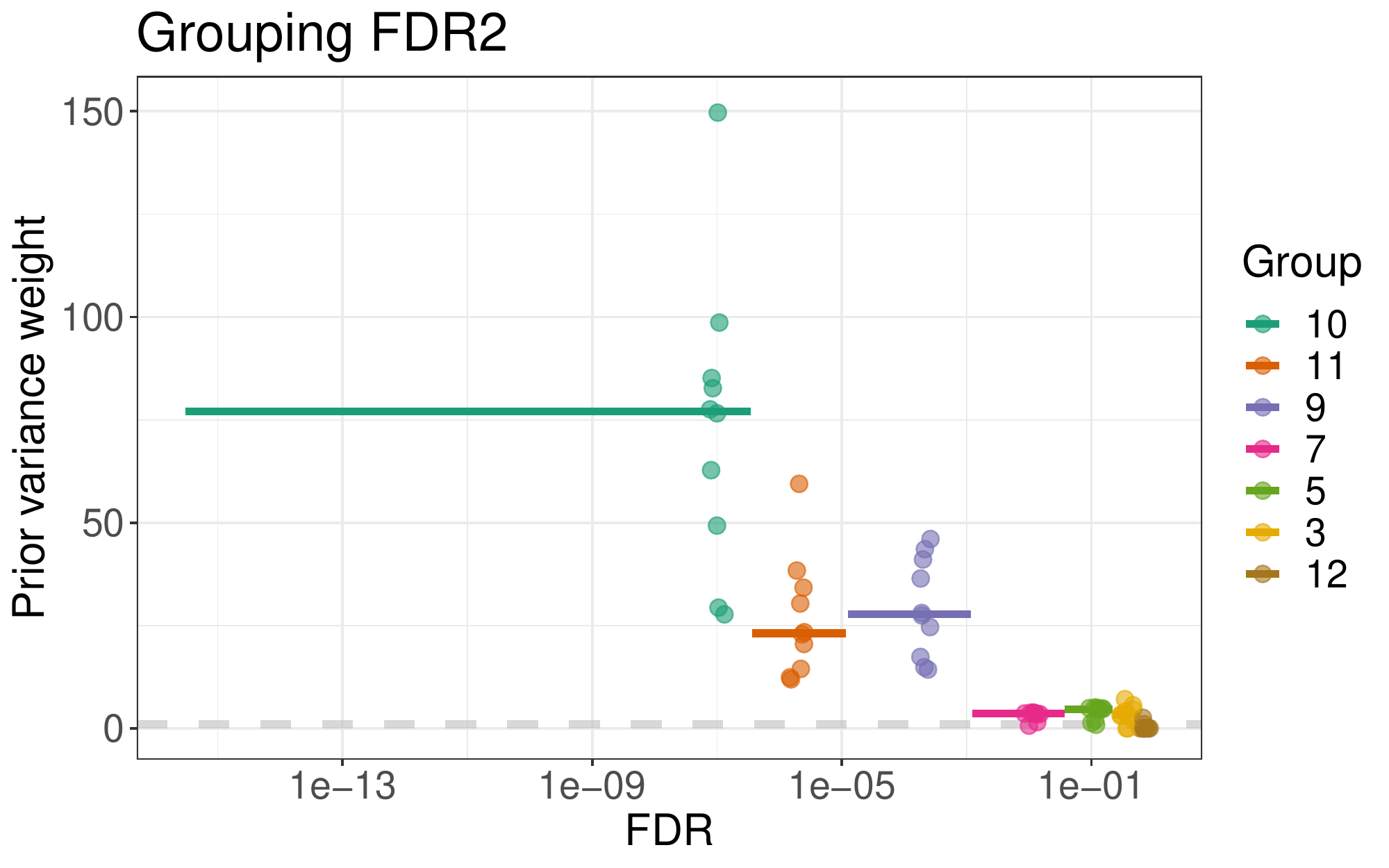}
    \end{subfigure}
    \caption{Results of 10-fold CV in miRNA data example. Left: estimated co-data grouping weights in each fold.
    Right: estimated local prior variance weight in \texttt{FDR2} for different continuous FDR values.
    The horizontal line segments indicate the median local prior variance in the leaf groups of the hierarchical tree illustrated in Figure \ref{fig:codataVerlaat}, ranging from the minimum to maximum p-value in that group.
    The points indicate the estimates in different folds, jittered along the median p-value in the leaf groups. 
    The dashed line at $1$ corresponds to ordinary ridge weights for non-informative co-data. A larger prior variance corresponds to a smaller penalty.
    }
    \label{fig:weightsmiRNA}
\end{figure}

\begin{figure}
    \centering
    \begin{subfigure}[c]{0.49\textwidth}
    \centering
    \includegraphics[width=\linewidth]{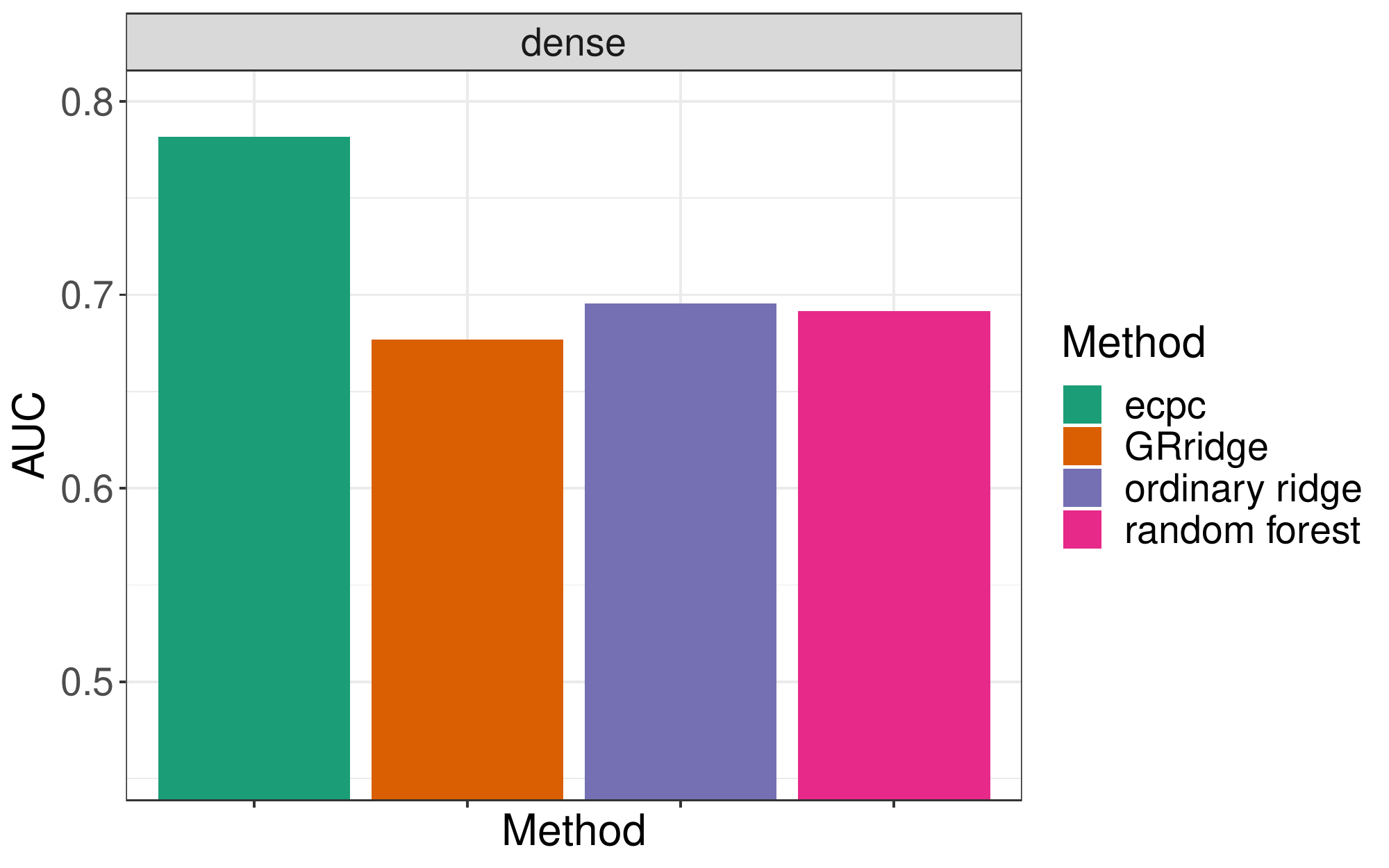}
    \end{subfigure}
    \begin{subfigure}[c]{0.49\textwidth}
    \centering
    \includegraphics[width=\linewidth]{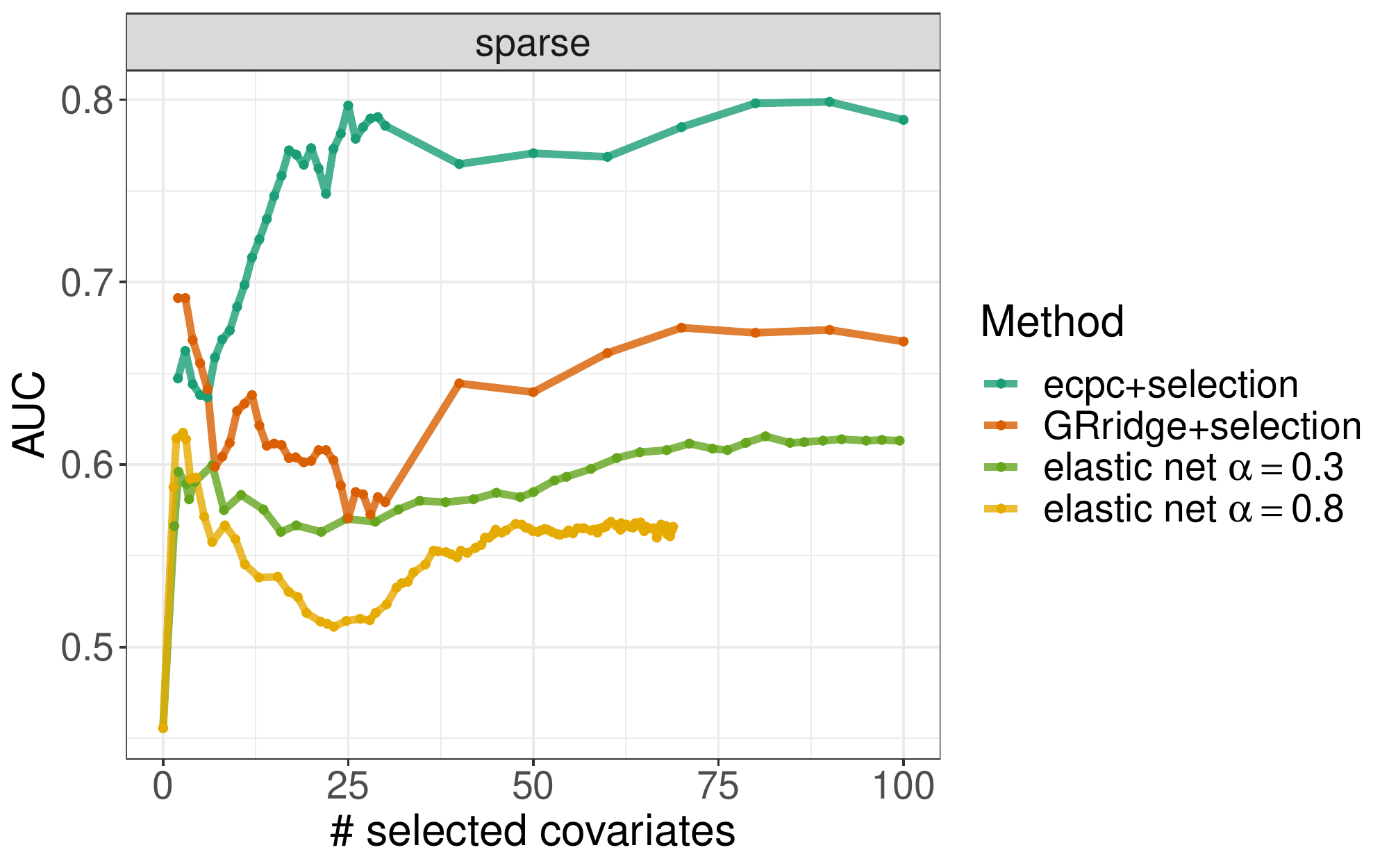}
    \end{subfigure}
    \caption{Results of 10-fold CV in miRNA data example. AUC in various dense models (left) and sparse models (right).}
    \label{fig:AUCmiRNA}
\end{figure}

\begin{figure}
    \centering
    \begin{subfigure}[c]{0.49\textwidth}
    \centering
    \includegraphics[width=\linewidth]{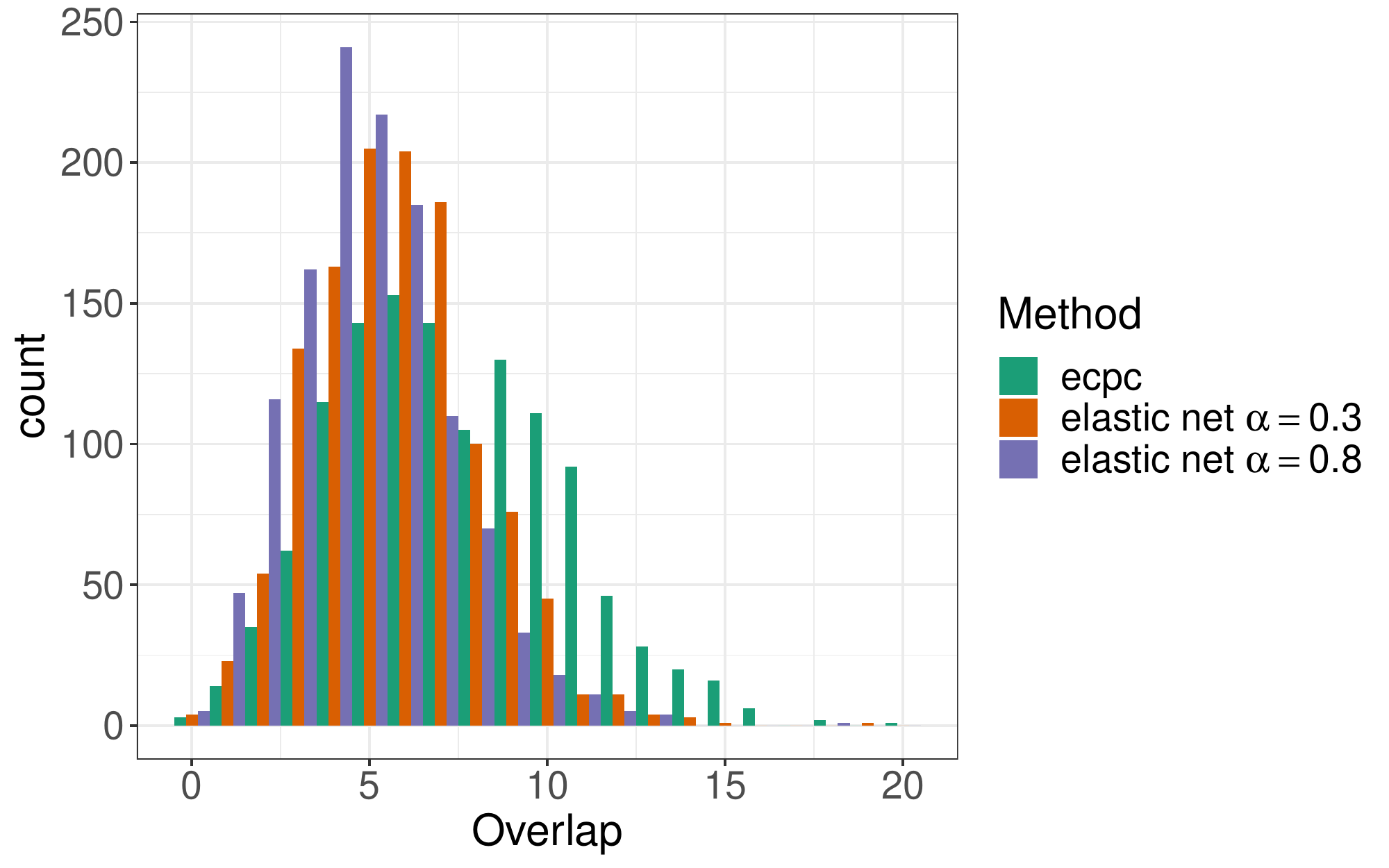}
    \end{subfigure}
    \begin{subfigure}[c]{0.49\textwidth}
    \centering
    \includegraphics[width=\linewidth]{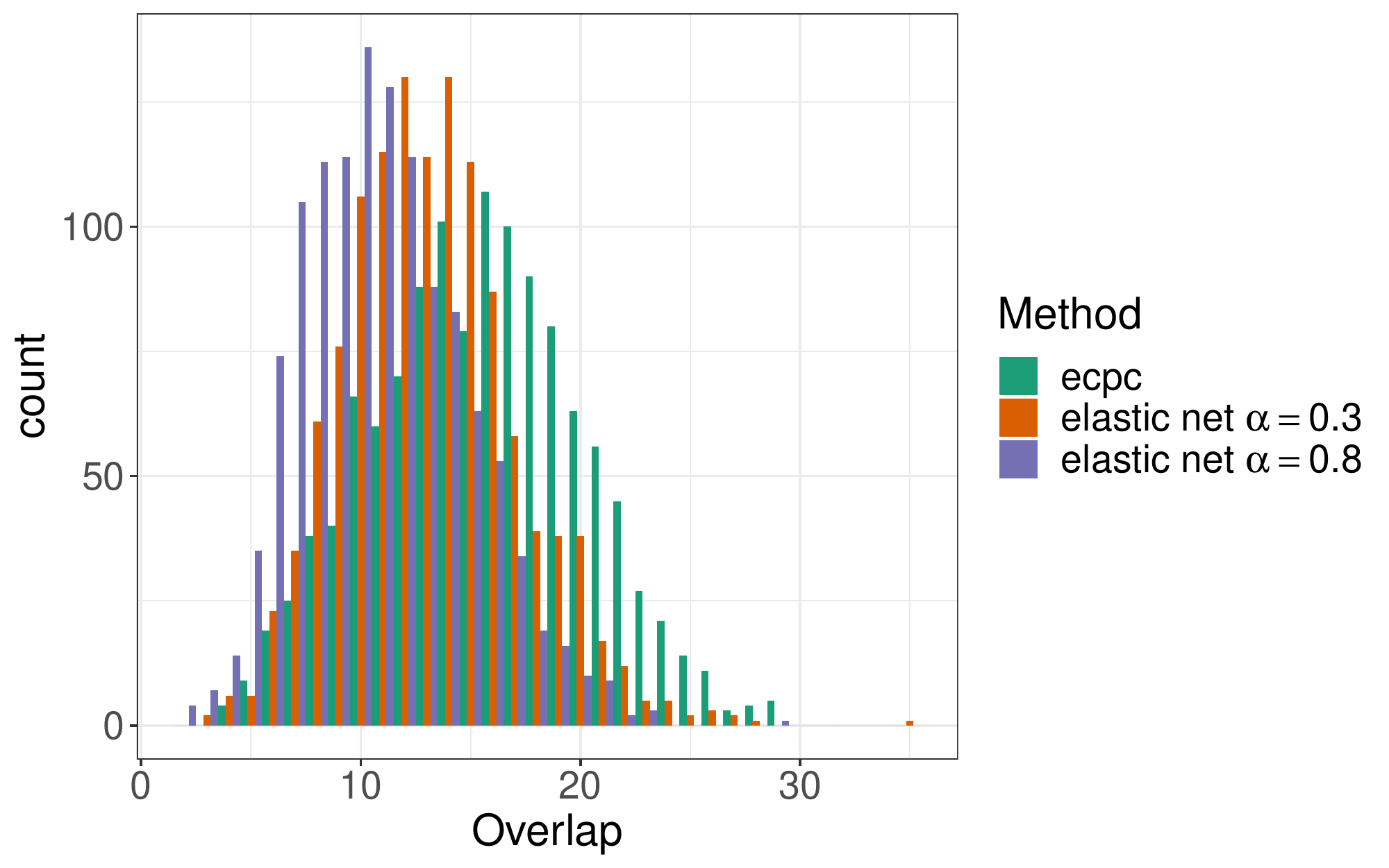}
    \end{subfigure}
    \caption{Results based on 50 stratified subsamples in miRNA data example. Histogram of number of overlapping variables in pairwise comparisons of selections of 25 covariates (left) or 50 covariates (right) in each subsample, for the methods \texttt{ecpc}, \texttt{elastic net} with $\alpha=0.3$ and $\alpha=0.8$..}
    \label{fig:overlapmiRNA}
\end{figure}

\subsection{Classifying cervical cancer stage}\label{par:appVerlaat}
We use methylation data from a study on cervical cancer extensively described in \citep{verlaat2018identification}. 
A CpG-location is a location on the DNA where a C base precedes a G base, with regions of a relatively high ratio of CpG locations called CpG-islands.
DNA methylation is a molecular mechanism that is known to play a role in cancer development. 
The goal is to find a classifier that best distinguishes normal tissue from CIN3 tissue, a stage with a high risk of progressing to cervical cancer, in self-taken samples of cervical tissue of women \citep{verlaat2018identification}. 
The methylation levels are measured in $n=64$ independent individuals with normal tissue (control) or CIN3 tissue (case). After prefiltering, the data consists of methylation levels of $p=2720$ probes corresponding to unique CpG-locations in the DNA.

We apply \texttt{ecpc} with and without post-hoc selection with the following two co-data sets, illustrated in Figure \ref{fig:codataVerlaat} in the Supplementary Material: 1) \texttt{CpG-islands}: five non-overlapping groups based on the genomic annotation of distance to the closest CpG-island. The five groups are, ordered in increasing distance: CpG-island, North Shore, South Shore, North Shelf and South Shelf. We use the default ridge shrinkage as extra level of shrinkage on the group level;
2) \texttt{p-values}: continuous p-values for each probe are obtained from an external, similar study \citep{farkas2013genome}. These data cannot be used directly for the classifier as the contamination by different cell types in these samples differs substantially from that of the primary data, the self-obtained samples. However, probes with lower p-values can be expected to be more important for the prediction than probes with high p-values. We adaptively discretise the p-values in a similar manner as the FDRs as described above.

We perform a 20-fold cross-validation to assess performance in terms of AUC for various dense, group sparse and covariate sparse methods. Different folds rendered similar results as shown below.
Again, we show the results for the default posterior selection strategy, using an additional L1-penalty. This matched or outperformed other posterior selection strategies, included in Figure \ref{fig:AUCposthocVerlaat} in the Supplementary Material.
Including standard deviations as another co-data grouping as in the first application rendered similar results in terms of performance.
Here we summarise the results. 

\textbf{Estimated model parameters}. 
Figure \ref{fig:weightsVerlaat} in the Supplementary Material shows the estimated grouping weights and group weights across the folds. The \texttt{p-value} grouping is the only grouping that is selected in all folds, indicating that this grouping is more informative for the prediction than the \texttt{CpG-islands} grouping. The Island and South shelf group obtain group weights higher than $1$ and are deemed more important for the prediction. Groups with lower average p-value obtain a higher prior variance or equivalently, lower penalty.
Similarly as in the first data application, \texttt{ecpc} facilitates posterior selection, as the distribution of the estimated regression coefficients is more heavy-tailed as compared to when \texttt{ordinary ridge} is used, illustrated in Figure \ref{fig:Verlaatheavytails} in the Supplementary Material.

\textbf{Performance}. 
Figure \ref{fig:AUCVerlaat} shows the AUC versus the number of selected parameters for several dense and covariate sparse methods. 
First, compared to other dense models, \texttt{ecpc} performs similar to \texttt{GRridge} and \texttt{ordinary ridge}, and outperforms \texttt{random forest}.
Then, compared to other covariate sparse models, \texttt{GRridge} outperforms the other methods for models with more than five selected covariates. \texttt{ecpc} results in a peak performance of an AUC$=0.73$ at $4$ parameters, outperforming the benchmark \texttt{elastic net} with $\alpha=0.3$ and $\alpha=0.8$. 
While \texttt{ecpc} is slightly superior to \texttt{GRridge} for very sparse models, its performance initially decreases when including more covariates, and then closes up on \texttt{GRridge} again when approaching $100$ covariates. 
We conjecture that this is due to the extremer weights \texttt{ecpc} assigns to the smallest p-value group. 
Besides, \texttt{ecpc} is combined with a lasso penalty on the group level to obtain a group sparse model. As shown in Figure \ref{fig:AUCVerlaatgroupsparse} in the Supplementary Material, \texttt{ecpc} selects more groups than group lasso and hierarchical lasso. Hierarchical lasso (AUC$=0.70$) selects only the one or two groups with lowest average p-value and slightly outperforms group lasso (AUC$=0.69$) and the group sparse version of \texttt{ecpc} (AUC$=0.67$).
Lastly, we apply \texttt{graper} to the leaf groups of the hierarchical p-value grouping, found to be most important by \texttt{ecpc} (Figure \ref{fig:weightsVerlaat} in the Supplementary Material). Then, \texttt{graper} slightly outperforms \texttt{ecpc} in the dense setting, with an AUC of $0.71$ and is competitive in the sparse setting, with an AUC of $0.70$. Note however that here \texttt{graper} uses information from \texttt{ecpc} on the most informative grouping for these data, which may introduce a benefit for the former.

\textbf{Covariate selection stability}.
We perform the same analysis based on subsamples of the data as used and described in the first data application to assess covariate selection stability. 
The AUC performance on the 50 test sets corresponding to the subsamples is included in Figure \ref{fig:AUCsubsamplesVerlaat} in the Supplementary Material.
Figure \ref{fig:overlapVerlaat} in the Supplementary Material shows histograms of the number of overlapping covariates in pairwise comparisons of selections of $25$ or $50$ covariates. 
Again, \texttt{ecpc} results in a larger overlap between selections when compared to \texttt{elastic net} for $\alpha=0.3$ and $\alpha=0.8$.

\begin{figure}
    \centering
    \begin{subfigure}[c]{0.49\textwidth}
    \centering
    \includegraphics[width=\linewidth]{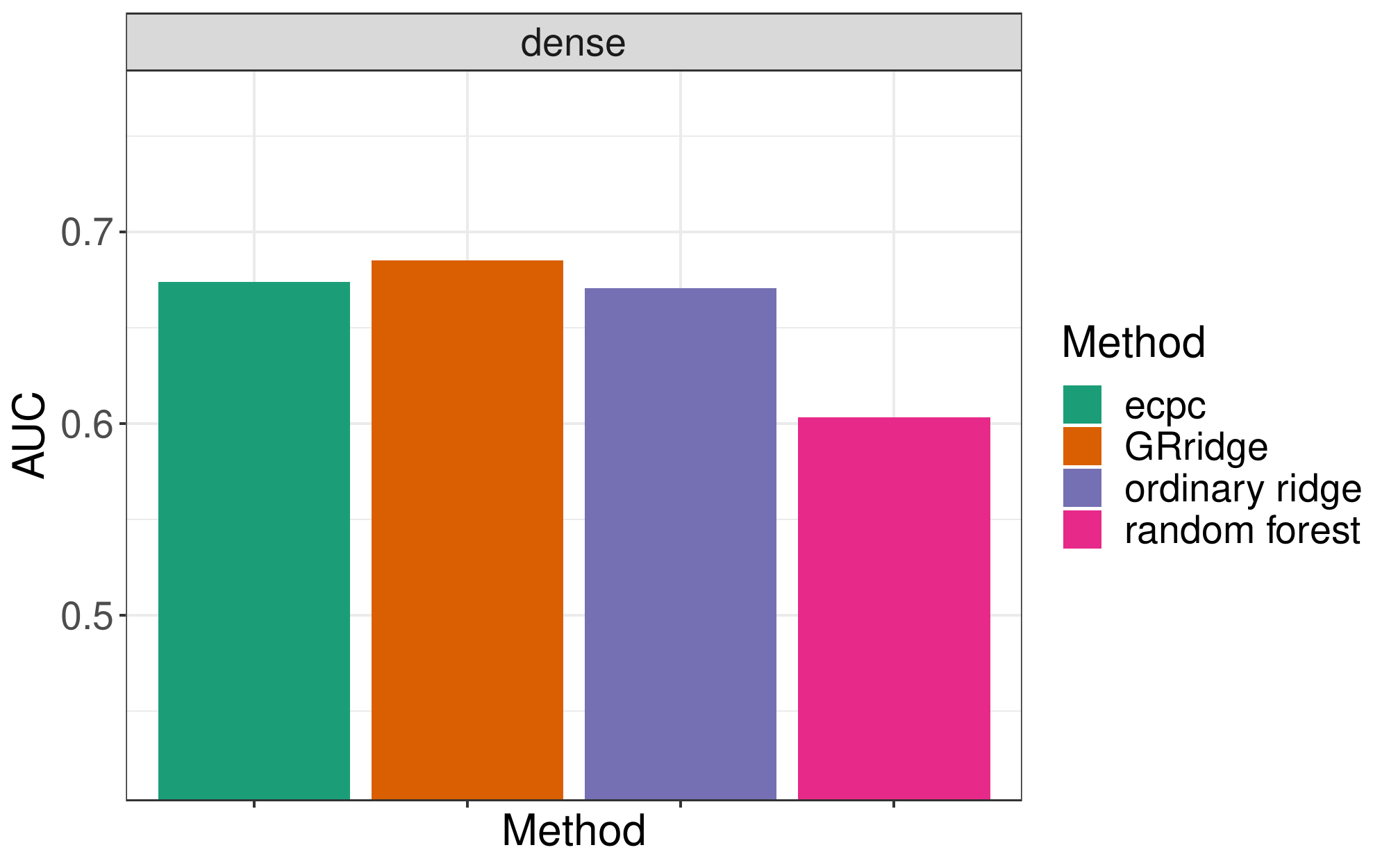}
    \end{subfigure}
    \begin{subfigure}[c]{0.49\textwidth}
    \centering
    \includegraphics[width=\linewidth]{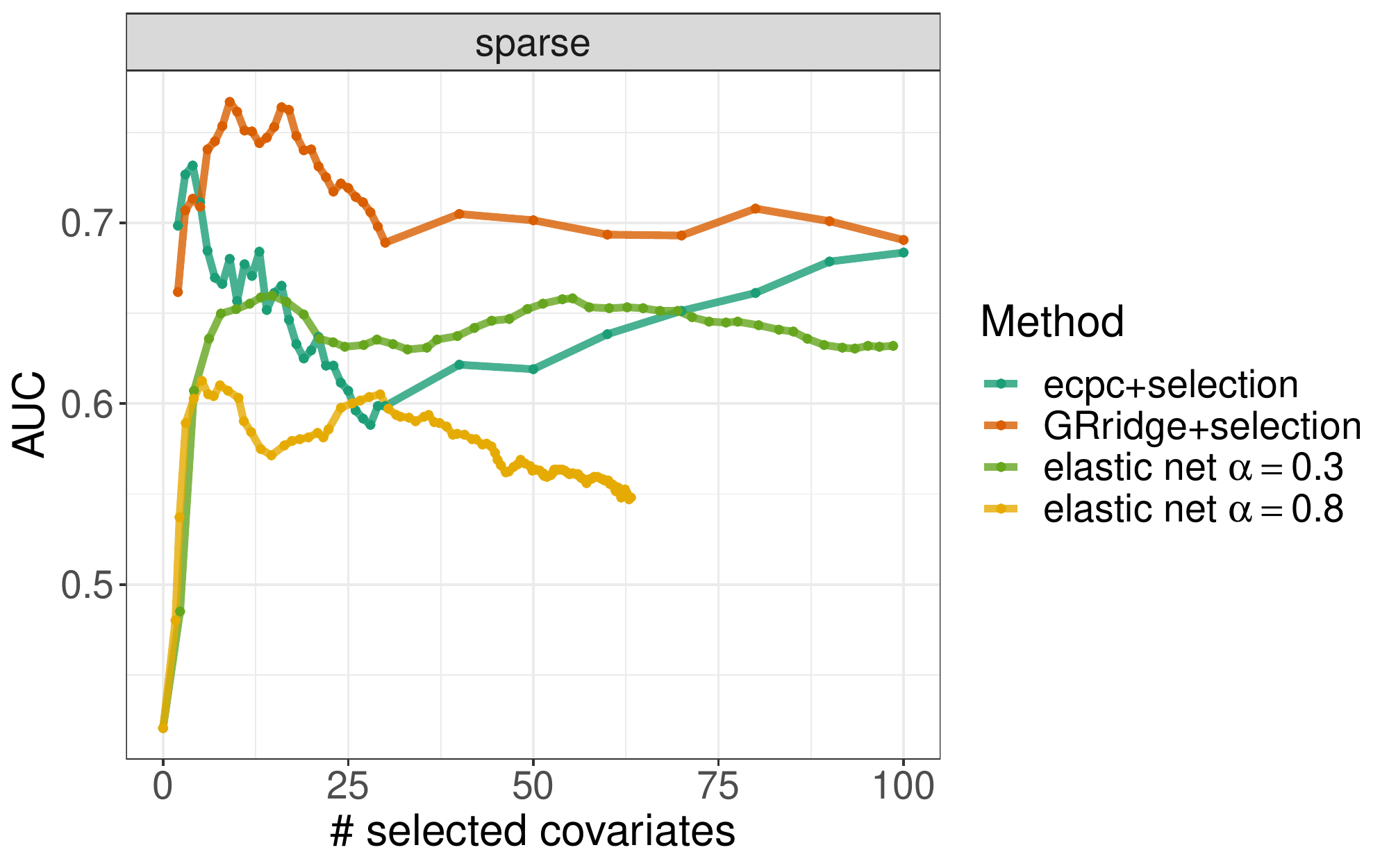}
    \end{subfigure}
    \caption{Results of 20-fold CV in Verlaat data example. AUC in various dense models (left) and sparse models (right).}
    \label{fig:AUCVerlaat}
\end{figure}

\section{Discussion}\label{par:discussion}
We presented a method, termed \texttt{ecpc}, to learn from multiple and various types of co-data to improve prediction and covariate selection for high-dimensional data, by adapting multi-group penalties in ridge penalised generalised linear models.
The method allows for missing co-data, unpenalised covariates and posterior variable selection. 
We introduced an extra level of shrinkage on the group level, rendering a unique, flexible framework that is able to obtain stable local penalties and to incorporate any additional imposed structure on the group level.
The benefit of stabilising the penalty estimates is illustrated in a simulation study; it prevents overfitting in the number of groups and borrows information to improve group penalty estimates.
We demonstrated the method in two cancer genomics applications using multiple, discrete and continuous, co-data. 
By adequately learning which co-data are informative and how to integrate multiple co-data sources, the method profits from the relevant co-data, without suffering from non-informative co-data. 
Thereby, it substantially improved performance in terms of AUC in the first application for dense and parsimonious models compared to the benchmark methods, ordinary ridge and elastic net. In the second application, it matched ordinary ridge and outperformed elastic net.
Moreover, for both applications \texttt{ecpc} is either competitive to or outperforming other methods that are able to include co-data, but from which none are able to handle multiple, various types of co-data adaptively.
Furthermore, we showed that the method stabilised covariate selection for parsimonious predictors in both applications compared to elastic net.

The proposed framework combines moment-based empirical Bayes estimation with an extra level of shrinkage. It is flexible as it allows many types of shrinkage on the group level, suitable for the co-data at hand. 
Our default hyperparameter shrinkage is a ridge hyperpenalty, which accounts for multiple, possibly overlapping groups. 
In addition, it may be combined with lasso-type penalties to handle group-sparsity, hierarchical co-data and continuous co-data.
The method may be extended by using different types of penalties for different types of co-data, such as fusion penalties for graphically group-structured co-data as discussed in \citep{beer2019incorporating}.
Besides fusion penalties on the group level, the method may also be extended to include fusion penalties on the covariate level. The latter is, however, less straightforward, as this changes the moment estimating equations non-trivially.

We account for potential differences in group sizes by using prior `null' group weights derived 
under the assumptions that i) groups are non-overlapping, and; ii) a priori, the group-
ing is not informative (the `null'; see Section \ref{ap:hypershrinkage} in the Supplementary Material). 
Potential group overlap could be accounted for by a generalised ridge hyperpenalty matrix with non-zero off-diagonal elements, although this may be time-consuming. 
While the prior `null' weights protect against overfitting (on the group level), they may not be optimal when the groups are informative. 
An interesting extension is to replace these `null' weights by hierarchical weights parsimoniously modeled from co-data on the group level, e.g. groups of, or test statistics for, groups of covariates. 
This essentially adds another level to our model.

The framework integrates multiple co-data by learning co-data weights, possibly deselecting non-informative co-data.
Estimating group weights per co-data source independently before integrating multiple co-data has computational advantages, as it can be done in parallel. Moreover, it again supports flexibility, as different types of hypershrinkage can be used for different types of co-data. 
Interactions between groups of different co-data sources are, however, not explicitly modelled.
If desired, co-data sources may be merged and expanded with interaction terms. This could be combined with additional hierarchical constraints on the group level, including group interaction effects only if one or both covariate groups are marginally important \citep{bien2013lasso} or vice versa, i.e. including covariate groups only when their group interaction effect is marginally important \citep{lim2015learning}.

The proposed model includes one global prior variance parameter to govern the overall level of regularisation. For multi-omics data, in which different sources of data are combined in one predictor, an omics type specific global prior variance parameter may be preferable in order to set different omics types to the same scale \citep{boulesteix2017ipf}. Multiple global prior variances can easily be included by rescaling the data matrix by the associated global variance weight \citep{wiel2016better}.

The proposed empirical Bayes approach utilises the Bayesian formulation with the normal prior as given in Equation \eqref{eq:model} to estimate the hyperparameters (or prior parameters).
Predictions for new samples are also based on point estimates of the best fitting generalised linear model.
Sample and model uncertainty is therefore not propagated in the predictions, which could be interesting to obtain prediction uncertainty intervals.
More computational expensive alternatives such as hierarchical full Bayes can account for uncertainty propagation, but are less flexible in usage. 
Hybrid versions of empirical and full Bayes approaches were demonstrated to leverage a good trade-off between the computational burden and ability to propagate model errors \citep{van2019learning}.
Hence, this is an interesting future direction.

The improvement in performance using the method compared to other applications depends on the quality and relevance of available co-data, but also on the level of sparseness of the ``true'' underlying data generating mechanism. 
Our method accommodates data ranging from group sparse to dense underlying distributions.
As demonstrated in the data applications, use of co-data facilitates posterior selection.
Yet at some point, sparse penalties may outweigh the benefits of including co-data and borrowing information using dense penalties.
Most omics prediction problems, however, are unlikely to be truly sparse \citep{boyle2017expanded}, although a parsimonious predictor can still \textit{predict} well.
Others have argued for ``decoupling shrinkage and selection" in dense \citep{bondell2012consistent} and sparse \citep{hahn2015decoupling} settings.
We follow their reasoning, although with a different implementation, namely by adding an L1 penalty to the ridge penalties, which performed superior for our applications.

We provide \texttt{R} scripts and data to reproduce analyses and figures, and the \texttt{R}-package \texttt{ecpc} and a script demonstrating the package on
\url{https://github.com/Mirrelijn/ecpc}. 
Currently, \texttt{ecpc} accommodates linear, logistic and Cox survival response, and multiple discrete or continuous co-data, using a ridge penalty as default hypershrinkage, possibly combined with a lasso penalty for group selection, or hierarchical lasso constraints for hierarchical group selection.

\section*{Acknowledgements}
The first author is supported by ZonMw TOP grant COMPUTE CANCER (40-
00812-98-16012).
The authors would like to thank Soufiane Mourragui (Netherlands Cancer Insitute) for the many fruitful discussions and Magnus M\"unch (Amsterdam UMC) for preparation of the microRNA data.


{\small
\bibliography{ManuscriptMvN} }

\bigskip
 
\noindent
\begin{tabular}{lll}
    {\scriptsize{\textsc{M.M. van Nee}}} & {\scriptsize{\textsc{L.F.A. Wessels}}} & {\scriptsize{\textsc{M.A. van de Wiel}}}\\
    {\scriptsize\textsc{Epidemiology and Biostatistics}} & {\scriptsize\textsc{Molecular Carcinogenesis}} & {\scriptsize\textsc{Epidemiology and Biostatistics}}\\
    {\scriptsize\textsc{Amsterdam UMC}} & {\scriptsize\textsc{Netherlands Cancer Institute}} & {\scriptsize\textsc{Amsterdam UMC}}\\
    {\scriptsize\textsc{The Netherlands}} & {\scriptsize\textsc{Computational Cancer Biology}} & {\scriptsize\textsc{The Netherlands}}\\
    {\scriptsize \textsc{E-mail:} \href{mailto:m.vannee@amsterdamumc.nl}{m.vannee@amsterdamumc.nl}} & {\scriptsize\textsc{Oncode Institute}} & {\scriptsize\textsc{MRC Biostatistics Unit}}\\
     & {\scriptsize\textsc{Intelligent Systems}}  & {\scriptsize\textsc{Cambridge University}}\\
     & {\scriptsize\textsc{Delft University of Technology}}  & {\scriptsize\textsc{UK}}\\
     & {\scriptsize\textsc{The Netherlands}}  & {\scriptsize \textsc{E-mail:} \href{mailto:mark.vdwiel@amsterdamumc.nl}{mark.vdwiel@amsterdamumc.nl}}\\
      & {\scriptsize \textsc{E-mail:} \href{mailto:l.wessels@nki.nl}{l.wessels@nki.nl}}  & \\
\end{tabular}%

\newpage

\begin{appendix}
\renewcommand{\appendixname}{Supplementary Material}
\section*{Supplementary Material}
\renewcommand{\appendixname}{}
\renewcommand{\thesection}{S\arabic{section}}
\setcounter{section}{0}
\setcounter{table}{0}
\renewcommand{\thetable}{S\arabic{table}}%
\setcounter{figure}{0}
\renewcommand{\thefigure}{S\arabic{figure}}%
\setcounter{equation}{0}
\renewcommand{\theequation}{S.\arabic{equation}}

\section{Details model estimation}\label{ap:mom}
The Method of Moments (MoM) can be used to obtain moment estimates for the prior parameters, as has been done before in \citep{wiel2016better} for obtaining group prior variance estimates for linear and logistic regression. Whereas the theoretical moments needed for MoM are analytical for linear regression, Taylor approximations as given in \citep{le1992ridge} are used and generalised to derive approximations for other generalised linear models (GLMs) using first and second order derivatives for GLMs as given in \citep{meijer2013efficient}. Besides, the approximation is extended to include moment estimations for group prior mean parameters as well. This could be used if one would want to shrink all $\beta$ not to $0$, but to a target (see for instance \citep{vanWieringen2015lecture}) where the target itself now is estimated based on the data.
By default, we use an inverse gamma penalty on the group level to ensure stable group variance estimates that are automatically shrunk towards an ordinary ridge prior weight when co-data is non-informative. Differences in group sizes are taken into account when shrinking group variance estimates.
The penalty matrices used will first be assumed to be of full rank, which doesn't hold in particular when unpenalised covariates are to be included. However, we can show that the MoM estimating equations can be derived independently of unpenalised covariates.

Below, we derive moment-estimates for group prior means and variances $\bs{\mu},\bs{\gamma}\in\mathbb{R}^{G}$, keeping notation similar to \citep{meijer2013efficient} in order to retrieve estimating equations general for all GLMs. We then fill in details for linear, logistic and Cox survival regression, and show how to use the same estimating equations to obtain co-data weights when combining multiple co-data sets. After showing how to handle unpenalised covariates, we give the details of the inverse gamma penalty function. Lastly, we give some details on the covariate selection approaches. 

\subsection{Generalised linear models and derivatives}
Consider one co-data set coded by the co-data matrix $Z\in\mathbb{R}^{p\times G}$, leaving out all superscripts $^{(d)}$ for notational convenience. Each $\beta_k$ is a priori Gaussian distributed with some covariate-specific mean $\mu_k$ and variance $\tau_k^2$ which are a function of the group specific prior mean vector $\bs{\mu}_{G\times 1}\in\mathbb{R}^G$ and overall and local prior variance $\tau^2_{overall}$, $\bs{\gamma}\in\mathbb{R}^G$:
\begin{align*}
    \beta_k\overset{ind.}{\sim} N(\mu_k,\tau_k^2) :=N(\bs{Z}_k\bs{\mu},\tau_{global}^2\bs{Z}_k\bs{\gamma}),\ k=1,..,p.
\end{align*}
Reparameterise by $\bs{\tau}_{G\times 1}^2=\tau_{global}^2\bs{\gamma}$, assume (an estimate of) $\tau_{global}^2$ to be given. Denote the prior mean vector and precision matrix in $p$ dimensions by
\begin{align}
    \bs{\mu}_{p\times 1}=Z\bs{\mu}_{G\times 1}\in\mathbb{R}^p,\ \Omega_{p\times p}=diag(Z\bs{\tau}^2_{G\times 1})^{-1}\in\mathbb{R}^{p\times p},
\end{align}
and assume that $\Omega_{p\times p}$ is of full rank.

The penalised log likelihood, denoted by $\ell^\lambda(\bs{\beta})$ in \citep{meijer2013efficient}, is, up to a constant $c$ independent of $\bs{\beta}$, the same as the log of the joint distribution over $Y$ and $\bs{\beta}$ given the penalty or prior parameters $\bs{\mu}_{G\times 1},\bs{\tau}_{G\times 1}$: $\pi(Y,\bs{\beta}|\bs{\mu}_{G\times 1},\bs{\tau}_{G\times 1})$:
\begin{align*}
    \ell^\lambda(\bs{\beta}) &= \ell(\bs{\beta}) - \frac{1}{2}[\bs{\beta}-\bs{\mu}_{p\times 1}]^T \Omega_{p\times p} [\bs{\beta}-\bs{\mu}_{p\times 1}] +c\\
    &= \log\pi(Y|\bs{\beta}) + \log\pi(\bs{\beta}|\bs{\mu}_{G\times 1},\bs{\tau}_{G\times 1}) + \frac{p}{2}\log |2\pi\Omega_{p\times p}|\\
    &= \log\pi(Y,\bs{\beta}|\bs{\mu}_{G\times 1},\bs{\tau}_{G\times 1}) + \frac{p}{2}\log |2\pi\Omega_{p\times p}|.
\end{align*}

\subsubsection{Derivatives of penalised likelihood}
Denote the first (partial) derivative of a function to a vector $\bs{\beta}$ by $\nabla_\beta$ and the second derivative by the Hessian $H_\beta$. As given in \citep{meijer2013efficient} and extended to including the target or prior mean vector $\bs{\mu}$, for a GLM with canonical link function, there exists a diagonal weight matrix $W(\bs{\beta})=\Var_{Y|\bs{\beta}}(Y)$, which is usually a function of $\beta$, such that the first and second derivative of the penalised likelihood are given by:
\begin{align}
    \frac{\partial \ell^\lambda(\bs{\beta})}{\partial \bs{\beta}} := \nabla_\beta \ell^\lambda(\bs{\beta}) 
    &= \nabla_\beta \log\pi(Y,\bs{\beta}|\bs{\mu}_{G\times 1},\bs{\tau}_{G\times 1}) = \nabla_\beta \log\pi(\bs{\beta}|Y,\bs{\mu}_{G\times 1},\bs{\tau}_{G\times 1})\nonumber\\
    &= X^T[\bs{y}-E_{\bs{y}|\bs{\beta}}(y)] - \Omega_{p\times p} [\bs{\beta}-\bs{\mu}_{p\times 1}].\\ 
    \frac{\partial^2 \ell^\lambda(\bs{\beta})}{\partial \bs{\beta} \partial \bs{\beta}^T} := H_\beta \ell^\lambda(\bs{\beta}) 
    &= H_\beta \log\pi(Y,\bs{\beta}|\bs{\mu}_{G\times 1},\bs{\tau}_{G\times 1}) = H_\beta \log\pi(\bs{\beta}|Y,\bs{\mu}_{G\times 1},\bs{\tau}_{G\times 1})\nonumber \\
    &= -X^TW(\bs{\beta})X - \Omega_{p\times p}.
\end{align}

\subsection{Moment estimating equations}

\subsubsection{Approximate mean and variance of penalised MLE}
As done in \citep{le1992ridge} for logistic regression, one can use a first order Taylor approximation of the score function in $\tilde{\bs{\beta}}(\bs{y},\tau^{overall})$ around $\bs{\beta}$ to find approximations for the mean and variance of the first smoothened estimate $\tilde{\bs{\beta}}$ using first estimates $\tilde{\bs{\mu}},\tilde{\bs{\tau}}^2,\tilde{\Omega},\tilde{W}:=W(\tilde{\bs{\beta}})$. Here we repeat some of the details, extended for GLMs with a target.

The first order Taylor approximation is given by
\begin{align}
\begin{split}
    \nabla_\beta \log\pi(\bs{y},\tilde{\bs{\beta}}|\tilde{\bs{\mu}}_{G\times 1},\tilde{\bs{\tau}}^2_{G\times 1}) & = \nabla_\beta \log\pi(\bs{y},\bs{\beta}|\tilde{\bs{\mu}}_{G\times 1},\tilde{\bs{\tau}}^2_{G\times 1}) \\
    &\qquad + H_\beta \log\pi(\bs{y},\bs{\beta}|\tilde{\bs{\mu}}_{G\times 1},\tilde{\bs{\tau}}^2_{G\times 1}) [\tilde{\bs{\beta}}-\bs{\beta}] + O(||\tilde{\bs{\beta}}-\bs{\beta}||^2).
\end{split}
\end{align}
As the score function is equal to $0$ in the penalised maximum likelihood estimate $\tilde{\bs{\beta}}$, we find the following first-order approximation for $\tilde{\bs{\beta}}$:
\begin{align}
    \tilde{\bs{\beta}} &\approx \bs{\beta} - [H_\beta\log\pi(\bs{y},\bs{\beta}|\tilde{\bs{\mu}}_{G\times 1},\tilde{\bs{\tau}}^2_{G\times 1})]^{-1}\nabla_\beta \log\pi(\bs{y},\bs{\beta}|\tilde{\bs{\mu}}_{G\times 1},\tilde{\bs{\tau}}^2_{G\times 1}).
\end{align}
For GLMs, this equation can be rewritten as:
\begin{align*}
    \tilde{\bs{\beta}} &\approx [X^TW(\bs{\beta})X + \tilde{\Omega}_{p\times p} ]^{-1}[X^T[\bs{y}-E_{\bs{y}|\bs{\beta}}(y)] - \tilde{\Omega}_{p\times p} [\bs{\beta}-\tilde{\bs{\mu}}_{p\times 1}] +[X^TW(\bs{\beta})X + \tilde{\Omega} ]\bs{\beta}] \\
    &= [X^TW(\bs{\beta})X + \tilde{\Omega}_{p\times p} ]^{-1}[X^T[\bs{y}-E_{\bs{y}|\bs{\beta}}(y)] + \tilde{\Omega}_{p\times p} \tilde{\bs{\mu}}_{p\times 1} +X^TW(\bs{\beta})X\bs{\beta}].
\end{align*}
The mean with respect to the likelihood $\pi(\bs{y}|\bs{\beta})$ is then given by:
\begin{align}\label{eq:meanbetatilde}
    E_{\bs{y}|\bs{\beta}}\tilde{\bs{\beta}}&\approx E_{\bs{y}|\bs{\beta}} \left[[X^TW(\bs{\beta})X + \tilde{\Omega}_{p\times p} ]^{-1}[X^T[\bs{y}-E_{\bs{y}|\bs{\beta}}(y)] + \tilde{\Omega}_{p\times p} \tilde{\bs{\mu}}_{p\times 1} +X^TW(\bs{\beta})X\bs{\beta}] \right] \nonumber\\
    &= [X^TW(\bs{\beta})X + \tilde{\Omega}_{p\times p} ]^{-1}[X^T[E_{\bs{y}|\bs{\beta}}(\bs{y})-E_{\bs{y}|\bs{\beta}}(y)] + \tilde{\Omega}_{p\times p} \tilde{\bs{\mu}}_{p\times 1} +X^TW(\bs{\beta})X\bs{\beta}] \nonumber \\
    &= [X^TW(\bs{\beta})X + \tilde{\Omega}_{p\times p} ]^{-1}[\tilde{\Omega}_{p\times p} \tilde{\bs{\mu}}_{p\times 1} +X^TW(\bs{\beta})X\bs{\beta}] \nonumber\\
    &= \tilde{\bs{\mu}}_{p\times 1} + [X^TW(\bs{\beta})X + \tilde{\Omega}_{p\times p} ]^{-1} X^TW(\bs{\beta})X [\bs{\beta} - \tilde{\bs{\mu}}_{p\times 1}] \nonumber\\
    &\approx \tilde{\bs{\mu}}_{p\times 1} + [X^T\tilde{W}X + \tilde{\Omega}_{p\times p} ]^{-1} X^T\tilde{W}X [\bs{\beta} - \tilde{\bs{\mu}}_{p\times 1}],
\end{align}
and the variance is given by the diagonal of the covariance matrix:
\begin{align}\label{eq:varbetatilde}
    \Cov_{\bs{y}|\bs{\beta}}\tilde{\bs{\beta}}&\approx \Cov_{\bs{y}|\bs{\beta}} \left[[X^TW(\bs{\beta})X + \tilde{\Omega}_{p\times p} ]^{-1}[X^T[\bs{y}-E_{\bs{y}|\bs{\beta}}(y)]  \right. \nonumber \\
    &\qquad + \tilde{\Omega}_{p\times p} \tilde{\bs{\mu}}_{p\times 1} \left.+X^TW(\bs{\beta})X\bs{\beta}] \right] \nonumber \\
    &= [X^TW(\bs{\beta})X + \tilde{\Omega}_{p\times p} ]^{-1}X^T\Cov_{\bs{y}|\bs{\beta}} \left[\bs{y}\right]X[X^TW(\bs{\beta})X + \tilde{\Omega}_{p\times p} ]^{-1} \nonumber \\
    &= [X^TW(\bs{\beta})X + \tilde{\Omega}_{p\times p} ]^{-1}X^TW(\bs{\beta}) X[X^TW(\bs{\beta})X + \tilde{\Omega}_{p\times p} ]^{-1} \nonumber \\
    &\approx [X^T\tilde{W}X + \tilde{\Omega}_{p\times p} ]^{-1}X^T\tilde{W} X[X^T\tilde{W}X + \tilde{\Omega}_{p\times p} ]^{-1}.
\end{align}
Note that we approximate the sample variance matrix $W$, which is still a function of $\bs{\beta}$, by $\tilde{W}$. For linear regression this approximation is in fact exact since $W$ does not depend on $\bs{\beta}$.

\subsubsection{Moment equations for prior mean}
The prior mean vector $\bs{\mu}_{G\times 1}$ can be computed by using the first moment. Denote $P_{G\leftarrow p}\in\mathbb{R}^{G\times p}$ as the matrix that averages the moments over each group, i.e. $[P_{G\leftarrow p}]_{gk}:=|\mc{G}_g|^{-1}\mathds{1}_{k\in\mc{G}_g}$. The system of moment estimating equations is given by:
\begin{align}
    &\left\{\begin{array}{l}
         \frac{1}{|\mathcal{G}_1|}\sum_{k\in\mathcal{G}_1} \tilde{\beta}_k = \frac{1}{|\mathcal{G}_1|}\sum_{k\in\mathcal{G}_1} E_{\bs{\beta}|\bs{\mu}_{G\times 1},\bs{\tau}_{G\times 1}}\left[E_{\bs{Y}|\bs{\beta}}\left[\tilde{\beta}_k\right]\right],\\
         \vdots\\
         \frac{1}{|\mathcal{G}_G|}\sum_{k\in\mathcal{G}_G} \tilde{\beta}_k = \frac{1}{|\mathcal{G}_G|}\sum_{k\in\mathcal{G}_g} E_{\bs{\beta}|\bs{\mu}_{G\times 1},\bs{\tau}_{G\times 1}}\left[E_{\bs{Y}|\bs{\beta}}\left[\tilde{\beta}_k\right]\right],\\
    \end{array}
    \right.\\
    &\Leftrightarrow \nonumber\\
    & P_{G\leftarrow p}\tilde{\bs{\beta}}=P_{G\leftarrow p}E_{\bs{\beta}|\bs{\mu}_{G\times 1},\bs{\tau}_{G\times 1}}\left[E_{\bs{Y}|\bs{\beta}}\left[\tilde{\bs{\beta}}\right]\right].
\end{align}

Plugging in the mean of Equation \eqref{eq:meanbetatilde} and further rewriting gives:
\begin{align*}
    P_{G\leftarrow p}\tilde{\bs{\beta}}&=P_{G\leftarrow p}E_{\bs{\beta}|\bs{\mu}_{G\times 1},\bs{\tau}_{G\times 1}}\left[E_{\bs{Y}|\bs{\beta}}\left[\tilde{\bs{\beta}}\right]\right]\\
    &\approx P_{G\leftarrow p}E_{\bs{\beta}|\bs{\mu}_{G\times 1},\bs{\tau}_{G\times 1}}\left[\tilde{\bs{\mu}}_{p\times 1} + [X^T\tilde{W}X + \tilde{\Omega}_{p\times p} ]^{-1} X^T\tilde{W}X [\bs{\beta} - \tilde{\bs{\mu}}_{p\times 1}]\right]\\
    &=P_{G\leftarrow p}\left[\tilde{\bs{\mu}}_{p\times 1} + [X^T\tilde{W}X + \tilde{\Omega}_{p\times p} ]^{-1} X^T\tilde{W}X [Z\bs{\mu}_{G\times 1} - \tilde{\bs{\mu}}_{p\times 1}]\right]. 
\end{align*}
If we define a matrix $C$ as follows then we can write the above as follows:
\begin{align}
    &C:=[X^T\tilde{W}X + \tilde{\Omega}_{p\times p} ]^{-1} X^T\tilde{W}X,\\
    &P_{G\leftarrow p}[\tilde{\bs{\beta}}-\tilde{\bs{\mu}}_{p\times 1}] = P_{G\leftarrow p} C Z[\bs{\mu}_{G\times 1}-\tilde{\bs{\mu}}_{G\times 1}].
\end{align}
So we find the following linear system:
\begin{align}\label{eq:linsysmu1}
&A_\mu\bs{\mu}_{G\times 1} = \bs{b}_\mu,\\ 
&A_\mu:= P_{G\leftarrow p} C Z,\\
&\bs{b}_\mu:= A_\mu\tilde{\bs{\mu}}_{G\times 1} + P_{G\leftarrow p}[\tilde{\bs{\beta}}-\tilde{\bs{\mu}}_{p\times 1}] = P_{G\leftarrow p}[\tilde{\bs{\beta}}-[I_{p\times p}-C]\tilde{\bs{\mu}}_{p\times 1}].
\end{align}
Note that $A_\mu\in\mathbb{R}^{G\times G}$ for $G\ll p$. Lastly, we can write each element of $A_\mu$ and $\bs{b}_\mu$ in the format of summing over groups as:
\begin{align}\label{eq:linsysmu2}
    [A_\mu]_{g,h} &= \frac{1}{|\mc{G}_g|}\sum_{k\in\mc{G}_g}\sum_{l\in\mc{G}_h}\frac{[C]_{k,l}}{|\mc{I}_l|},\\
    [b_\mu]_{g} &= \frac{1}{|\mc{G}_g|}\sum_{k\in\mc{G}_g}[\tilde{\bs{\beta}}-[I_{p\times p}-C]\tilde{\bs{\mu}}_{p\times 1}]_k.
\end{align}

\begin{remark}
In high-dimensional data, by default we will shrink to $0$, so $\tilde{\bs{\mu}}_{G\times 1}=\bs{0}=\bs{\mu}_{G\times 1}$.
\end{remark}

\subsubsection{Moment equations for prior variance}
The prior variance vector $\bs{\tau}_{G\times 1}$ can be computed by using the second moment equations and the estimate for $\bs{\mu}_{G\times 1}$. Use the same notation as above to denote $P_{G\leftarrow p}\in\mathbb{R}^{G\times p}$ as the matrix that averages the moments over each group, where $.^2$ denotes element-wise squaring:
\begin{align}
    &\left\{\begin{array}{l}
         \frac{1}{|\mathcal{G}_1|}\sum_{k\in\mathcal{G}_1} \tilde{\beta}_k^2 = \frac{1}{|\mathcal{G}_1|}\sum_{k\in\mathcal{G}_1} E_{\bs{\beta}|\bs{\mu}_{G\times 1},\bs{\tau}_{G\times 1}}\left[E_{\bs{Y}|\bs{\beta}}\left[\tilde{\beta}_k^2\right]\right],\\
         \vdots\\
         \frac{1}{|\mathcal{G}_G|}\sum_{k\in\mathcal{G}_G} \tilde{\beta}_k^2 = \frac{1}{|\mathcal{G}_G|}\sum_{k\in\mathcal{G}_g} E_{\bs{\beta}|\bs{\mu}_{G\times 1},\bs{\tau}_{G\times 1}}\left[E_{\bs{Y}|\bs{\beta}}\left[\tilde{\beta}_k^2\right]\right],\\
    \end{array}
    \right.\\
    &\Leftrightarrow \nonumber\\
    & P_{G\leftarrow p}\tilde{\bs{\beta}}.^2=P_{G\leftarrow p}E_{\bs{\beta}|\bs{\mu}_{G\times 1},\bs{\tau}_{G\times 1}}\left[E_{\bs{Y}|\bs{\beta}}\left[\tilde{\bs{\beta}}.^2\right]\right].
\end{align}
Use $diag(M):=([M]_{11},[M]_{22},..,[M]_{pp})^T$ to denote the diagonal vector of some matrix $M\in\mathbb{R}^{p\times p}$. Then we can derive, plugging in expressions of Equations \eqref{eq:meanbetatilde} and \eqref{eq:varbetatilde}:
\begin{align*}
    P_{G\leftarrow p}\tilde{\bs{\beta}}.^2&=P_{G\leftarrow p}E_{\bs{\beta}|\bs{\mu}_{G\times 1},\bs{\tau}_{G\times 1}}\left[\Var_{\bs{Y}|\bs{\beta}}\left[\tilde{\bs{\beta}}\right] + \left[E_{\bs{Y}|\bs{\beta}}\left[\tilde{\bs{\beta}}\right]\right].^2\right]\\
    &=P_{G\leftarrow p}\left\{E_{\bs{\beta}|\bs{\mu}_{G\times 1},\bs{\tau}_{G\times 1}}\left[\Var_{\bs{Y}|\bs{\beta}}\left[\tilde{\bs{\beta}}\right]\right] \right.\\
    &\qquad \left.+ \Var_{\bs{\beta}|\bs{\mu}_{G\times 1},\bs{\tau}_{G\times 1}}\left[E_{\bs{Y}|\bs{\beta}}\left[\tilde{\bs{\beta}}\right]\right] + E_{\bs{\beta}|\bs{\mu}_{G\times 1},\bs{\tau}_{G\times 1}}\left[E_{\bs{Y}|\bs{\beta}}\left[\tilde{\bs{\beta}}\right]\right].^2\right\}\\
    &=P_{G\leftarrow p}\left\{E_{\bs{\beta}|\bs{\mu}_{G\times 1},\bs{\tau}_{G\times 1}}\left[diag\left([X^T\tilde{W}X + \tilde{\Omega}_{p\times p} ]^{-1}X^T\tilde{W} X[X^T\tilde{W}X + \tilde{\Omega}_{p\times p} ]^{-1}\right)\right] \right.\\
    &\qquad \left. + \Var_{\bs{\beta}|\bs{\mu}_{G\times 1},\bs{\tau}_{G\times 1}}\left[[X^T\tilde{W}X + \tilde{\Omega}_{p\times p} ]^{-1}[\tilde{\Omega}_{p\times p} \tilde{\bs{\mu}}_{p\times 1} +X^T\tilde{W}X\bs{\beta}]\right] \right.\\
    &\left.\qquad + E_{\bs{\beta}|\bs{\mu}_{G\times 1},\bs{\tau}_{G\times 1}}\left[[X^T\tilde{W}X + \tilde{\Omega}_{p\times p} ]^{-1}[\tilde{\Omega}_{p\times p} \tilde{\bs{\mu}}_{p\times 1} +X^T\tilde{W}X\bs{\beta}]\right].^2\right\}\\
    &=P_{G\leftarrow p}\left\{diag\left([X^T\tilde{W}X + \tilde{\Omega}_{p\times p} ]^{-1}X^T\tilde{W} X[X^T\tilde{W}X + \tilde{\Omega}_{p\times p} ]^{-1}\right) \right.\\
    &\qquad + diag\left([X^T\tilde{W}X + \tilde{\Omega}_{p\times p} ]^{-1} X^T\tilde{W}X \Cov_{\bs{\beta}|\bs{\mu}_{G\times 1},\bs{\tau}_{G\times 1}}\left[\bs{\beta}\right]   \right.\\
    &\qquad \qquad \cdot\left.X^T\tilde{W}X [X^T\tilde{W}X + \tilde{\Omega}_{p\times p} ]^{-1}  \right)\\
    &\left.\qquad + \left[[X^T\tilde{W}X + \tilde{\Omega}_{p\times p} ]^{-1}[\tilde{\Omega}_{p\times p} \tilde{\bs{\mu}}_{p\times 1} +X^T\tilde{W}XZ\bs{\mu}_{G\times 1}]\right].^2\right\}.
\end{align*}
Again using the matrix $C$ as above, and $\tilde{\bs{v}}$ as vector for the variance, we can write:
\begin{align}\label{eq:Cv}
    &C:=[X^T\tilde{W}X + \tilde{\Omega}_{p\times p} ]^{-1} X^T\tilde{W}X,\\
    &\tilde{\bs{v}}:=diag\left([X^T\tilde{W}X + \tilde{\Omega}_{p\times p} ]^{-1}X^T\tilde{W}X[X^T\tilde{W}X + \tilde{\Omega}_{p\times p} ]^{-1}\right),\\
    &P_{G\leftarrow p}\tilde{\bs{\beta}}.^2 
    = P_{G\leftarrow p}[\tilde{\bs{v}}+ C.^2Z\bs{\tau}_{G\times 1} + [[I-C]\tilde{\bs{\mu}}_{p\times 1} + CZ\bs{\mu}_{G\times 1}].^2],
\end{align}
and then we find the linear system
\begin{align}\label{eq:linsystau1}
&A_\tau\bs{\tau}_{G\times 1} = \bs{b}_\tau,\\ 
&A_\tau:= P_{G\leftarrow p} C.^2 Z,\\
&\bs{b}_\tau:= P_{G\leftarrow p}[\tilde{\bs{\beta}}.^2-[[I-C]\tilde{\bs{\mu}}_{p\times 1} + CZ\bs{\mu}_{G\times 1}].^2 - \tilde{\bs{v}}].
\end{align}
Note that $A_\tau\in\mathbb{R}^{G\times G}$ for $G\ll p$. Again, we can write each element of $A_\tau$ and $\bs{b}_\tau$ in the format of summing over groups as:
\begin{align}\label{eq:linsystau2}
    [A_\tau]_{g,h} &= \frac{1}{|\mc{G}_g|}\sum_{k\in\mc{G}_g}\sum_{l\in\mc{G}_h}\frac{[C].^2_{k,l}}{|\mc{I}_l|},\\
    [b_\tau]_{g} &= \frac{1}{|\mc{G}_g|}\sum_{k\in\mc{G}_g}[\tilde{\bs{\beta}}.^2-[[I-C]\tilde{\bs{\mu}}_{p\times 1} + CZ\bs{\mu}_{G\times 1}].^2 - \tilde{\bs{v}}]_k.
\end{align}

\begin{remark}
In high-dimensional data, most of the times we will shrink to $0$, so $\tilde{\bs{\mu}}_{G\times 1}=\bs{0}=\bs{\mu}_{G\times 1}$
\end{remark}

\subsection{Moment equations for multiple co-data sets}\label{ap:multiplecodata}
For multiple co-data sets, each $\beta_k$ is a priori distributed as:
\begin{align*}
    \beta_k\overset{ind.}{\sim} N(\mu_k,\tau_k^2) :=N\left(\sum_{d=1}^Dw^{(d)}\bs{Z}^{(d)}_k\bs{\mu}^{(d)},\tau_{global}^2\sum_{d=1}^Dw^{(d)}\bs{Z}^{(d)}_k\bs{\gamma}^{(d)}\right),\ k=1,..,p.
\end{align*}
 We can pool all $G_{total}:=\sum_{d=1}^DG^{(d)}$ groups of all co-data sets together and use the same method of moment equations as above to derive moment estimates for the co-data weights. 
 In what follows, assume that we shrink all $\beta_k$ to $0$, i.e. $\bs{\mu}^{(d)}=0$ for all $d=1,..,D$. A similar argument using the first moments only can be used if non-zero targets are to be used. 
 To be able to use the same notation as above, define:
\begin{align}
    &Z= \left[Z^{(1)}\ \cdots \ Z^{(D)}\right],\\
    &\bs{\tau}_{G_{total}\times 1}:= \tau_{overall}^2[(w^{(1)}\bs{\gamma}^{(1)})^T\ \cdots\ (w^{(D)}\bs{\gamma}^{(D)})^T] ^T,\\ 
    &\bs{\tau}_{p\times 1} = \tau_{overall}^2\sum_{d=1}^Dw^{(d)}\bs{Z}^{(d)}\bs{\gamma}^{(d)} = Z\bs{\tau}_{G_{total}\times 1}.
\end{align}
Then we can follow the reasoning similar to above to arrive at the linear system as in Equation \eqref{eq:linsystau1}, where we have used that $\tilde{\bs{\mu}}_{p\times 1}=0=\bs{\mu}_{p\times 1}$:
\begin{align*}
&A_w\bs{\tau}_{G_{total}\times 1} = \bs{b}_w,\\ 
&A_w:= P_{G_{total}\leftarrow p} C.^2 Z,\\
&\bs{b}_w:= P_{G_{total}\leftarrow p}[\tilde{\bs{\beta}}.^2- \tilde{\bs{v}}],
\end{align*}
but now for $A_w\in\mathbb{R}^{G_{total}\times G_{total}}$ and $\bs{b}_w\in\mathbb{R}^{G_{total}}$. Plugging in the estimates for $\hat{\tau}^2_{overall}$ and $\hat{\bs{\gamma}}^{(d)}$, $d=1,..,D$, we find the linear system for the vector of $D$ unknown co-data weights $\bs{w}=(w^{(1)},..,w^{(D)})^T$:
\begin{align*}
&\tilde{A}_w\bs{w} = \bs{b}_w,
\end{align*}
with $\tilde{A}_w\in\mathbb{R}^{G_{total}\times D}$, and each column $[\tilde{A}]_{*,d}$ given by:
\begin{align*}
    [\tilde{A}]_{*,d}&=\hat{\tau}^2_{overall}\left[A_w\right]_{*,(1+\sum_{d'=1}^{d-1}G^{(d')}):(\sum_{d'=1}^{d}G^{(d')})}\hat{\bs{\gamma}}^{(d)}.
\end{align*}

\subsection{Details for specific examples}
The moment equations boil down to a linear system for $\bs{\mu}$ as given in Equations \eqref{eq:linsysmu1} and \eqref{eq:linsysmu2} and one for $\bs{\tau}$ as given in Equations \eqref{eq:linsystau1} and \eqref{eq:linsystau2}. These equations use the matrix $C\in\mathbb{R}^{p\times p}$ and vector $v\in\mathbb{R}^p$ as defined in Equation \eqref{eq:Cv}. To retrieve the moment equations for a specific GLM with link function $g^{-1}(\cdot)$, we only need an expression for the GLM-specific variance matrix $W(\bs{\beta})=\Var_{Y|\bs{\beta}}(Y)$. 

Below we give the details for linear, logistic and Cox survival regression.

\subsubsection{Linear regression}
For linear regression, the response $Y$ is gaussian distributed around the mean $X\bs{\beta}$ with variance $\sigma^2$ and following link function:
\begin{align}
    y_i\overset{ind.}{\sim} N\left(\bs{X}_i\bs{\beta}, \sigma^2\right),\ g^{-1}(\bs{X}_i\bs{\beta})&=\bs{X}_i\bs{\beta},\ i=1,..,n.
\end{align}
The matrix $\tilde{W}:=W(\tilde{\bs{\beta}})$ is given by:
\begin{align}
    W(\tilde{\bs{\beta}}) = \sigma^2I_{n\times n}.
\end{align}
The approximations for the mean and variance in Equations \eqref{eq:meanbetatilde} and \eqref{eq:varbetatilde} are in fact exact for the linear regression case.

\subsubsection{Logistic regression}
For linear regression, the response $Y$ follows a Bernoulli distribution with the vector of probabilities denoted by $\bs{p}=(p_1,..,p_n)^T$, and with the following link function:
\begin{align}
    y_i\overset{ind.}{\sim} Ber\left(p_i\right),\ g^{-1}(\bs{X}_i\bs{\beta})=p_i:=\frac{\ee(\bs{X}_i\bs{\beta})}{1+\ee(\bs{X}_i\bs{\beta})},\ i=1,..,n.
\end{align}
The matrix $\tilde{W}:=W(\tilde{\bs{\beta}})$ is the diagonal matrix with diagonal elements given by:
\begin{align}
    [W(\tilde{\bs{\beta}})]_{ii} = \tilde{p}_i(1-\tilde{p}_i) = \frac{\ee(\bs{X}_i\tilde{\bs{\beta}})}{(1+\ee(\bs{X}_i\tilde{\bs{\beta}}))^2}.
\end{align}

\subsubsection{Cox survival regression}
In Cox survival regression, the outcome $y_i=(t_i,d_i)$ denotes at which time $t_i$ an event occurred, $d_i=1$, or was censored, $d_i=0$. Details for Cox survival regression are given in for example \citep{meijer2013efficient}. The hazard function $h_i(t)$ is proportional to a baseline hazard $h_0(t)$ with cumulative hazard $H_0(t)$:
\begin{align}
    h_i(t)=h_0(t)\ee (\bs{X}_i\bs{\beta}),\ i=1,..,n,\ H_0(t)=\int_{s=0}^t h_0(s)\dd s.
\end{align}
Similar to as mentioned in \citep{meijer2013efficient}, the vector $\bs{y}-E_{\bs{y}|\bs{\beta}}[\bs{y}]$ in Equation \eqref{eq:meanbetatilde} is replaced by the vector of martingale residuals:
\begin{align}
    \Delta_i := d_i - H_0(t_i)\ee(\bs{X}_i\tilde{\bs{\beta}}),\ i=1,..,n.
\end{align}
The $W$ matrix (denoted by $D$ in \citep{meijer2013efficient}) is given by the following diagonal matrix:
\begin{align}
    \left[W(\tilde{\bs{\beta}})\right]_{ii} &:= H_0(t_i)\ee(\bs{X}_i\tilde{\bs{\beta}}),\ i=1,..,n.
\end{align}
We use the well-known Breslow estimator to estimate $H_0$, which is based on the times of observed events, i.e. $t_i$ for which $d_i=1$:
\begin{align}
    \hat{H}_0(t)=\sum_{i:\ t_i\leq t}\hat{h}_0(t_i),\ \hat{h}_0(t_i)=d_i\left(\sum_{j:\ t_j\geq t_i} \ee(\bs{X}_j\tilde{\bs{\beta}})\right)^{-1}.
\end{align}

\subsection{Hypershrinkage ridge penalty}\label{ap:hypershrinkage}
Consider the prior model for the regression coefficients for one co-data set matrix $Z$:
\begin{align}
    \beta_k\overset{ind.}{\sim} N\left(0,\tau_{global}^2\bs{Z}_k\bs{\gamma}\right).
\end{align}
The goal is to shrink the group parameter estimates $\bs{\gamma}$ in such a way that if the co-data is not informative, we shrink towards the ordinary ridge prior as a target prior distribution, i.e. all local variances are set to $1$. Furthermore, the variance of the local variance estimates should then be the same for all $p$ covariates and should not depend on the co-data matrix $Z$. These two assumptions can be expressed as follows:
\begin{align}
    E(\bs{\tau}^2_{local})=E(Z\bs{\gamma})=\mathds{1}_{p\times 1},\ \Var(\bs{\tau}^2_{local})=\Var(Z\bs{\gamma}) = \sigma_\gamma^2I_{p\times p},
\end{align}
for some variance $\sigma^2_\gamma\geq 0$.
Rewriting the expression above gives expressions for the mean and variance of $\bs{\gamma}$:
\begin{align}
    E(\bs{\gamma})&=E((Z^TZ)^{-1}Z^TZ\bs{\gamma})=(Z^TZ)^{-1}Z^T\mathds{1}_{p\times 1}:=(Z^TZ)^{-1}Z^TZ\mathds{1}_{G\times 1}=\mathds{1}_{G\times 1},\\ 
    \Var(\bs{\gamma}) &= \Var((Z^TZ)^{-1}Z^TZ\bs{\gamma}) = \sigma_\gamma^2(Z^TZ)^{-1}Z^TZ (Z^TZ)^{-1} = \sigma_\gamma^2(Z^TZ)^{-1}.\label{eq:priorhypervar}
\end{align}
For disjunct groups, this latter expression reduces to
\begin{align}
    \Var(\bs{\gamma}) &= \sigma_\gamma^2 \left[\begin{array}{ccc}
        |\mc{G}_1| & & \emptyset  \\
         & \ddots & \\
         \emptyset & & |\mc{G}_G|
    \end{array}\right]^{-1} := \sigma_\gamma^2 W_\gamma^{-1}.
\end{align}
We rescale $\bs{\gamma}$ such that all variances are on the same scale:
\begin{align}\label{eq:priorgamma'}
    \bs{\gamma}'&= W_\gamma^{1/2}\bs{\gamma},\ E(\bs{\gamma}')=W_\gamma^{1/2}\mathds{1}_{G\times 1},\ \Var(\bs{\gamma}')=\sigma_\gamma^2I_{G\times G}.
\end{align}
We use a ridge penalty for $\bs{\gamma}'$ corresponding to the normal distribution with mean and variance given above, with hyperpenalty $\lambda_\gamma$ inversely proportional to the variance $\sigma_\gamma^2$.
Finally, given an estimate $\hat{\lambda}_\gamma$ we solve the optimisation problem given in Equation \eqref{eq:tauest} for the rescaled $\bs{\gamma}'$ and scale back to obtain the parameter estimates for $\bs{\gamma}$:
\begin{align}
\begin{split}
    W_\gamma^{1/2}\tilde{\bs{\gamma}} = \tilde{\bs{\gamma}}' &= \argmin{\bs{\gamma}'}\left\{ ||AW_\gamma^{-1/2} \bs{\gamma}' -\bs{b}||^2_2 + \hat{\lambda}_{\gamma}\sum_{g=1}^G \left(\gamma'_g-\left[W_\gamma^{1/2}\right]_{gg}\right)^2 \right\}.
\end{split}
\end{align}


\subsection{Covariate selection for prediction}\label{ap:posthoc}
Below we give the technical details needed for implementation of the options for post-hoc variable selection using the approaches described in \citep{novianti2017better,carvalho2009handling,bondell2012consistent}, using an elastic net penalty, DSS criterion and marginal penalised credible intervals respectively. 

\subsubsection{Using elastic net}
As is widely known, the lasso penalty is known to be able to automatically select variables, but is not stable when covariates are correlated. The elastic net penalty, a combination of the ridge and lasso penalty, can be seen as a stabilised lasso, in the sense that the added ridge penalty stabilises the covariate selection. In a similar manner, the elastic net penalty is used in \citep{novianti2017better}, by rescaling the covariates with the weighted ridge penalty and adding a lasso penalty to perform selection.
The procedure can be summarised as follows.

First rescale $X$ and $\bs{\beta}$ to $X'$ and $\bs{\beta}'$:
\begin{align}
\Delta :=\left[\begin{array}{ccc}
    \frac{1}{\hat{\tau}_{1,local}^2} &  & \emptyset\\
     & \ddots & \\
     \emptyset &  & \frac{1}{\hat{\tau}^2_{p,local}}
\end{array}\right],\
X':= X\Delta^{-\frac{1}{2}},\   \bs{\beta}':= \Delta^{\frac{1}{2}}\bs{\beta}.
\end{align}
Note that $X'\bs{\beta}'=X\bs{\beta}$, and $\beta'_k\sim N(0,\hat{\tau}_{global}^2)$, $k=1,..,p$.
Then find the penalised maximum likelihood estimate for $\bs{\beta}'$ such that the desired number of covariates $s$ is selected:
\begin{align}
\begin{split}
    \hat{\bs{\beta}}'&=\argmax{\bs{\beta}'}\left\{\log \pi\left(Y|X',\bs{\beta}'\right) + \frac{1}{\hat{\tau}^2_{global}}||\bs{\beta}'||^2_2 + \lambda_1 ||\bs{\beta}'||_1\right\},\\
    \lambda_1&\in\{\lambda_1\in\mathbb{R}:\ |\{k:\beta_k'\neq 0\}|=s\}.
\end{split}
\end{align}
Define $\mc{I}_s=\{k\in\{1,..,p\}:\hat{\beta}_k'\neq 0\}$ as the set of indices of selected covariates. Denote by $\bs{\beta}_s$ the regression coefficients of the selected covariates and $\bs{\beta}_{-s}$ the remaining regression coefficients. 
Lastly, refit the selected covariates to obtain the sparsified predictor $\hat{\bs{\beta}}_{sp.}$ on the right scale.
\begin{align}
    \hat{\bs{\beta}}_{sp.}&=\argmax{\bs{\beta}:\ \bs{\beta}_{-s}=\bs{0}}\left\{\log \pi\left(Y|X,\bs{\beta}\right) + \frac{1}{\hat{\tau}^2_{global}}\sum_{k\in\mc{I}_s}\frac{1}{\hat{\tau}^2_{k,local}}\beta_k^2 \right\}.
\end{align}
We propose to use either the previous weighted ridge estimates for $\hat{\tau}_{global}$ and $\hat{\bs{\tau}}_{local}$ to prevent overestimating in dense models, or set the local weights to $1$ and refit $\hat{\tau}_{global}$ using maximum marginal likelihood or cross-validation to undo overshrinkage in sparse models.

\subsubsection{Using DSS}
\cite{hahn2015decoupling} propose to decouple shrinkage and selection (DSS). Decoupling here means that inference is done first using any prior, and selection is done afterwards based on the posterior, resulting in a sequence of sparse linear models. The posterior summary variable selection approach they propose is based on a loss function which balances the prediction error and sparseness of the point estimate of the regression coefficients $\bs{\beta}$. Given the posterior mean $\hat{\bs{\beta}}$, they first propose to use the following sparsified point estimate $\hat{\bs{\beta}}_{sp.}$:
\begin{align}
    \hat{\bs{\beta}}_{sp.}= \argmin{\bs{\gamma}} \lambda||\bs{\gamma}||_0 + \frac{1}{n}||X\hat{\bs{\beta}}-X\bs{\gamma}||_2^2.
\end{align}
As the optimisation problem corresponding to the $L_0$-penalty is intractable, they propose to approximate the loss function by a local linear approximation with a weighted $L_1$-penalty:
\begin{align}
    \hat{\bs{\beta}}_{sp.}= \argmin{\bs{\gamma}} \sum_j\frac{\lambda}{|w_j|}|\bs{\gamma}| + \frac{1}{n}||X\hat{\bs{\beta}}-X\bs{\gamma}||_2^2,
\end{align}
where they use $w_j=\hat{\beta}_j$. This optimisation problem can be solved with existing software like \texttt{glmnet}.

\subsubsection{Using marginal penalised credible regions}
\cite{bondell2012consistent} show that variable selection can be done consistently via penalised credible regions. They prove that their proposed approach using marginal posterior credible sets is consistent in variable selection even when $p$ grows exponentially fast relative to the sample size, useful for high-dimensional data where $p\gg n$.

They propose to use the following set $A_n$ of selected variables based on a thresholding selection rule:
\begin{align}
    A_n&=\{j:\ |\beta_j|> t_{n,j}\},
\end{align}
where the threshold $t_{n,j}$ determines the size of $A_n$, or equivalently, the number of variables that is selected. They propose to use the following threshold:
\begin{align}\label{eq:margthreshold}
    t_{n,j}=s_jt_n,\ s_j=\frac{\sqrt{\Var_{\bs{\beta}|\bs{Y}}(\beta_j)}}{\min_j \sqrt{\Var_{\bs{\beta}|\bs{Y}}(\beta_j)}}.
\end{align}
Note that whereas the selection procedure is done marginally, the threshold depends on the full posterior.

We approximate the marginal posterior standard deviation in Equation \eqref{eq:margthreshold} for GLMs penalised with a weighted ridge prior, using a Laplace approximation around the posterior mode $\hat{\bs{\beta}}$.

\textbf{Result.}
Consider a GLM with diagonal weight matrix $W=Var_{\bs{Y}|\bs{\beta}}(\bs{Y})$, that is penalised by a weighted ridge penalty, denoted by the diagonal penalty matrix $\Delta$ and corresponding prior variance $\bs{\tau}^2_{global}$. Define $\tilde{X}=W^{1/2}X\Delta^{-1/2}$ and denote the SVD of $\tilde{X}$ as $\tilde{X}=UDV^T$. The posterior standard deviation of $\beta_j$ can be approximated by: 
\begin{align}
    \sqrt{\Var_{\bs{\beta}|\bs{Y},\bs{\tau}}(\beta_j)} &\approx \Delta_{jj}^{-1/2}\sqrt{1-[VD^2(D^2+I)^{-1}V^T]_{jj}}.
\end{align}
For the linear regression case, this approximation is in fact an equality.

\textbf{Derivation}.
Denote the maximum penalised likelihood estimate, and equivalently the posterior mode, by $\hat{\bs{\beta}}$. We can approximate the posterior by a Laplace approximation using a Taylor expansion of the log posterior around the mode. The Taylor expansion is given by:
\begin{align*}
    \log\pi(\bs{\beta}|\bs{y},\bs{\tau})&\approx \log\pi(\hat{\bs{\beta}}|\bs{y},\bs{\tau}) + (\bs{\beta}-\hat{\bs{\beta}})^T\nabla_\beta\log\pi(\hat{\bs{\beta}}|\bs{y},\bs{\tau}) \\
    &\qquad + \frac{1}{2}(\bs{\beta}-\hat{\bs{\beta}})^T\nabla_\beta^2\log\pi(\hat{\bs{\beta}}|\bs{y},\bs{\tau})(\bs{\beta}-\hat{\bs{\beta}})\\
    &=\log\pi(\hat{\bs{\beta}}|\bs{y},\bs{\tau}) + \frac{1}{2}(\bs{\beta}-\hat{\bs{\beta}})^T\nabla_\beta^2\log\pi(\hat{\bs{\beta}}|\bs{y},\bs{\tau})(\bs{\beta}-\hat{\bs{\beta}}),
\end{align*}
where the approximation is in fact an equality when linear regression is considered.
Taking the exponential on both sides leads to:
\begin{align*}
    \pi(\bs{\beta}|\bs{y},\bs{\tau})\overset{\cdot}{\propto} \exp\left(\frac{-1}{2}(\bs{\beta}-\hat{\bs{\beta}})^T\left[-\nabla_\beta^2\log\pi(\hat{\bs{\beta}}|\bs{y},\bs{\tau})\right](\bs{\beta}-\hat{\bs{\beta}})\right),
\end{align*}
where we use $\overset{\cdot}{\propto}$ to denote ``approximately proportional to''. So we can approximate the posterior with the following multivariate gaussian:
\begin{align*}
    \bs{\beta}|\bs{y},\bs{\tau}\overset{\cdot}{\sim} N\left(\hat{\bs{\beta}}, \left[-\nabla_\beta^2\log\pi(\hat{\bs{\beta}}|\bs{y},\bs{\tau})\right]^{-1}\right),
\end{align*}
where we use $\overset{\cdot}{\sim}$ to denote ``approximately distributed as". The posterior covariance matrix for a GLM is approximated by:
\begin{align*}
    \Cov_{\bs{\beta}|\bs{Y},\bs{\tau}}(\bs{\beta})\approx\left[-\nabla_\beta^2\log\pi(\hat{\bs{\beta}}|\bs{y},\bs{\tau})\right]^{-1} = \left[X^TW(\hat{\bs{\beta}})X + \Delta\right]^{-1}.
\end{align*}
which in turn we can write as, using Woodbury's matrix inversion identity, substituting $\tilde{X}=W^{-1/2}X\Delta^{-1/2}$ and the SVD of $\tilde{X}$:
\begin{align*}
    \left[X^TWX + \Delta\right]^{-1} &= \bs{\Delta}^{-1} - \bs{\Delta}^{-1}X^TW^{1/2}\left(I_{n\times n}+W^{1/2}X\bs{\Delta}^{-1}X^TW^{1/2}\right)^{-1}W^{1/2}X\bs{\Delta}^{-1}\\
    &=\bs{\Delta}^{-1} - \bs{\Delta}^{-1/2}\tilde{X}^T\left(I_{n\times n}+\tilde{X}\tilde{X}^T\right)^{-1}\tilde{X}\bs{\Delta}^{-1/2}\\
    &=\bs{\Delta}^{-1} - \bs{\Delta}^{-1/2}VDU^T\left(I_{n\times n}+UDV^TVDU^T\right)^{-1}UDV^T\bs{\Delta}^{-1/2}\\
    &=\bs{\Delta}^{-1} - \bs{\Delta}^{-1/2}VD^2\left(I_{n\times n}+D^2\right)^{-1}V^T\bs{\Delta}^{-1/2}.
\end{align*}
The marginal posterior standard deviations are given by the square root of the diagonal elements:
\begin{align*}
    \sqrt{\Var_{\bs{\beta}|\bs{Y},\bs{\tau}}(\beta_j)} &\approx \Delta_{jj}^{-1/2}\sqrt{1-[VD^2(D^2+I)^{-1}V^T]_{jj}}.
\end{align*}

\subsection{Unpenalised covariates}\label{ap:unpenalised}
We can group covariates that we do not want to penalise (e.g. an intercept) in a group, say group $\mathcal{G}_0$. Not penalising corresponds to a Bayesian prior with mean $\mu^\beta_0=0$ and $\tau^2_0=\infty$, and penalty $0$. 
Furthermore, for the matrix $C$ as defined in Equation \eqref{eq:Cv}, $[C]_{kl}=0$ for every $l\in\mathcal{G}_0$, $k\neq l$:
\begin{lemma}
Let $l\in\mathcal{G}_0$ be an unpenalised covariate without correlation with other covariates. Then, for $k\neq l$:
\begin{align}
 [C]_{kl}=\left[(X^T\tilde{W}X+\tilde{\Omega})^{-1}X^T\tilde{W}X\right]_{kl}=0,
\end{align}
and therefore also $[C]_{lk}=[C]_{kl}=0$.
\end{lemma}
\begin{proof}
First, note that the matrix $C$ is equal to:
\begin{align*}
    C&=(X^T\tilde{W}X+\tilde{\Omega})^{-1}X^T\tilde{W}X=(X^T\tilde{W}X+\tilde{\Omega})^{-1}(X^T\tilde{W}X + \tilde{\Omega} - \tilde{\Omega}) \\
    &= I-(X^T\tilde{W}X+\tilde{\Omega})^{-1}\tilde{\Omega}.
\end{align*}
So, for $k\neq l$:
\begin{align*}
    [C]_{kl} &= -\left[(X^T\tilde{W}X+\tilde{\Omega})^{-1}\tilde{\Omega}\right]_{kl}\\
    &= -\sum_{i=1}^p\left[(X^T\tilde{W}X+\tilde{\Omega})^{-1}\right]_{ki}\left[\tilde{\Omega}\right]_{il}\\
    &= 0,
\end{align*}
where the latter equation holds since the $l^{\text{th}}$ column of the precision matrix corresponding to an unpenalised variable contains only $0$.
Note that $C$ is symmetric since it is a product of sums of symmetrix matrices. Therefore we can conclude that $[C]_{lk}=[C]_{kl}=0$.
\end{proof}

As a result from this lemma and Equations \eqref{eq:linsysmu2},\eqref{eq:linsystau2}, we see that:
\begin{align}
    [A_\mu]_{g0}=[A_\mu]_{0g}=0,\ [A_\tau]_{g0}=[A_\tau]_{0g}=0,\ \forall g=1,..,G.
\end{align}
Therefore we can compute the moment estimates using the block matrix of $A_\mu$ and $A_\tau$ corresponding to the penalised groups only. 
So, after we have computed $C$ using both penalised and unpenalised covariates, we only need the rows and columns of $C$ corresponding to penalised covariates to obtain the moment estimates.

\section{Data applications}\label{ap:applications}
\subsection{Predicting therapy response in colorectal cancer}\label{ap:appmiRNA}
The results of the first data application using miRNA expression are discussed in Section \ref{par:appmiRNA}. Here we provide mentioned additional figures. 
The group weights of the other groupings are shown in Figure \ref{fig:weightsmiRNArest}. 
Figure \ref{fig:weightsmiRNAabunsd} shows the performance of \texttt{ecpc} in the dense setting and covariate sparse setting when abundance and standard deviation are discretised in $5,10$ or $20$ groups. The performances are comparable, with the model based on $20$ groups in abundance and standard deviation performing slightly better in the dense setting, and slightly worse in the sparse setting. 
Figure \ref{fig:miRNAheavytails} shows the absolute values of the estimated regression coefficients for \texttt{ecpc} and \texttt{ordinary ridge}. The density plot is more heavy-tailed for \texttt{ecpc}, which facilitates posterior selection.
The performance of the group sparse models is shown in Figure \ref{fig:AUCgroupsparsemiRNA}. 
Here, \texttt{ecpc} is combined with a lasso penalty on the group level on all groups of the five groupings to obtain a group sparse model. 
Group lasso uses a latent overlapping group (LOG) penalty \citep{jacob2009group,yan2017hierarchical} on all groups of the first three co-data sources and the leaf groups in the tree of the \texttt{FDR1} and \texttt{FDR2} groupings, without distinguishing between co-data sources.
Hierarchical lasso uses a LOG penalty on all groups of all co-data sources. 
For the FDR groupings, the implied hierarchical constraints are that covariates in an FDR group can be included only when all covariates in the groups with lower FDRs are included as well \citep{yan2017hierarchical}.
\texttt{ecpc} adequately learns from co-data and outperforms group lasso and hierarchical lasso. 
Then, Figure \ref{fig:AUCposthoc} shows the AUC performance of various post-hoc selection methods on the cross-validation folds.
Lastly, Figure \ref{fig:AUCsubsamplesmiRNA} shows the AUC performance of \texttt{ecpc} and \texttt{elastic net} with $\alpha=0.3$ and $\alpha=0.8$ on the test sets corresponding to the 50 subsamples used for Figure \ref{fig:overlapmiRNA} in Section \ref{par:appmiRNA} to assess covariate selection stability.

\begin{figure}
    \centering
    \begin{subfigure}[c]{0.45\textwidth}
    \centering
    \includegraphics[width=\linewidth]{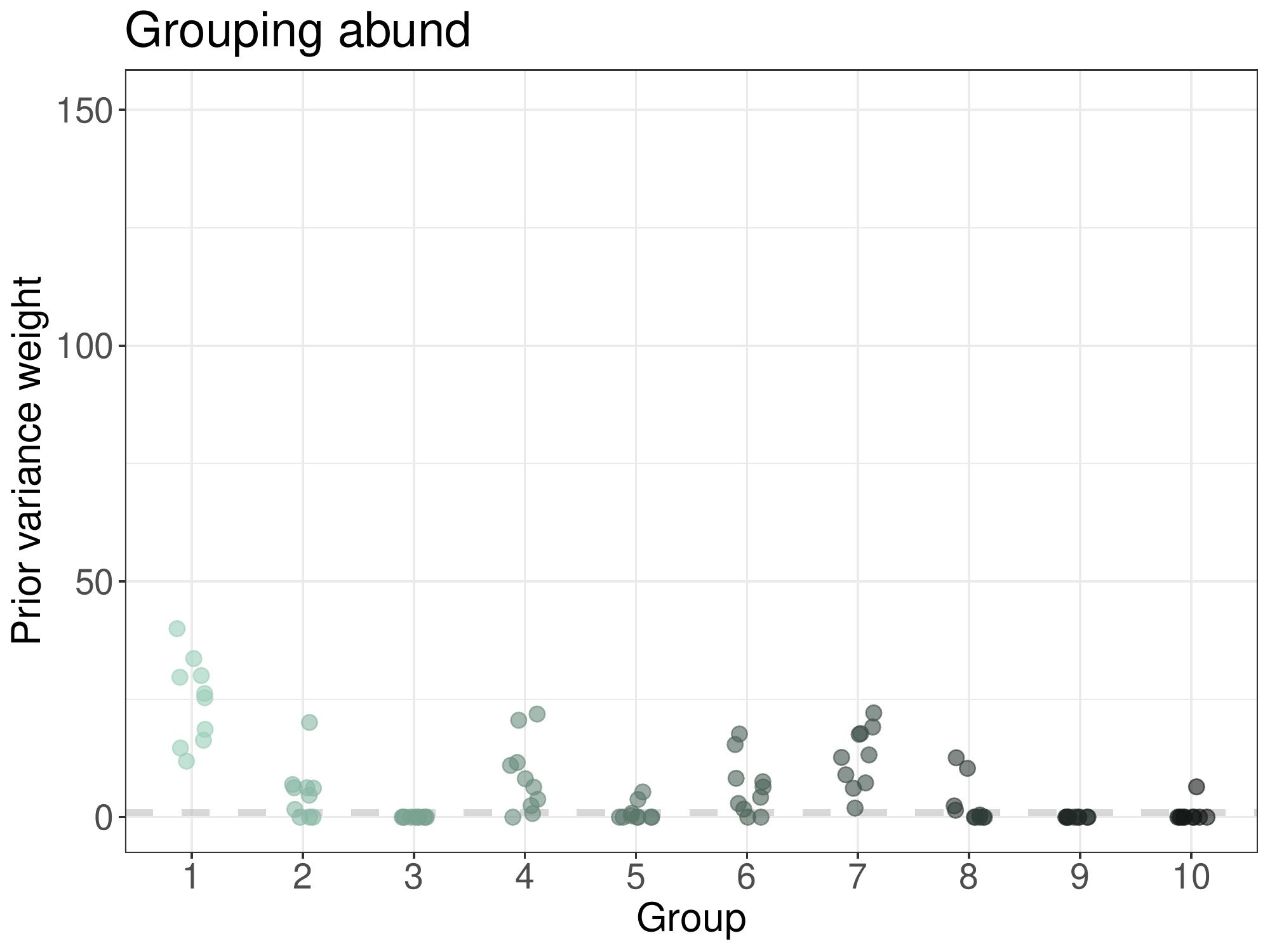}
    \end{subfigure}
    \begin{subfigure}[c]{0.45\textwidth}
    \centering
    \includegraphics[width=\linewidth]{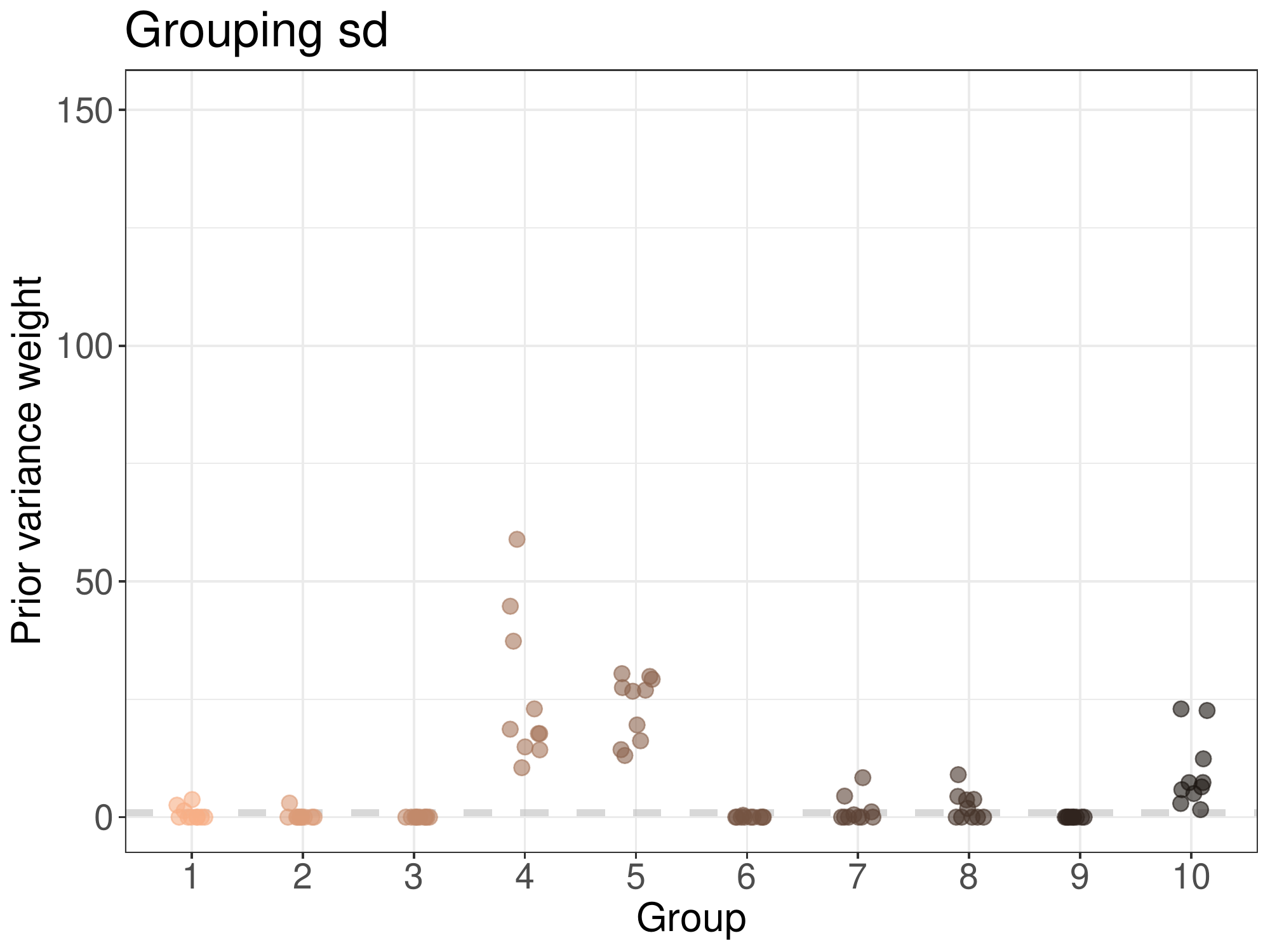}
    \end{subfigure}
        \begin{subfigure}[c]{0.45\textwidth}
    \centering
    \includegraphics[width=\linewidth]{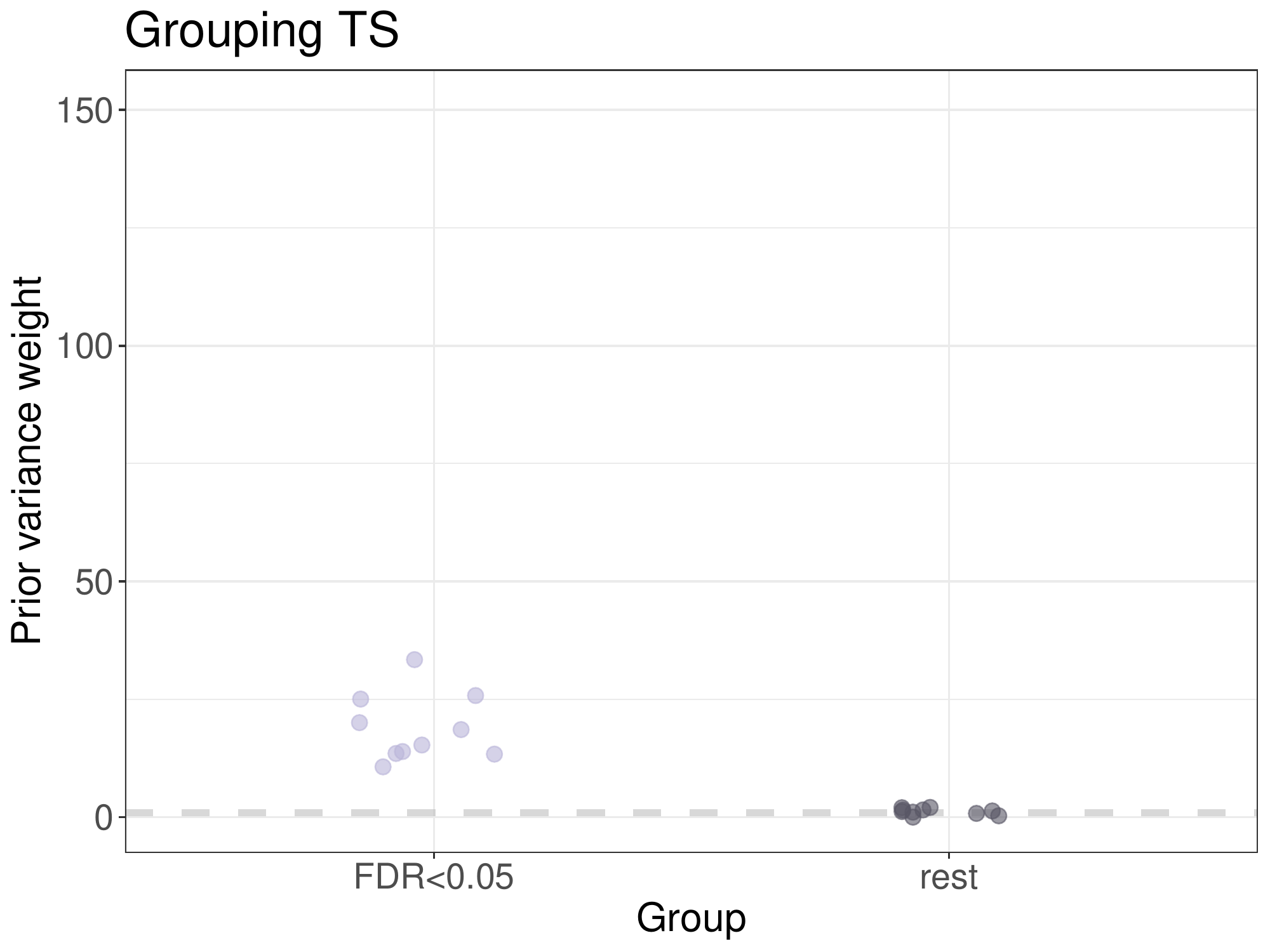}
    \end{subfigure}
    \begin{subfigure}[c]{0.45\textwidth}
    \centering
    \includegraphics[width=\linewidth]{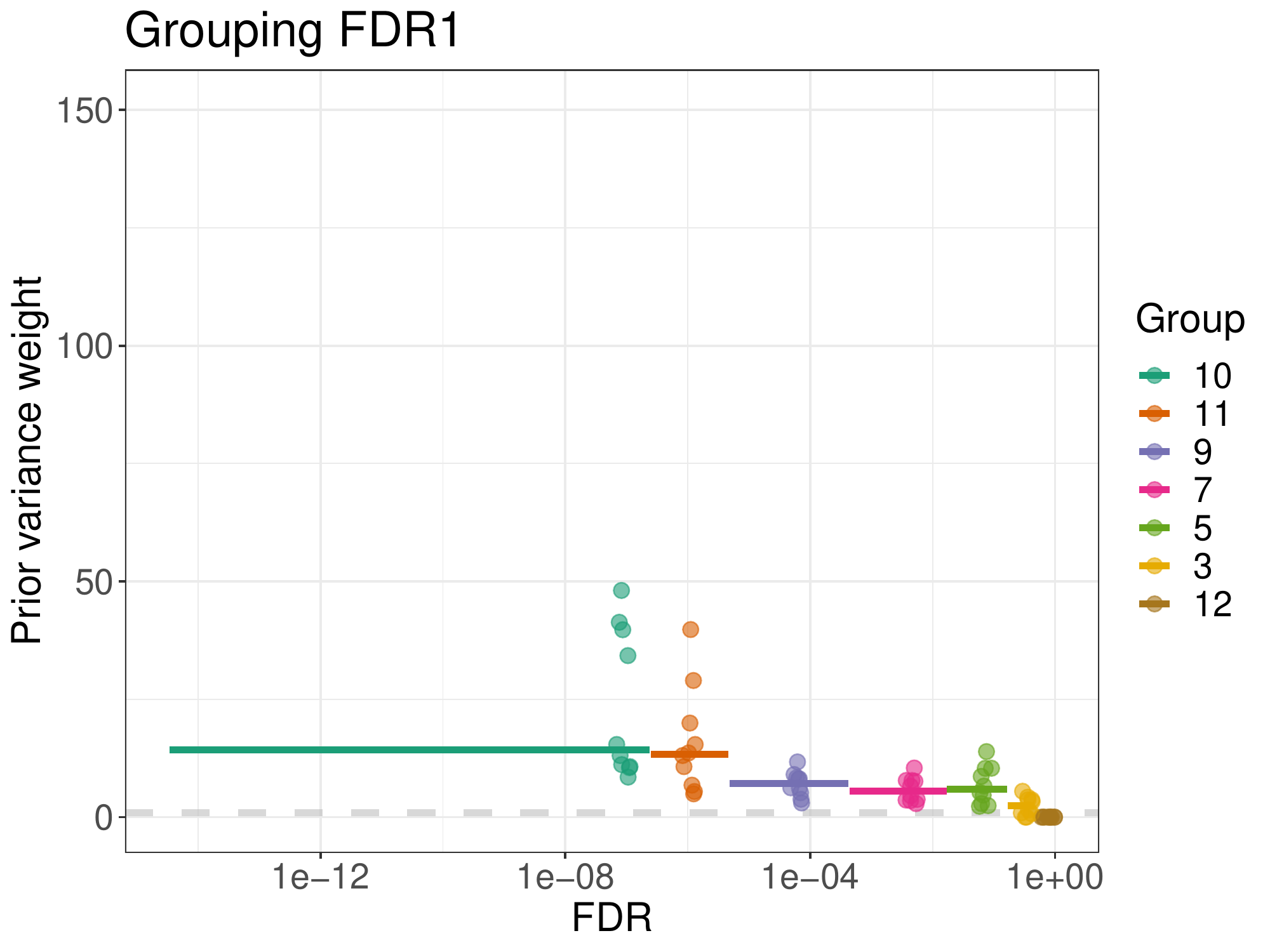}
    \end{subfigure}
    \caption{Results of 10-fold CV in miRNA data example. Estimated local variance for the first four groupings.}
    \label{fig:weightsmiRNArest}
\end{figure}

\begin{figure}
    \centering
    \begin{subfigure}[c]{0.45\textwidth}
    \centering
    \includegraphics[width=\linewidth]{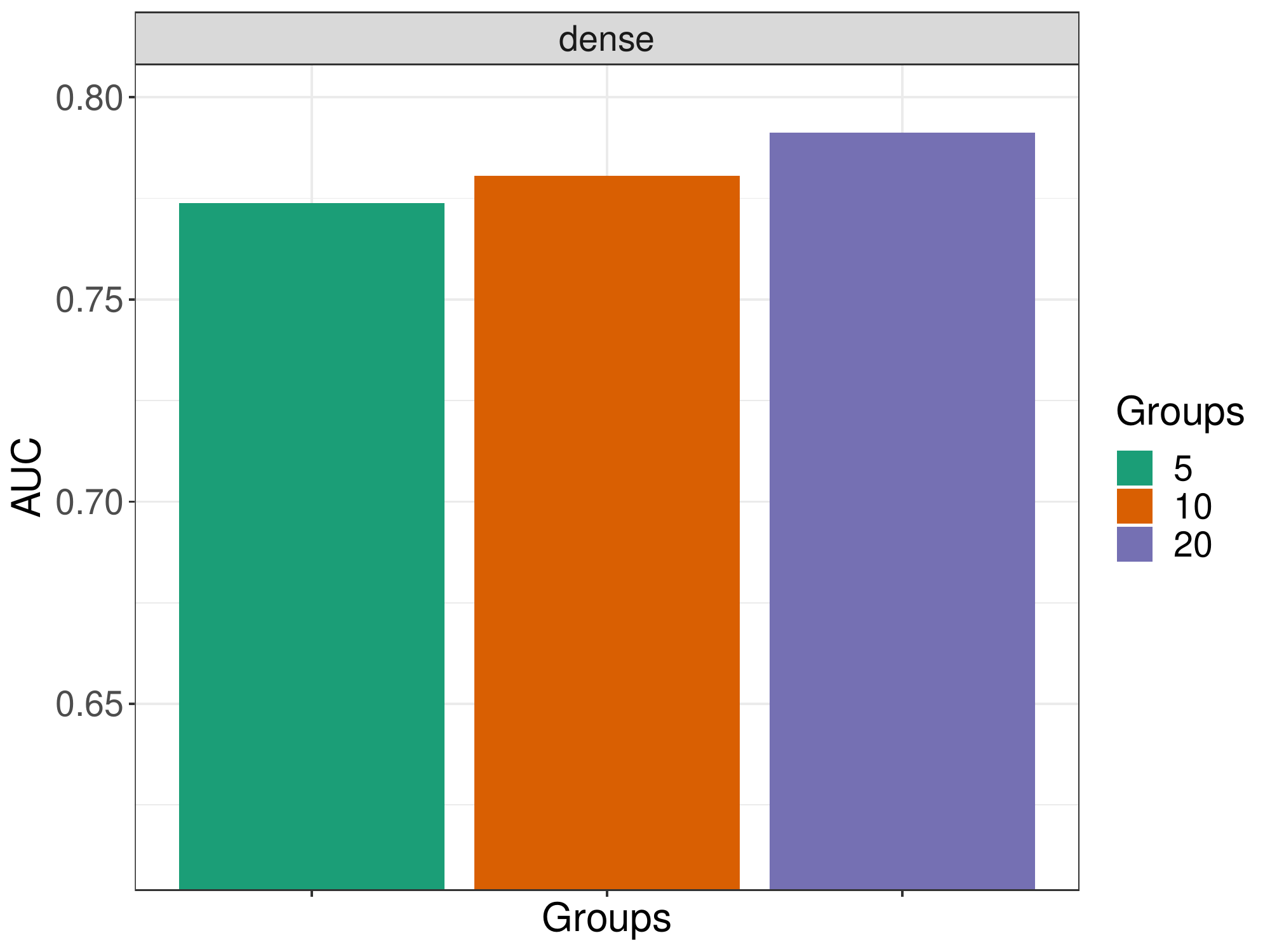}
    \end{subfigure}
    \begin{subfigure}[c]{0.45\textwidth}
    \centering
    \includegraphics[width=\linewidth]{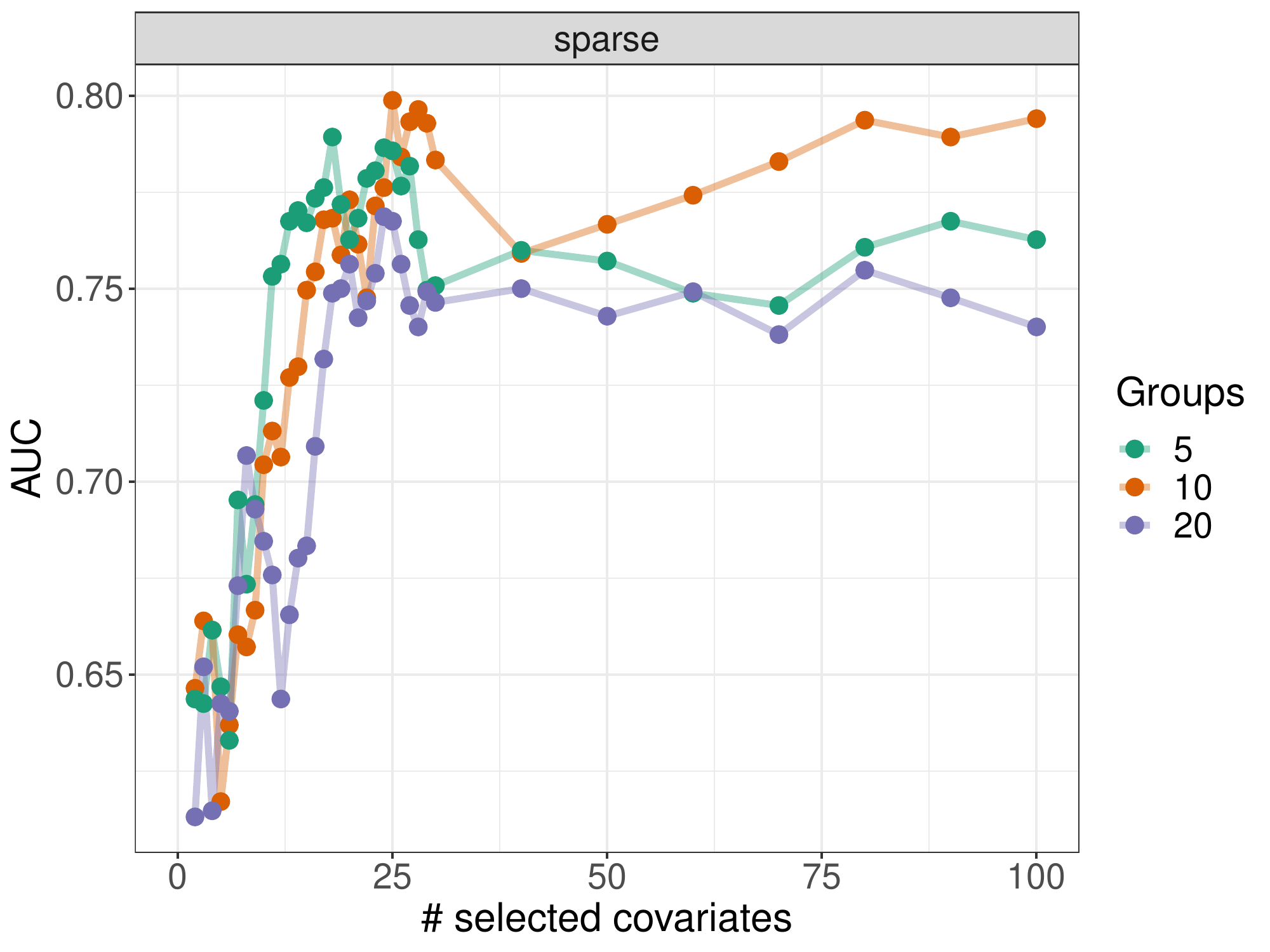}
    \end{subfigure}
    \caption{AUC performance in 10-fold CV as in miRNA data example when abundance and standard deviation are discretised in 5, 10 or 20 groups.}
    \label{fig:weightsmiRNAabunsd}
\end{figure}


\begin{figure}
    \centering
    \begin{subfigure}[c]{0.45\textwidth}
    \centering
    \includegraphics[width=\linewidth]{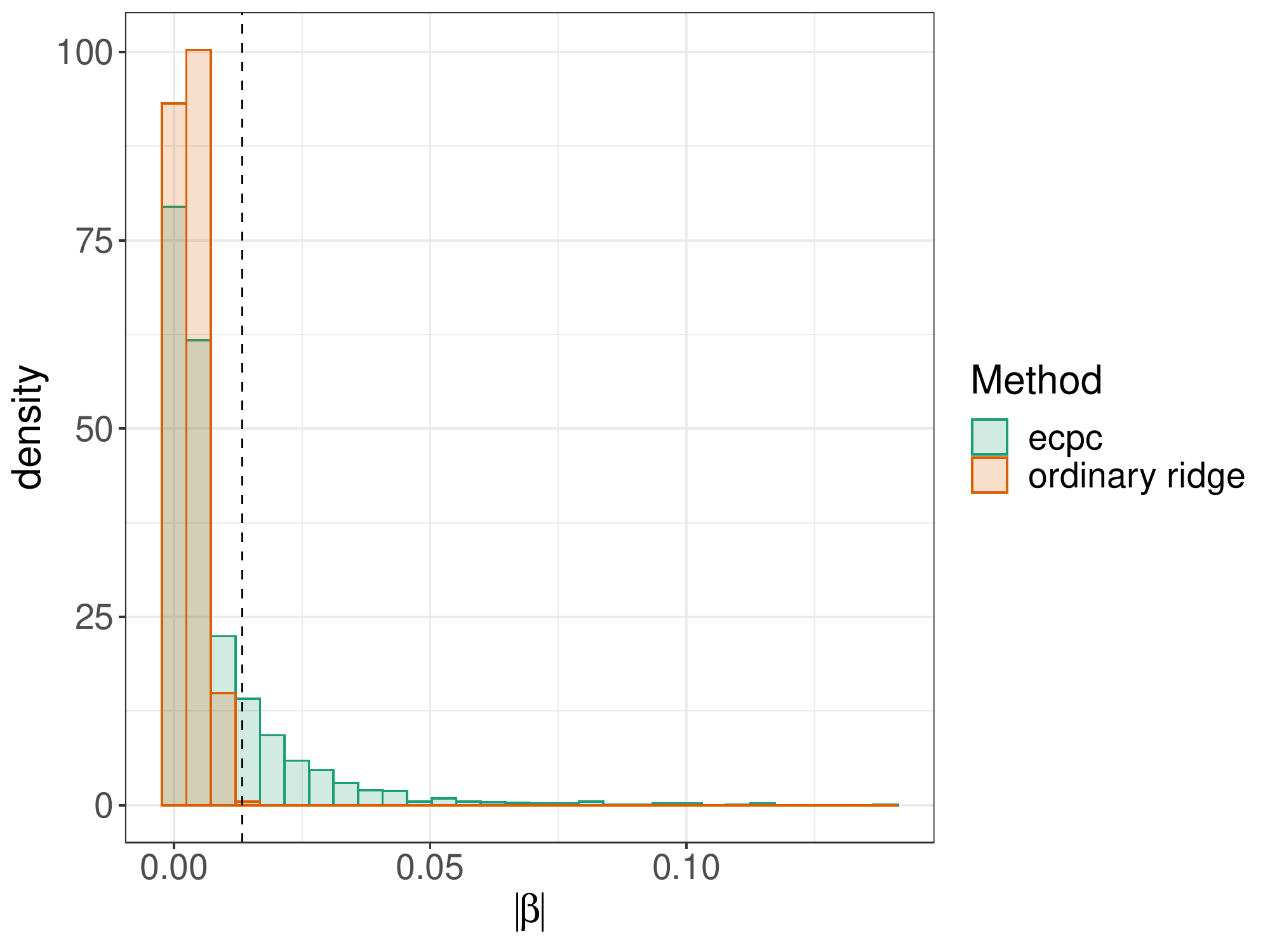}
    \end{subfigure}
    \begin{subfigure}[c]{0.45\textwidth}
    \centering
    \includegraphics[width=\linewidth]{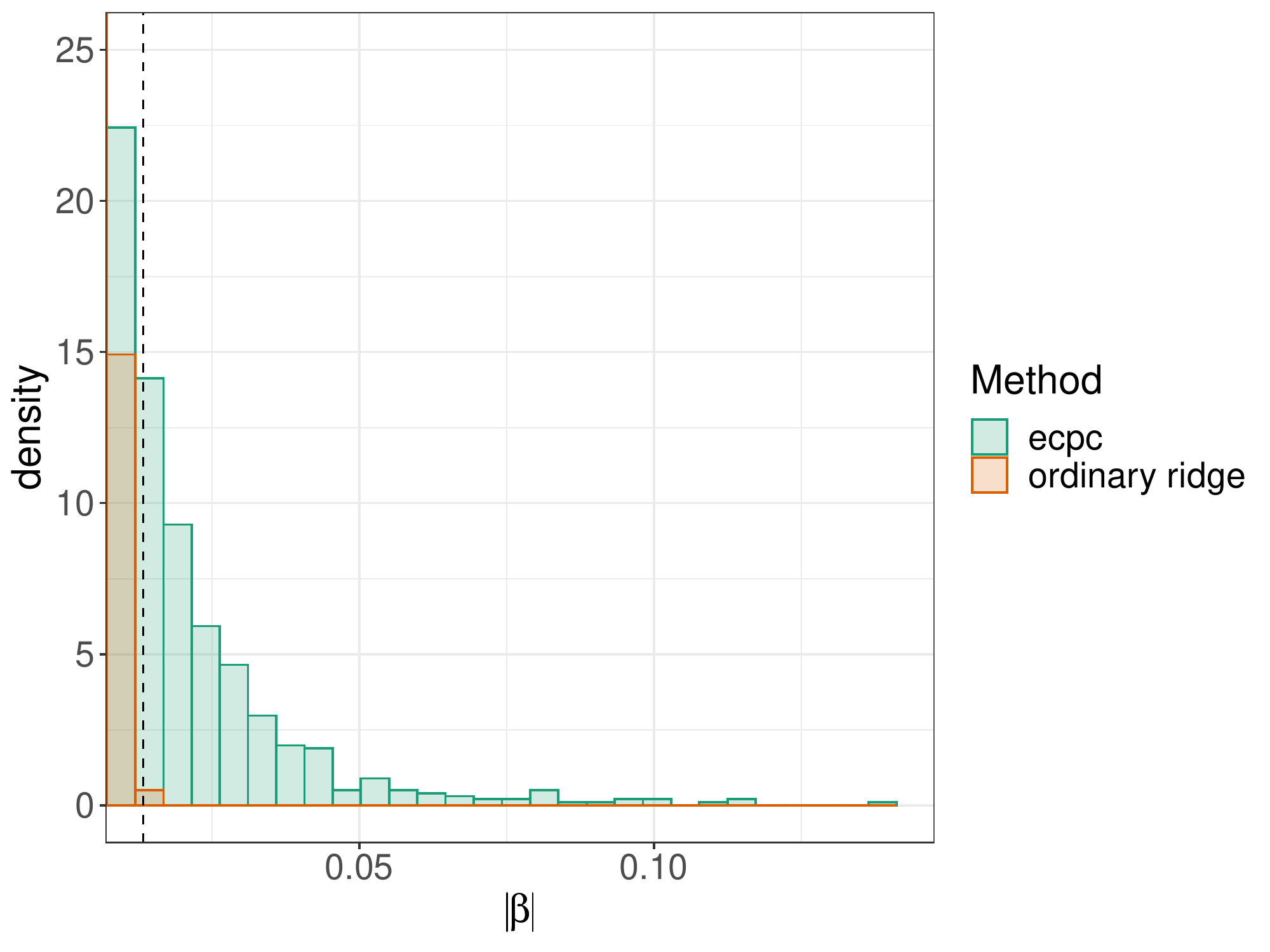}
    \end{subfigure}
    \caption{miRNA data example. Left: histogram and density plot of absolute value of estimated regression coefficients using \texttt{ecpc} or \texttt{ordinary ridge}. Right: histogram of highest 0.1 quantile of the absolute value of the regression coefficients. \texttt{ecpc} results in more heavy-tailed distributed estimates compared to \texttt{ordinary ridge}.}
    \label{fig:miRNAheavytails}
\end{figure}

\begin{figure}
    \centering
    \includegraphics[width=0.45\linewidth]{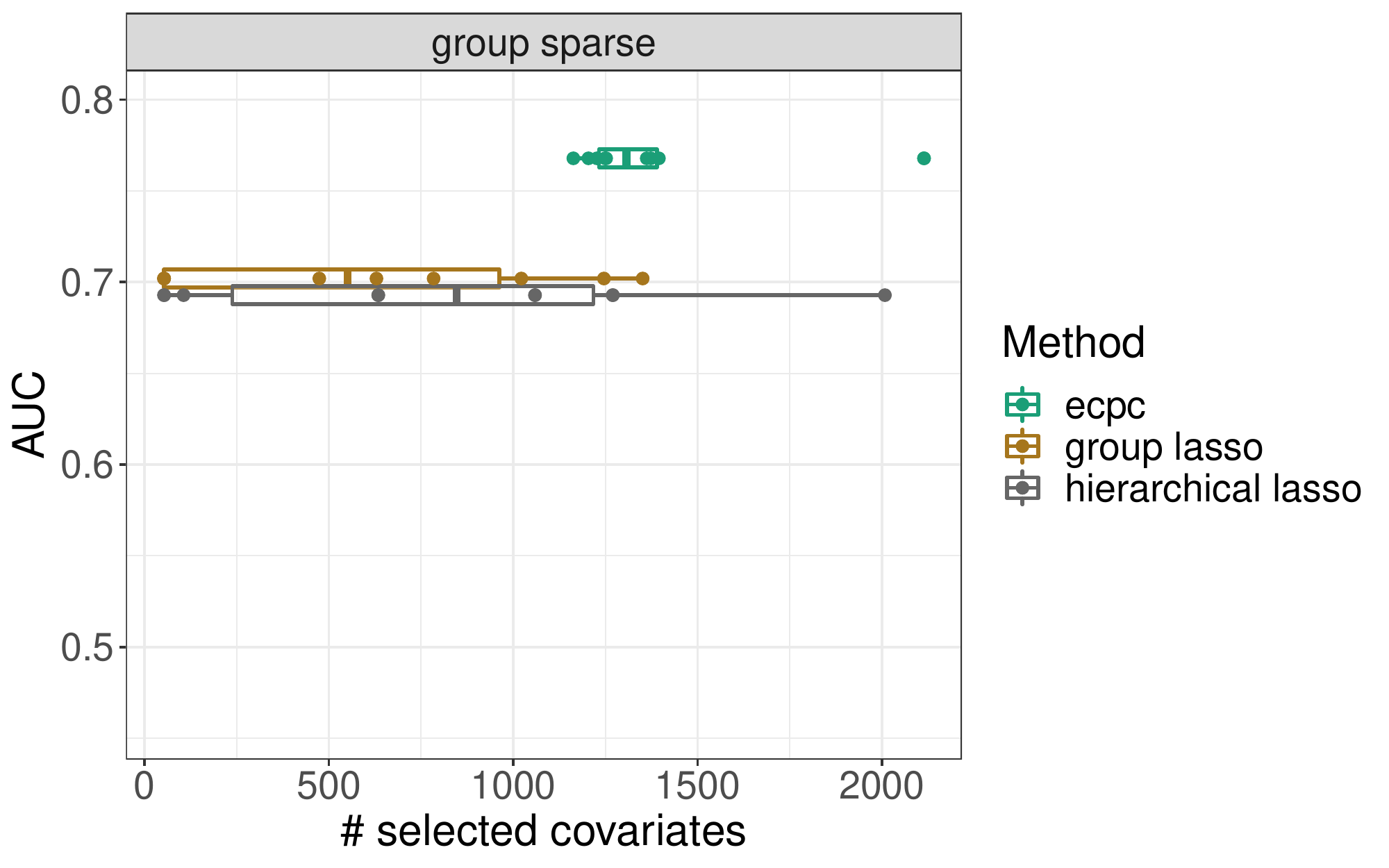}
    \caption{Results of 10-fold CV in miRNA data example. AUC in various group sparse models, with the boxplot and points illustrating the variance in selected number of variables in the folds.}
    \label{fig:AUCgroupsparsemiRNA}
\end{figure}

\begin{figure}
    \centering
    \includegraphics[width=0.55\linewidth]{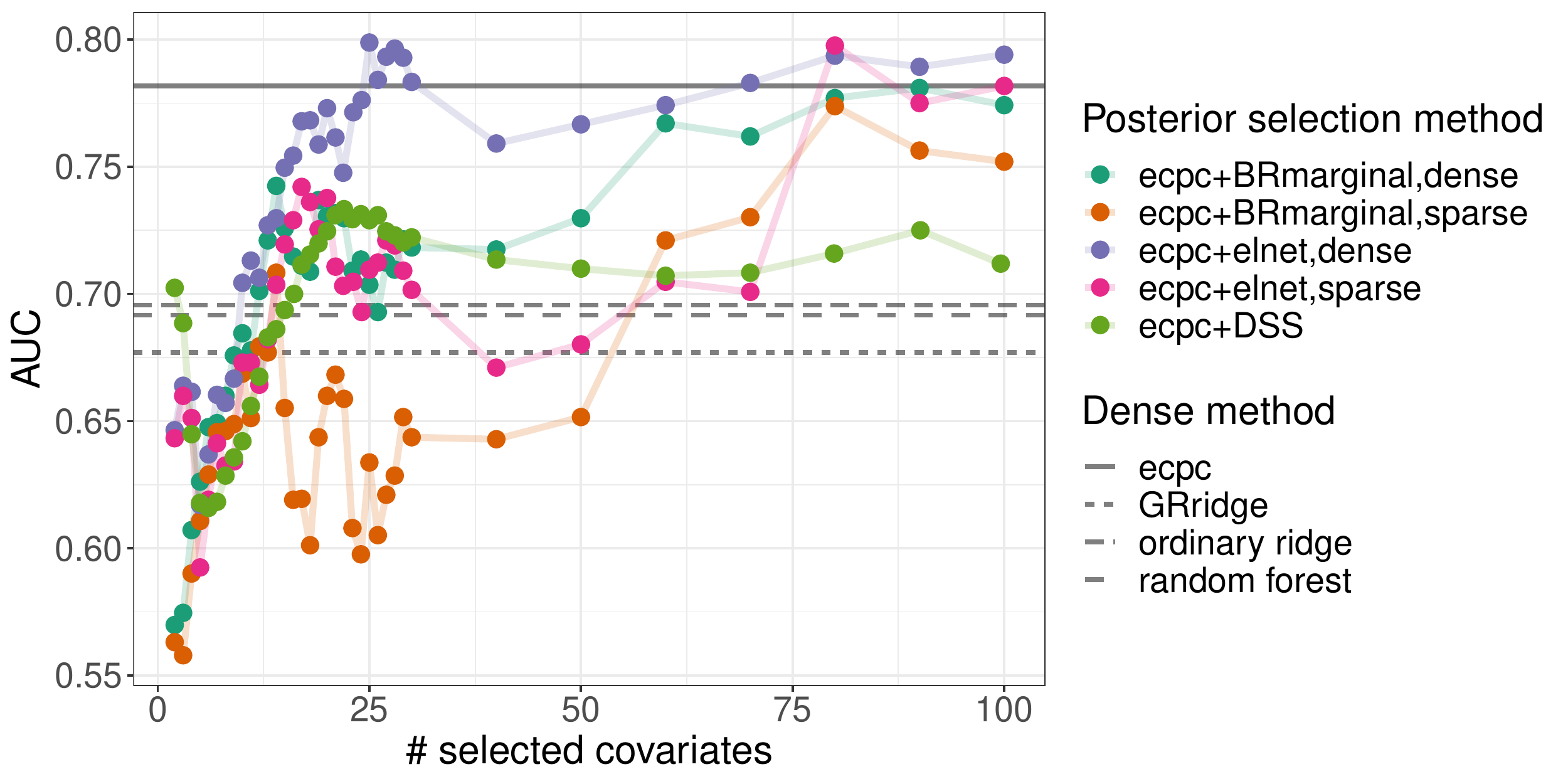}
    \caption{Results of 10-fold CV in miRNA data example. AUC for sparse models using various post-hoc selection methods.}
    \label{fig:AUCposthoc}
\end{figure}

\begin{figure}
    \centering
    \begin{subfigure}[c]{0.45\textwidth}
    \centering
    \includegraphics[width=\linewidth]{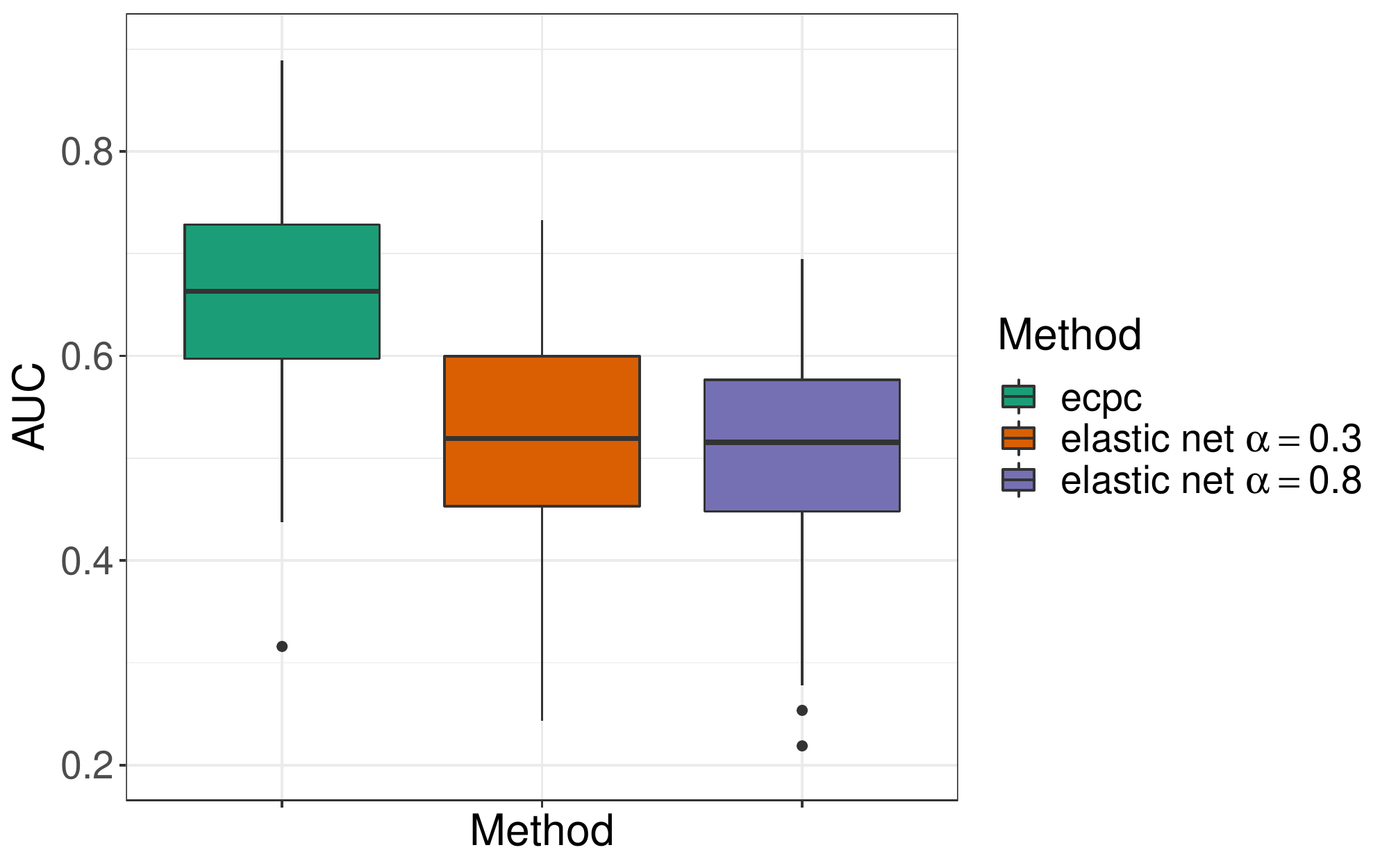}
    \end{subfigure}
    \begin{subfigure}[c]{0.45\textwidth}
    \centering
    \includegraphics[width=\linewidth]{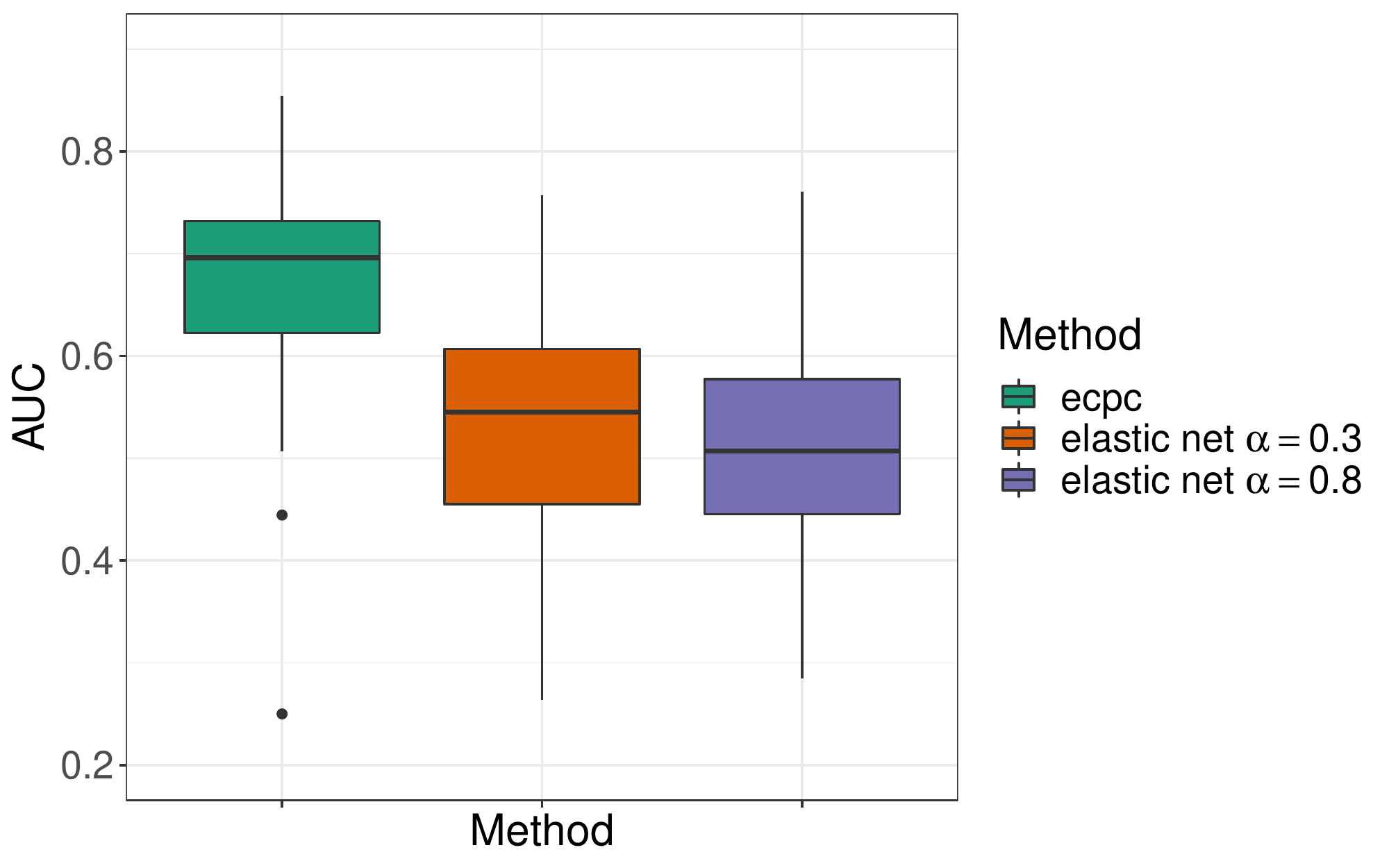}
    \end{subfigure}
    \caption{Results based on 50 stratified subsamples and corresponding test sets in miRNA data example. Boxplot of the AUC performance of \texttt{ecpc}, \texttt{elastic net} with $\alpha=0.3$ and $\alpha=0.8$ on the test set based on selections of 25 covariates (left) or 50 covariates (right) in each subsample.}
    \label{fig:AUCsubsamplesmiRNA}
\end{figure}

\subsection{Classifying cervical cancer stage}\label{ap:appVerlaat}
The results of the second data application using methylation data are discussed in Section \ref{par:appVerlaat}. Here we provide all additional figures. 
Figure \ref{fig:codataVerlaat} illustrates the used co-data groupings. 
Figure \ref{fig:weightsVerlaat} shows the estimated grouping weights and group weights in the 20 folds of the cross-validation. 
Figure \ref{fig:Verlaatheavytails} shows the absolute values of the estimated regression coefficients for \texttt{ecpc} and \texttt{ordinary ridge}. 
Figure \ref{fig:AUCVerlaatgroupsparse} shows the performance in terms of AUC for various group sparse models, and for various post-hoc selection methods in Figure \ref{fig:AUCposthocVerlaat}.
The group lasso and hierarchical lasso use a LOG penalty similar as used in the first data application described above.
Then, Figure \ref{fig:AUCsubsamplesVerlaat} shows boxplots of the AUC performance on the test sets corresponding to 50 subsamples, of \texttt{ecpc} and \texttt{elastic net} with $\alpha=0.3$ and $\alpha=0.8$ when $25$ or $50$ covariates are selected.
Lastly, Figure \ref{fig:overlapVerlaat} shows histograms for the amount of overlap in sets of selected covariates in these subsamples.
        
\begin{figure}
    \centering
    \begin{subfigure}[c]{0.45\textwidth}
    \centering
    \includegraphics[width=\linewidth]{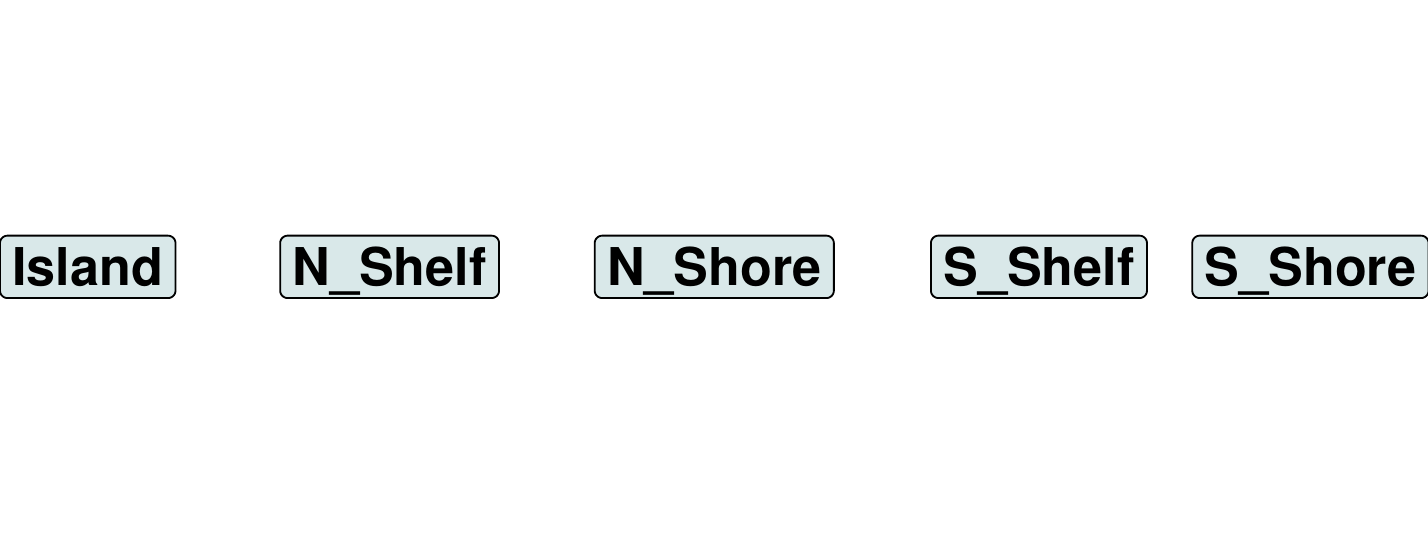}
    \end{subfigure}
    \begin{subfigure}[c]{0.45\textwidth}
    \centering
    \includegraphics[width=\linewidth]{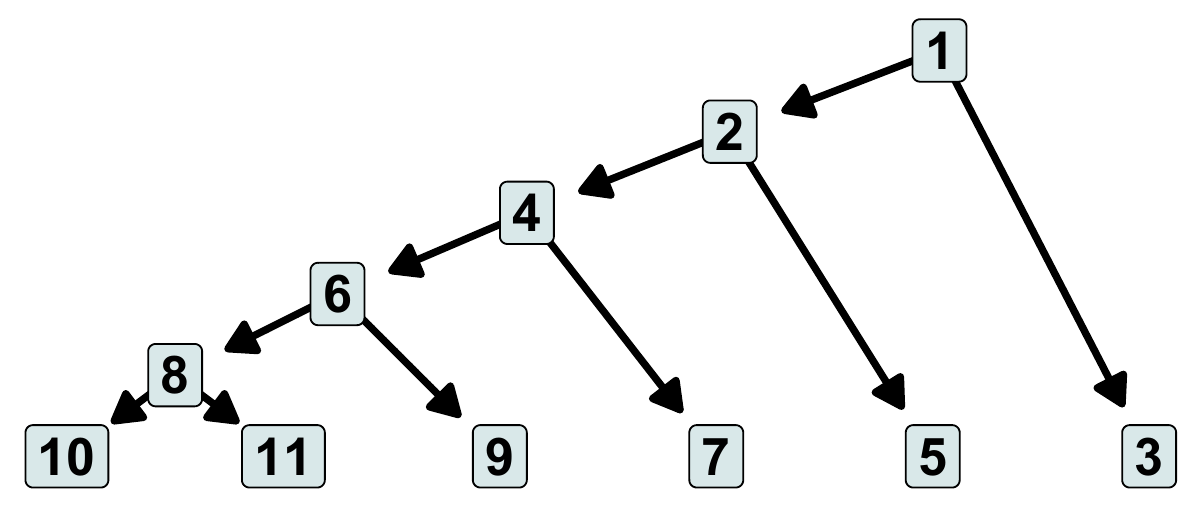}
    \end{subfigure}
    \caption{Illustration of the co-data groupings used the Verlaat data. Left: \texttt{CpG-islands}, five non-overlapping groups ordered in distance to the nearest CpG-island. Right: \texttt{p-values}, groups on the left correspond to lower p-values and are split recursively into two groups. The hierarchical structure on the groups is used in the extra level of shrinkage to find a discretisation that fits the data well as described in Section \ref{par:continuous}.}
    \label{fig:codataVerlaat}
\end{figure}

\begin{figure}
    \centering
    \begin{subfigure}[c]{0.32\textwidth}
    \centering
    \includegraphics[width=\linewidth]{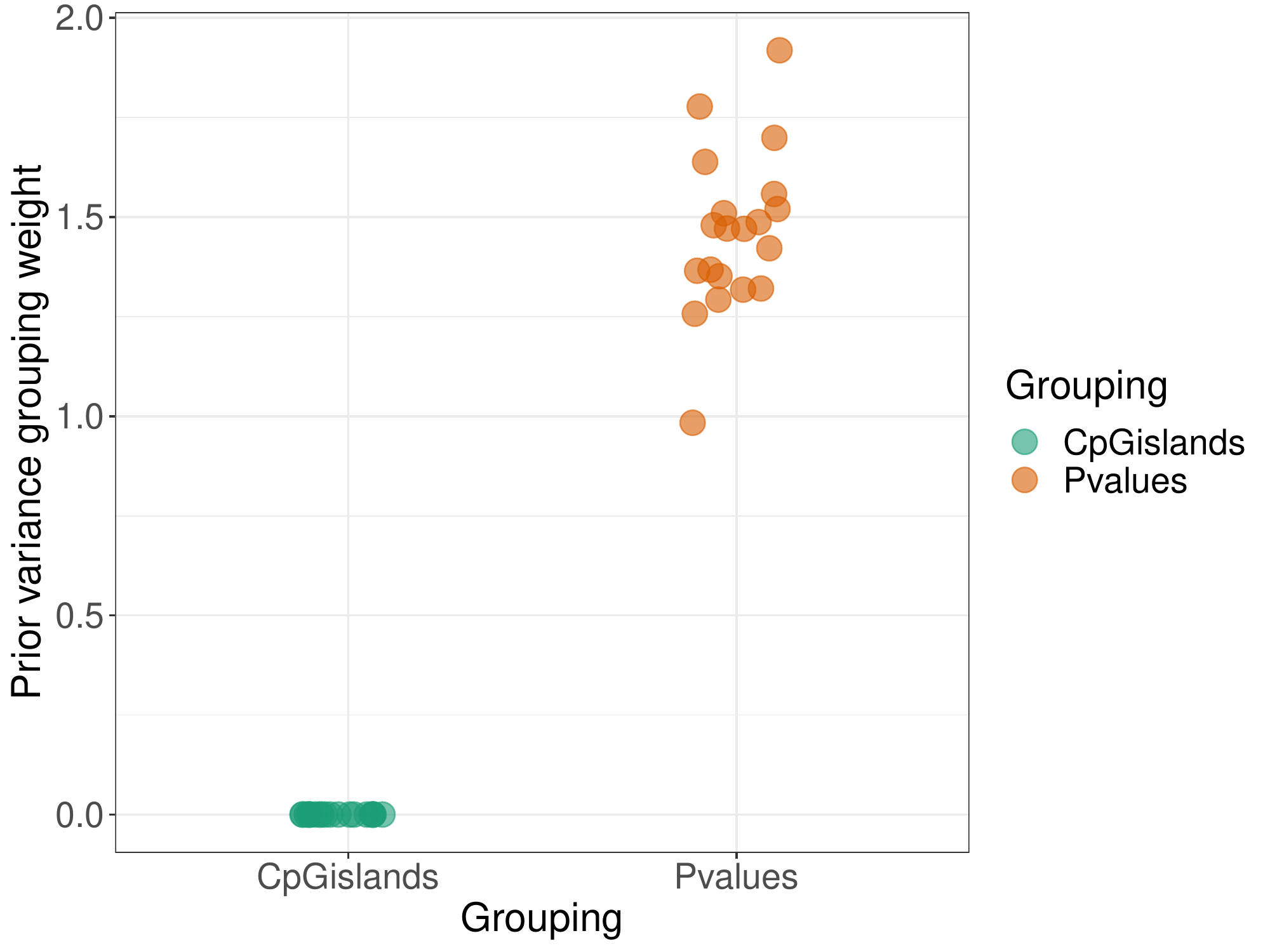}
    \end{subfigure}
    \begin{subfigure}[c]{0.32\textwidth}
    \centering
    \includegraphics[width=\linewidth]{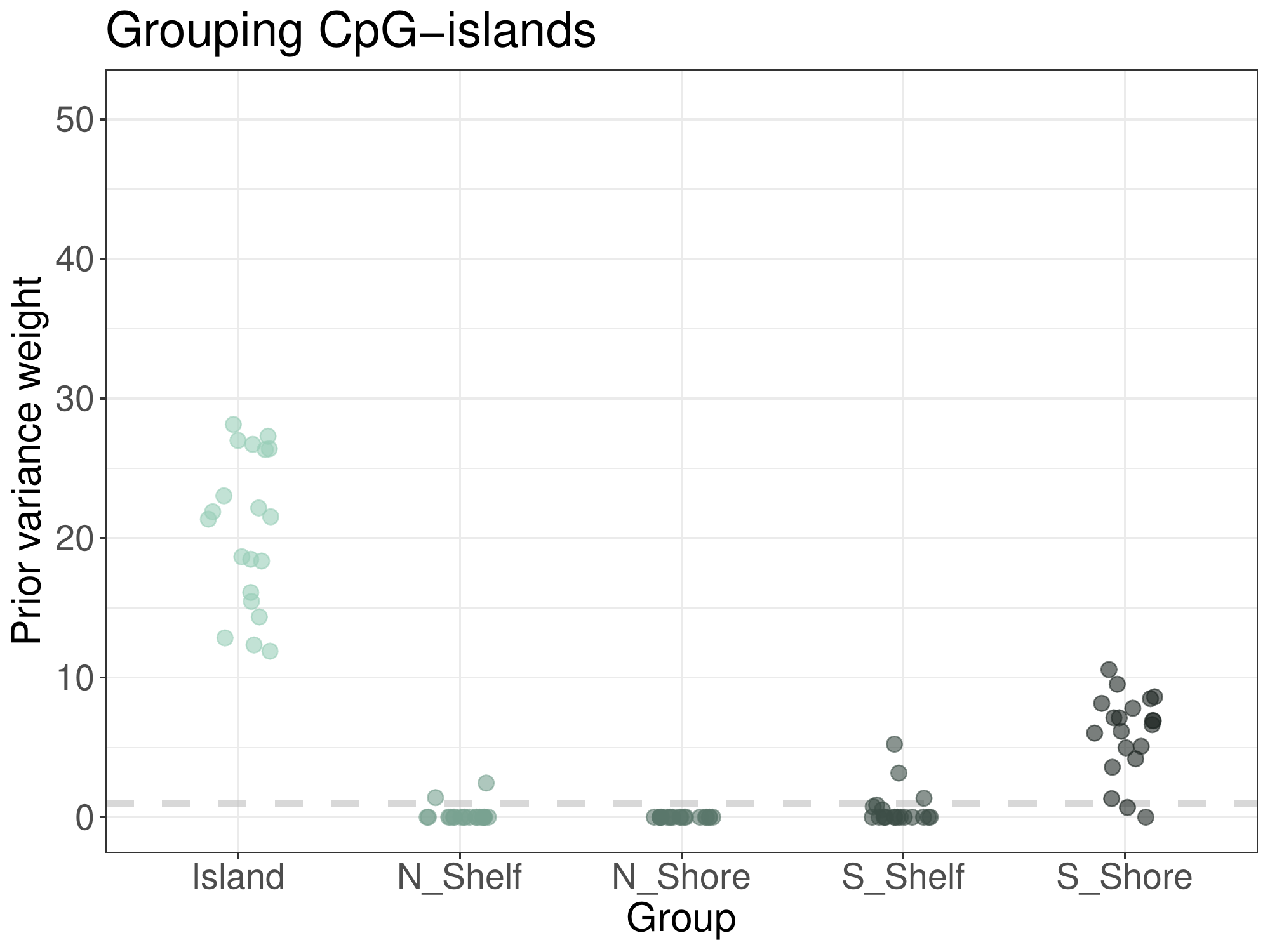}
    \end{subfigure}
    \begin{subfigure}[c]{0.32\textwidth}
    \centering
    \includegraphics[width=\linewidth]{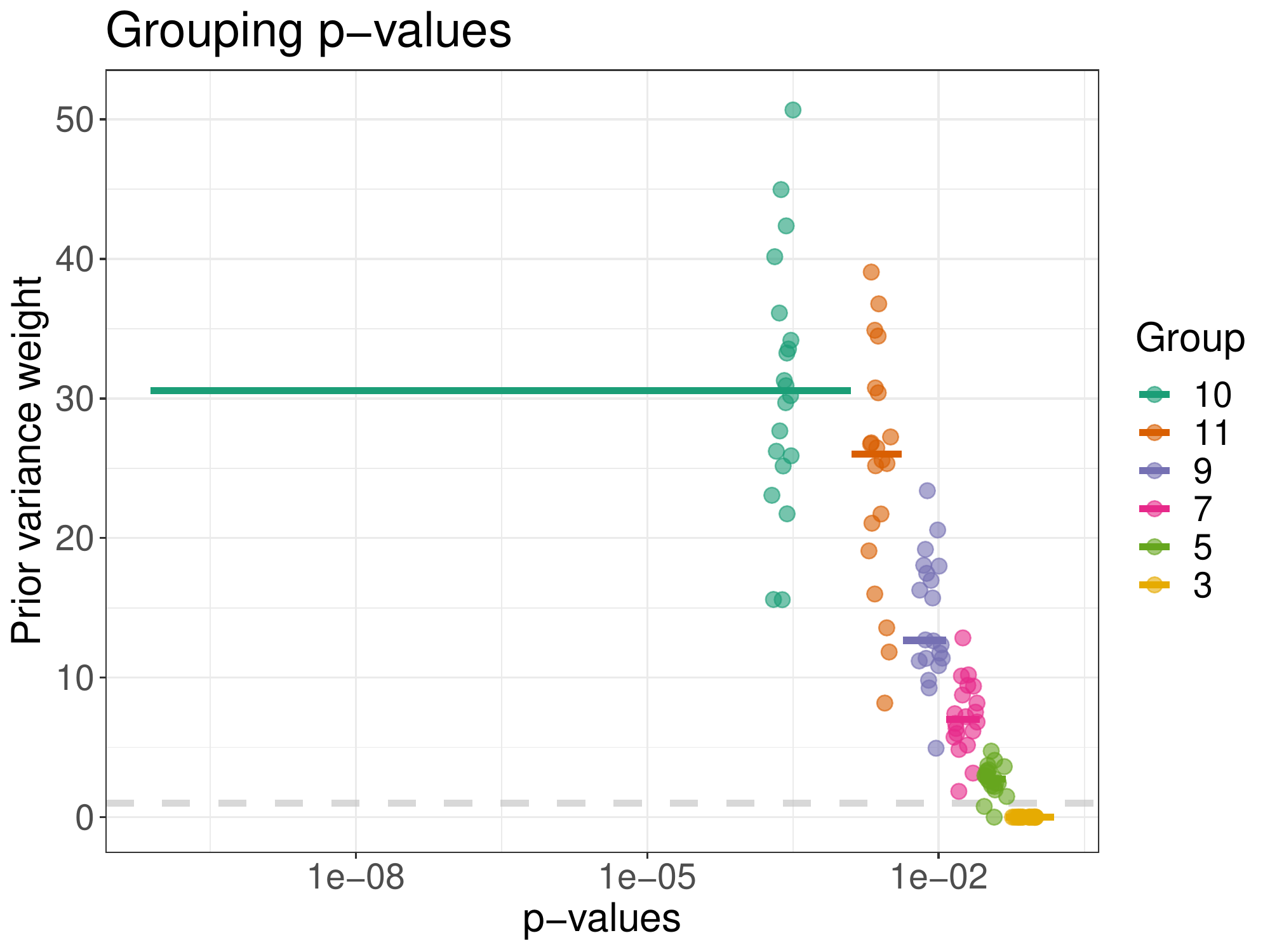}
    \end{subfigure}
    \caption{Results of 20-fold CV in Verlaat data example. Left: estimated co-data grouping weights. Middle: estimated group weights in \texttt{CpG-islands} grouping. Right: estimated local variance in \texttt{p-values} grouping; the median is shown from covariates in the leaf groups of the hierarchical tree illustrated in Figure \ref{fig:codataVerlaat} (horizontal line ranging from the minimum to maximum p-value in that group), and the corresponding estimates in the folds are shown (points, jittered along the median p-value in that group). A larger prior variance corresponds to a smaller penalty.}
    \label{fig:weightsVerlaat}
\end{figure}

\begin{figure}
    \centering
    \begin{subfigure}[c]{0.45\textwidth}
    \centering
    \includegraphics[width=\linewidth]{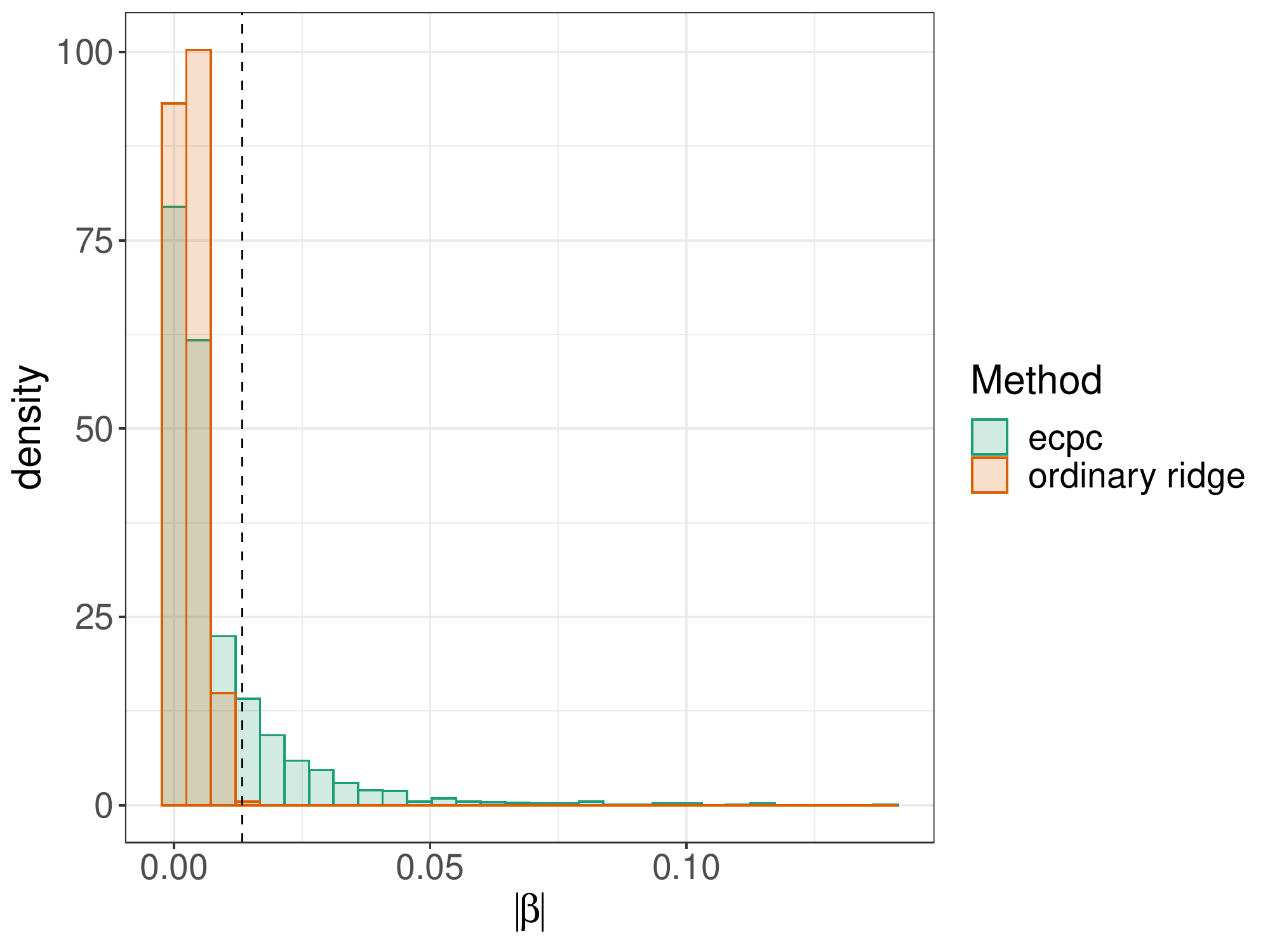}
    \end{subfigure}
    \begin{subfigure}[c]{0.45\textwidth}
    \centering
    \includegraphics[width=\linewidth]{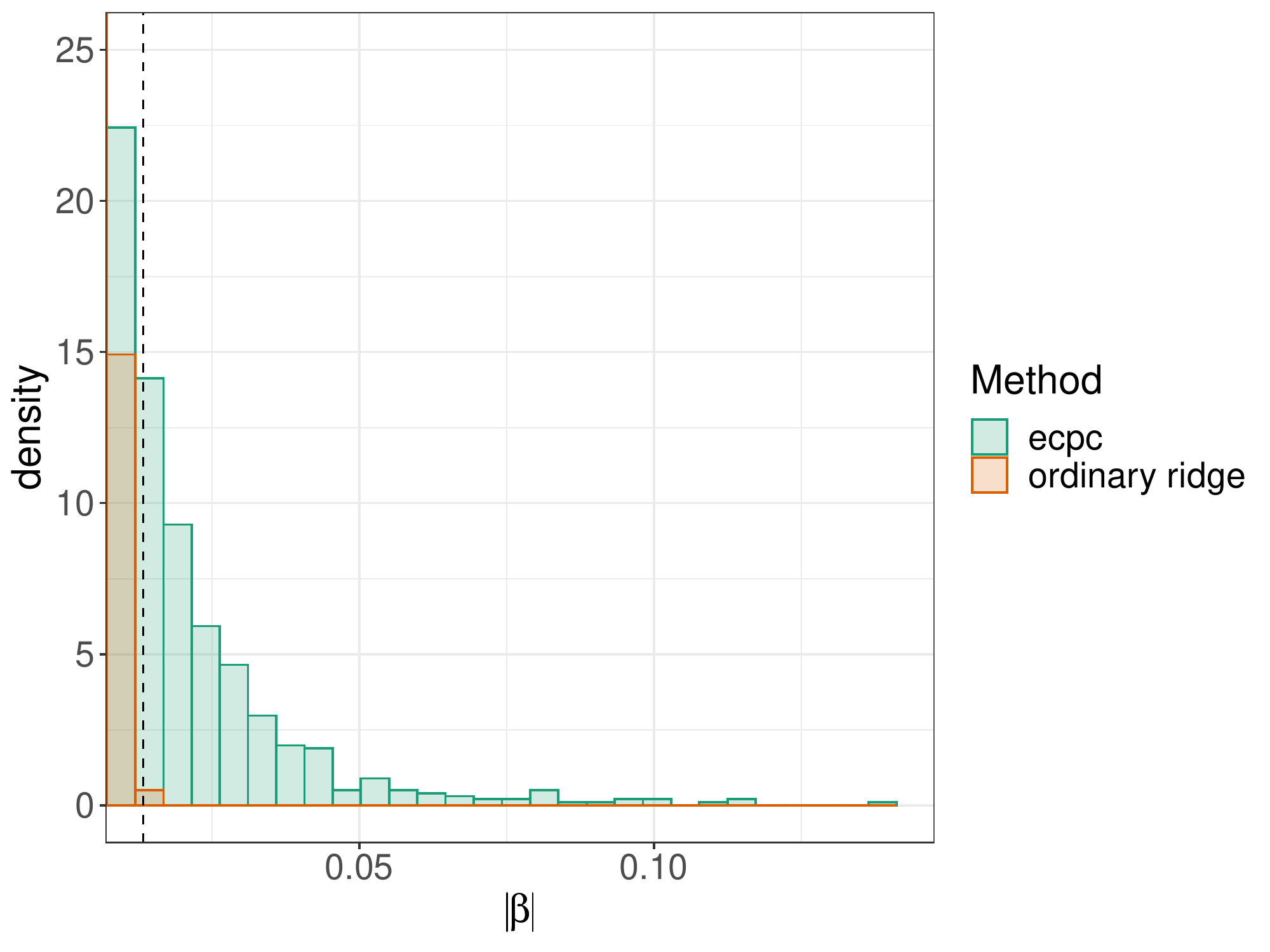}
    \end{subfigure}
    \caption{Verlaat data example. Left: histogram and density plot of absolute value of estimated regression coefficients using \texttt{ecpc} or \texttt{ordinary ridge}. Right: histogram of highest 0.1 quantile of the absolute value of the regression coefficients. \texttt{ecpc} results in more heavy-tailed distributed estimates compared to \texttt{ordinary ridge}.}
    \label{fig:Verlaatheavytails}
\end{figure}

\begin{figure}
    \centering
    \begin{subfigure}[c]{0.45\textwidth}
    \centering
    \includegraphics[width=\linewidth]{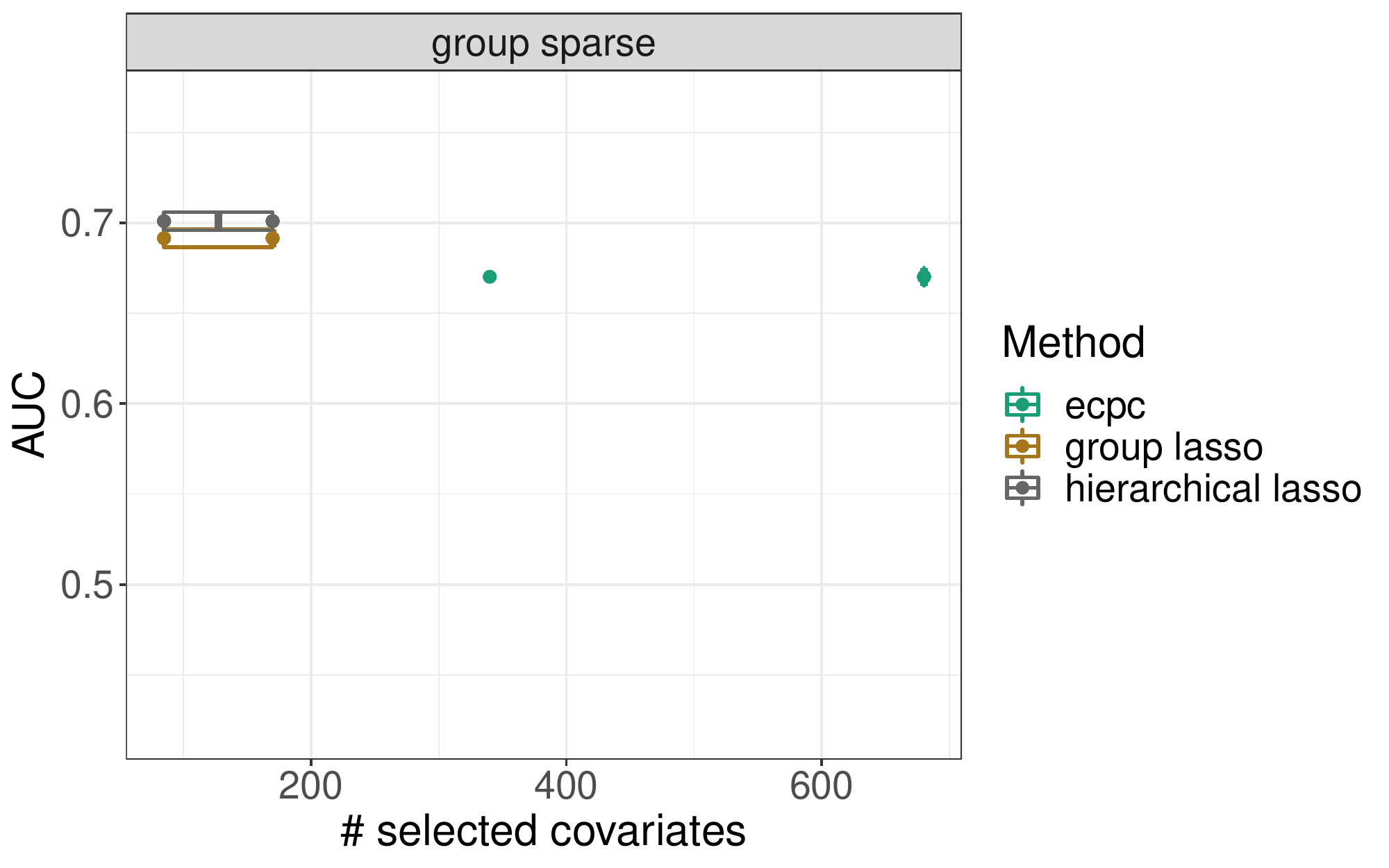}
    \end{subfigure}
    \caption{Results of 20-fold CV in Verlaat data example. AUC in various group sparse models}
    \label{fig:AUCVerlaatgroupsparse}
\end{figure}

\begin{figure}
    \centering
    \includegraphics[width=0.55\linewidth]{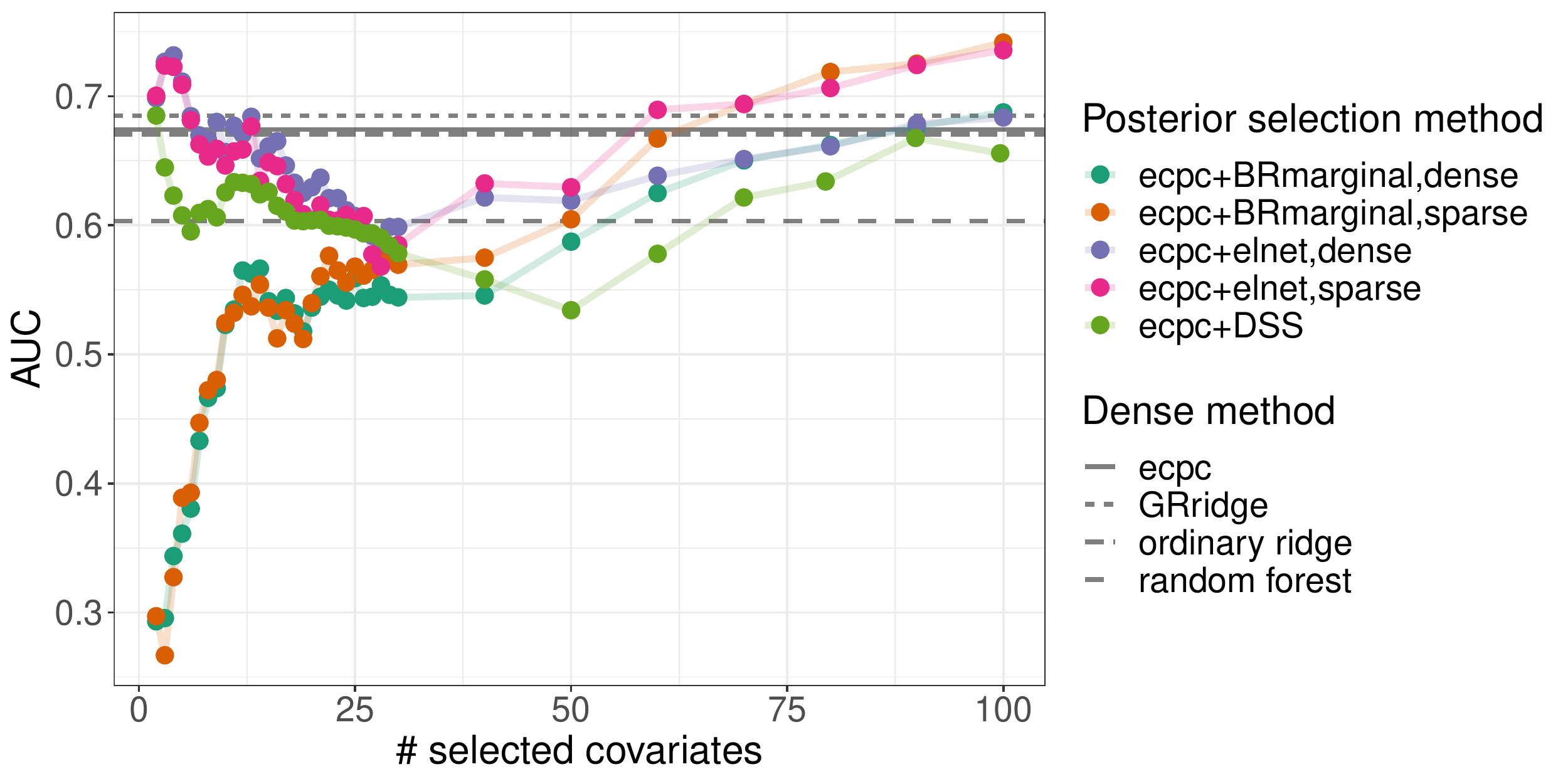}
    \caption{Results of 20-fold CV in Verlaat data example. AUC for sparse models using various post-hoc selection methods.}
    \label{fig:AUCposthocVerlaat}
\end{figure}

\begin{figure}
    \centering
    \begin{subfigure}[c]{0.45\textwidth}
    \centering
    \includegraphics[width=\linewidth]{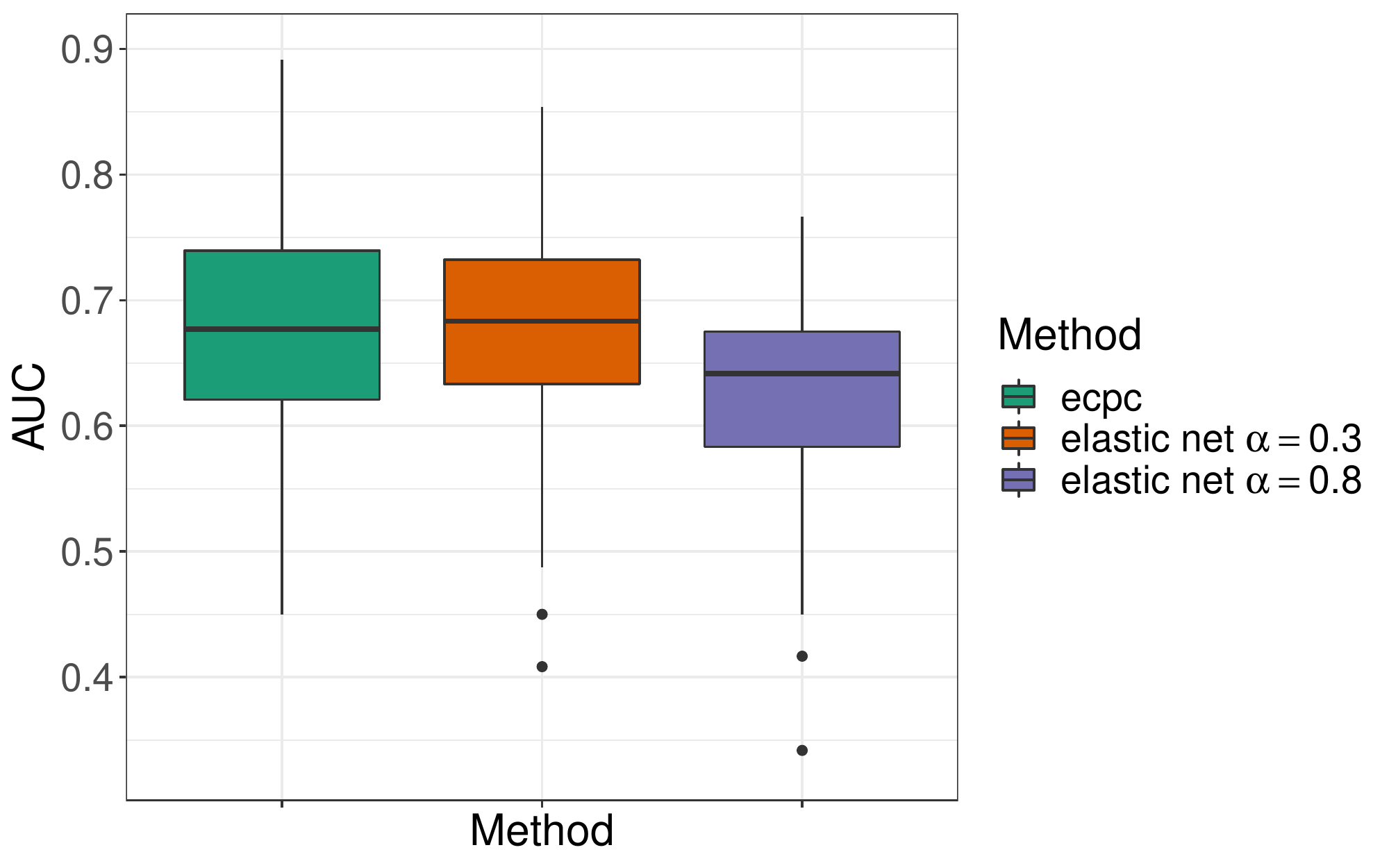}
    \end{subfigure}
    \begin{subfigure}[c]{0.45\textwidth}
    \centering
    \includegraphics[width=\linewidth]{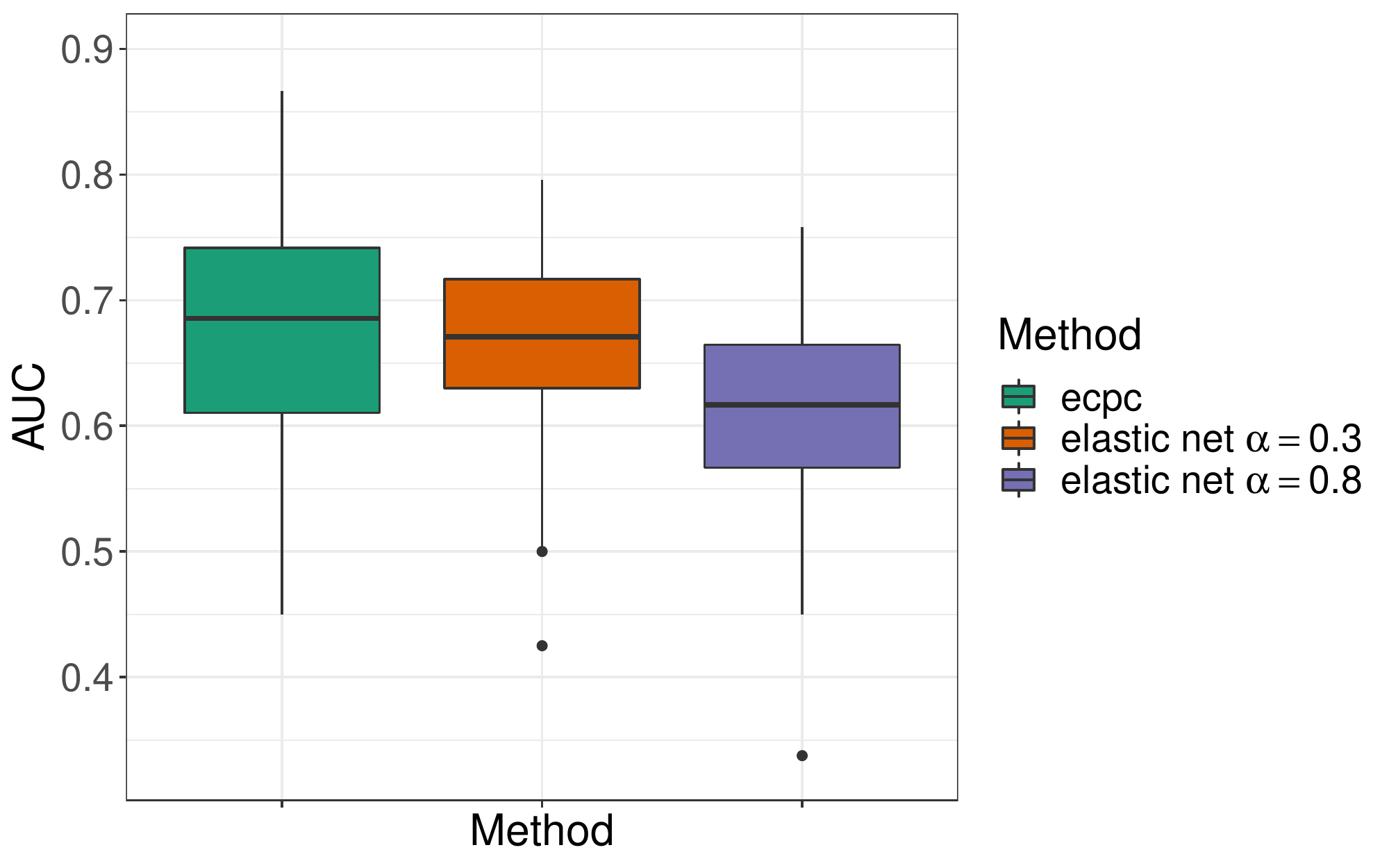}
    \end{subfigure}
    \caption{Results based on 50 stratified subsamples and corresponding test sets in Verlaat data example. Boxplot of the AUC performance of \texttt{ecpc}, \texttt{elastic net} with $\alpha=0.3$ and $\alpha=0.8$ on the test set based on selections of 25 covariates (left) or 50 covariates (right) in each subsample.}
    \label{fig:AUCsubsamplesVerlaat}
\end{figure}

\begin{figure}
    \centering
    \begin{subfigure}[c]{0.45\textwidth}
    \centering
    \includegraphics[width=\linewidth]{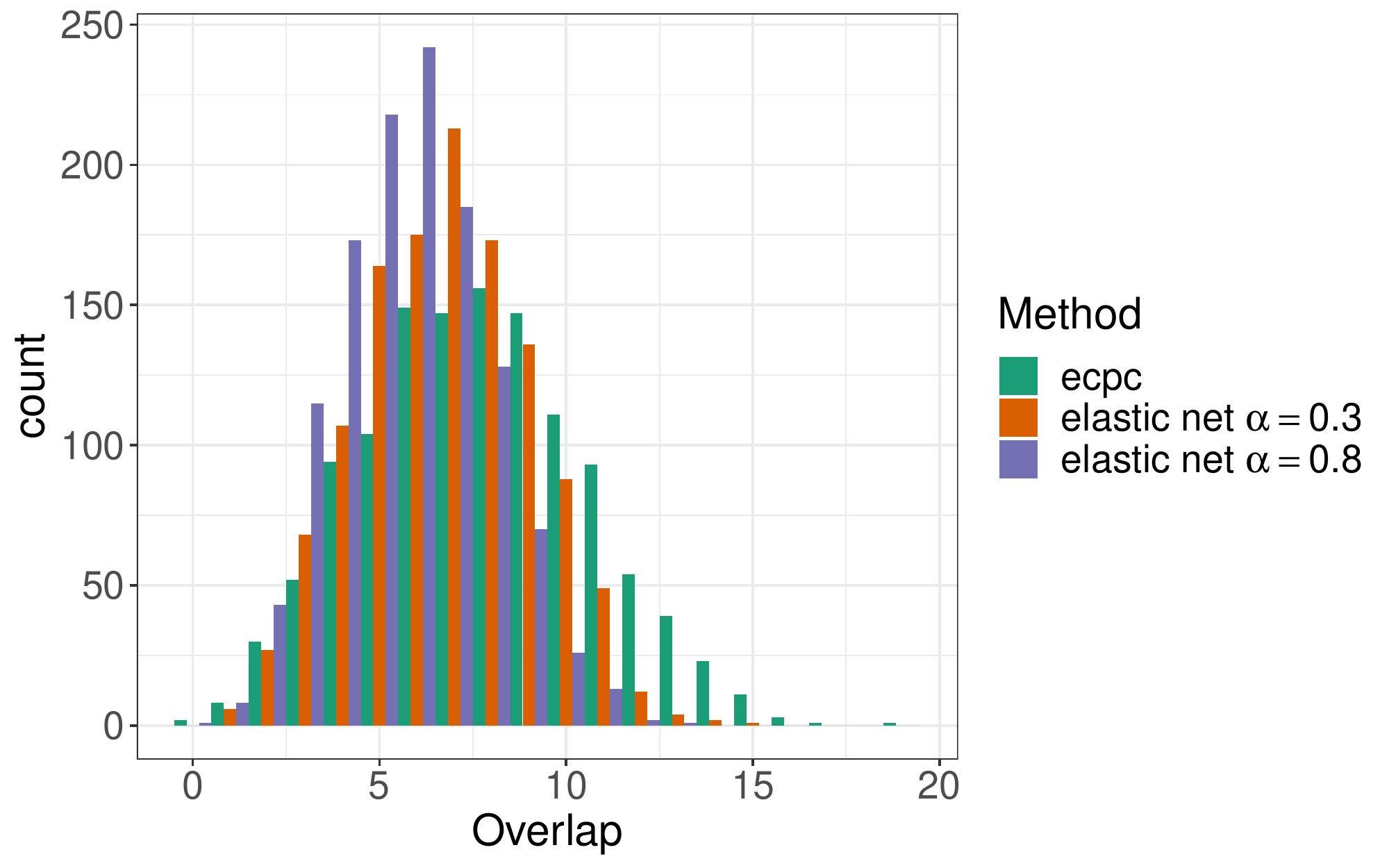}
    \end{subfigure}
    \begin{subfigure}[c]{0.45\textwidth}
    \centering
    \includegraphics[width=\linewidth]{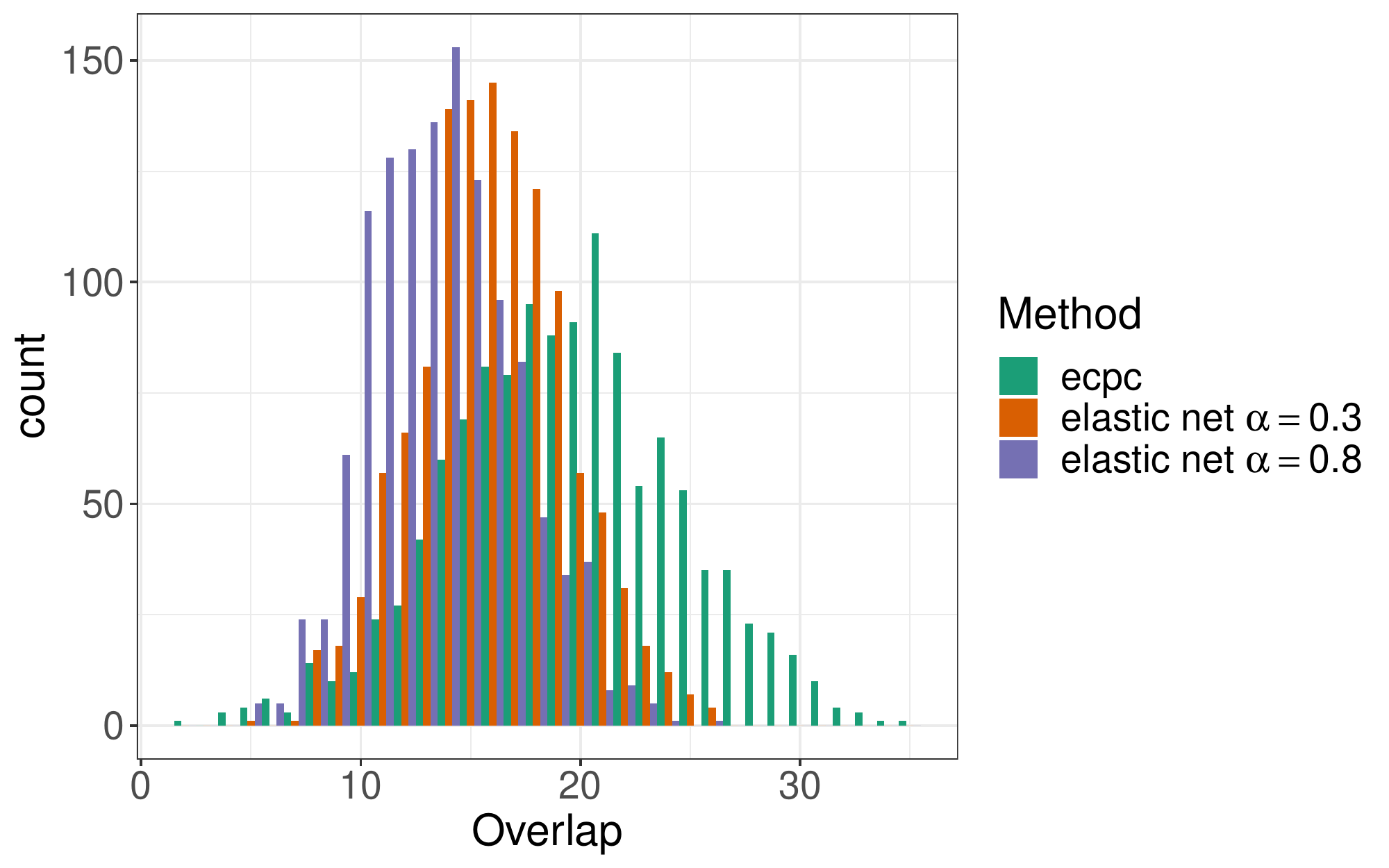}
    \end{subfigure}
    \caption{Results based on 50 stratified subsamples in Verlaat data example. Histogram of number of overlapping variables in pairwise comparisons of selections of 25 covariates (left) or 50 covariates (right) in each subsample, for the methods \texttt{ecpc}, \texttt{elastic net} with $\alpha=0.3$ and $\alpha=0.8$.}
    \label{fig:overlapVerlaat}
\end{figure}

\end{appendix}

\end{document}